\documentclass[a4paper,12pt]{article}
\usepackage{cite}

\usepackage{amsmath,amssymb,amsthm,mathtools}
\makeatletter
\renewenvironment{proof}[1][\proofname]{\par
\pushQED{\qed}%
\normalfont \topsep6\p@\@plus6\p@\relax
\trivlist
\item\relax
{\bfseries  
#1\@addpunct{.}}\hspace\labelsep\ignorespaces 
}{%
\popQED\endtrivlist\@endpefalse
}
\makeatother

\usepackage{array}
\usepackage{authblk}
\usepackage{caption}
\usepackage{fancyhdr}
\usepackage{finite_corr}
\usepackage{framed}
\usepackage[T1]{fontenc}
\usepackage{mathrsfs,stmaryrd}
\usepackage{physics}
\usepackage[section]{placeins}
\usepackage{tocloft}
\usepackage{xcolor}
\usepackage{adjustbox,rotating}

\usepackage{pgf}
\definecolor{xgray}{rgb}{0.75, 0.75, 0.75}

\usepackage{xparse}
\usepackage{tikz}
\usetikzlibrary{arrows,arrows.meta,backgrounds,calc,cd,decorations.markings,decorations.pathmorphing,decorations.pathreplacing,decorations.shapes,external,fit,graphs,math,matrix,positioning,shapes.geometric}

\usepackage{hyperref}

\numberwithin{equation}{section}

\def\swmap{\mathscr{O}}
\def\orbitP#1{P_{#1}}
\def\poly{\mathcal{F}}
\def\polt{f}
\def\cnN#1{\mathcal{N}^{#1}}
\def\hN#1{\hat{\mathcal{N}}^{#1}}
\def\su2{su(2)}
\def\gss{\xi}
\def\projD{\mathcal L}
\def\dos{\mathrm{dos}}
\def\intr{\mathrm{int}}
\def\coeffN{\mathrm{Coeff}}
\def\centz{\omega}

\DeclareMathOperator{\Span}{Span}
\DeclareMathOperator{\im}{im}

\DeclareMathOperator{\diag}{diag}

\theoremstyle{definition}

\newtheorem{definition}{Definition}[section]
\newtheorem{proposition}[definition]{Proposition}
\newtheorem{corollary}[definition]{Corollary}
\newtheorem{lemma}[definition]{Lemma}
\newtheorem{remark}[definition]{Remark}
\newtheorem{example}[definition]{Example}

\makeatletter
\let\c@mainprop\c@definition
\makeatother

\newlength{\barcut}
\newlength{\barshrink}
\newenvironment{barenv}[2][]{%
  \par\vspace{-\barshrink}%
  \def\barenvname{#2}%
  \def\barname{#1}%
  \MakeFramed{\advance\hsize-\width \FrameRestore}%
  \ifx\barname\empty \begin{#2}\else \begin{#2}[#1]\fi}%
 {\expandafter\end\expandafter{\barenvname}\endMakeFramed%
  \par\vspace{-\barshrink}}
\newenvironment{rebarenv}[3][]{%
  \par\vspace{-\barshrink}%
  \def\barenvname{#2}%
  \def\barname{#1}%
  \MakeFramed{\advance\hsize-\width \FrameRestore}%
  \expandafter\def\csname the#2\endcsname{\ref{#3}}%
  \def\barsharedname{mainprop}%
  \ifx\barenvname\barsharedname
    \def\thedefinition{\ref{#3}}%
  \fi
  \ifx\barname\empty \begin{#2}\else \begin{#2}[#1]\fi}%
 {\expandafter\end\expandafter{\barenvname}%
  \addtocounter{\barenvname}{-1}\endMakeFramed%
  \par\vspace{-\barshrink}}
\newcommand{\newbar}[2]{%
  \newenvironment{#1}[1][]{\begin{barenv}[##1]{#2}}{\end{barenv}}%
  \newenvironment{re#1}[2][]{\begin{rebarenv}[##1]{#2}{##2}}{\end{rebarenv}}}

\newbar{maintheorem}{theorem}
\newbar{mainconjecture}{conjecture}
\newbar{mainproposition}{mainprop}

\title{ \hfill { \normalsize QMUL-PH-26-36}\\  [1em] \bf 
Polynomial and Eigenspace Factorisations \\
for the One-loop Dilatation Operator at Finite $N$
}

\author{{\bf \small  Adrian Padellaro}\thanks{adrian.padellaro@icts.res.in}}
\affil{International Centre for Theoretical Sciences-TIFR,\authorcr Shivakote, Hesaraghatta Hobli, Bengaluru North 560089, India}
\author{{\bf  \small Sanjaye Ramgoolam}\thanks{s.ramgoolam@qmul.ac.uk}}
\affil{Centre for Theoretical Physics, Department of Physics and Astronomy, Queen~Mary University of London,  London E1 4NS, United Kingdom}
\author{{\bf \small  Ryo Suzuki}\thanks{rsuzuki.mp@gmail.com}}
\affil{School of Physics and Shing-Tung Yau Centre of Southeast University,\authorcr No.2 Sipailou, Xuanwu district, Nanjing, Jiangsu, 210096, China}

\date{}

\begin{document}

\maketitle
\thispagestyle{plain}
\begin{abstract}

We study gauge-invariant two-matrix polynomials of degree $L$, corresponding to the states in the $su(2)$ sector of $\mathcal{N}=4$ super Yang-Mills theory with gauge group $U(N)$ or $SU(N)$ at finite $N$. For the $U(N)$ theory, the state space has the structure of an algebra $\mathcal{A}_N(L)$, which is a direct sum of permutation centraliser algebras, i.e. sub-algebras of the group algebra  $\mathbb{C}[S_L]$ of $S_L$. For $N \ge L$, $\mathcal{A}_N (L) = \mathcal{A} (L)$ is independent of $N$, and the one-loop dilatation operator acts on $\mathcal{A}(L)$ with $N$ as a parameter. For $N < L$, $\mathcal{A}_N (L)$ is a subspace of $\mathcal{A}(L)$ defined using a finite-$N$ cutoff on Young diagrams with $L$ boxes. We establish two  results on the eigenspaces of the one-loop dilatation operator. Our first main result is that the characteristic polynomial of the one-loop dilatation operator on $\mathcal{A}_N(L)$ divides the characteristic polynomial from $\mathcal{A}(L)$, continued to $N < L$. Our second result is that each eigenstate of the one-loop dilatation operator in the $SU(N)$ theory generates an infinite tower of eigenstates in the $U(N)$ theory, through the multiplication of single-letter traces. These towers define an eigenspace factorisation of the $U(N)$ theory labelled by the number of such traces. Further, we conjecture that the space of $SU(N)$-invariant polynomials in two traceless matrices at finite $N$ is isomorphic to the intersection of $\mathcal{A}_N(L)$ with the subspace of $\mathbb{C}[S_L]$ spanned by derangements. This intersection conjecture is supported by dimension counting. Finally, we propose a one-loop dilatation operator on the intersection, thus providing a direct permutation framework for the $SU(N)$ spectral problem at finite $N$. The computed spectra based on the proposal agree with the prediction from the eigenspace factorisation in the $U(N)$ theory.

\end{abstract}
\flushbottom
\pagestyle{plain}

\newpage

\tableofcontents

\clearpage
\section{Introduction}\label{sec:introduction}

Physical observables of four-dimensional $\cN=4$ super Yang-Mills theory (SYM)
with $U(N)$ or $SU(N)$ gauge group are equivalently calculable from superstring
theory on \AdSxS, according to the AdS/CFT correspondence
\cite{Maldacena:1997re,Gubser:1998bc,Witten:1998qj}. The two theories have
independent formulations, in different spacetime dimensions and in terms of
different degrees of freedom. By the operator-state correspondence
\cite{Polchinski_1998}, the quantum states of the CFT are labelled by
gauge-invariant local operators; AdS/CFT identifies these with the quantum
states arising from the gravitational theory of strings and branes in \AdSxS,
and their conformal dimensions $\Delta$ with the energies $E$ of those states.
For fixed radius of the \AdSxS, the string coupling is proportional to $1/N$,
so that AdS/CFT at finite $N$ requires understanding of the non-perturbative regime on the string side.

In this paper we study local operators constructed as gauge-invariant polynomials in two complex adjoint scalar fields $Z,W$, regarded as two complex matrices.
An operator containing $m$ copies of $Z$ and $n$ copies of $W$ has bidegree $(m,n)$ and classical conformal dimension $\Delta_0=L$, where $L=m+n$.
For weak non-zero coupling, the spectrum of conformal dimensions can be obtained perturbatively from the eigenvalues of the loop-corrected dilatation operator,
\begin{equation}
\Delta = \Delta_0 + \frac{N g_{\rm YM}^2}{8\pi^2} \, \Delta_1 + \dots \,, \qquad
\fD_N = \fD^{(0)}_N + \frac{N g_{\rm YM}^2}{8\pi^2} \, \fD^{(1)}_N + \cdots \,.
\end{equation}
Here $g_{\rm YM}$ is the Yang-Mills coupling, and $\Delta_1$ is the eigenvalue of the one-loop dilatation operator $\fD_N^{(1)}$.
The gauge-invariant polynomials in $Z$ and $W$ form the $\su2$ sector, which is closed under one-loop dilatation at general $N$, with $m$ and $n$ separately preserved \cite{Beisert:2003tq,Beisert:2003jj}.

The space of gauge invariant polynomial functions has a highly non-trivial dependence on $N$.
We write the space of $U(N)$-invariant polynomials in $Z,W$ of fixed degrees $(m,n)$ as $\mathcal{H}_N (m,n)$.
This space is finite-dimensional for each
$(m,n,N)$, so finding the one-loop corrections to conformal dimensions becomes a problem of computing roots of the characteristic polynomial of a finite-dimensional matrix. 
In the stable regime $N \geq L$, $\dim \mathcal{H}_N(m,n)$ counts multi-traces and is independent
of $N$; in the unstable regime $N < L$, distinct multi-traces are no longer linearly
independent due to trace relations and $\dim \mathcal{H}_N(m,n)$ depends non-trivially on $N$.
In terms of the dual string theory, holding $L$ fixed while varying $N$ moves between the perturbative and non-perturbative regimes of the string theory.
Perturbative string theory corresponds to the stable regime, namely fixed-$L$ large-$N$ limit, while the unstable regime probes the non-perturbative physics of branes and black holes.

Symmetric groups provide a powerful framework for understanding spaces of
invariant polynomials for varying $N$ and fixed degree $(m,n)$.
As we recall below, any invariant in $\cH_N(m,n)$ can be written as a
linear combination of multi-trace invariants constructed from permutations $g \in S_{m+n}$,
\begin{equation}\label{PermsAndInvariants} 
  Z_{i_1 i_{g(1)}}\cdots Z_{i_m i_{g(m)}}\,
  W_{i_{m+1} i_{g(m+1)}}\cdots W_{i_{m+n} i_{g(m+n)}} 
\end{equation}
with repeated indices summed.
A permutation specifies how to contract the colour indices of $Z$ and $W$, and each cycle gives a single-trace factor. 
That these invariants span the space for every $N$ is the first fundamental theorem of invariant theory for $GL(N)$, whose polynomial invariants coincide with those of $U(N)$, and goes back to Weyl \cite{Weyl1939classical}; see \cite{GoodmanWallach2009} for a modern treatment of the mixed-tensor form used here. 
Combined with Schur--Weyl duality, this description turns computations with invariants into computations in the group algebra $\mathbb{C}[S_{m+n}]$, where $N$ enters through $\mathrm{tr}_{V^{\otimes L}}(g)=N^{\mathrm{cycles}(g)}$ for $V=\mathbb{C}^N$, and through the cutoff $\ell(R)\leq N$ on the number of rows of Young diagrams $R$.

The application of symmetric group techniques to large-$N$ gauge theories has a fruitful history.
Permutations and Schur-Weyl duality were central to the reformulation of large-$N$ two-dimensional gauge theory as a string theory
\cite{GrossTaylor1993,CMR1994largeN,Cordes:1994sd,Ramgoolam1994wilson}.
In applications to $\mathcal{N}=4$ SYM \cite{Corley:2001zk}, the permutation parameterisation of holomorphic gauge
invariant functions of a complex matrix in the half-BPS sector allowed the
computation of orthogonal bases labelled by Young diagrams $R$, extremal
three-point functions and a general proposal for the CFT duals of giant
gravitons
\cite{McGreevy:2000cw,Balasubramanian:2001nh,Hashimoto:2000zp,Grisaru:2000zn} and Lin-Lunin-Maldacena (LLM) geometries \cite{Lin:2004nb}.
The norms of these orthogonal basis elements vanish for $\ell(R)>N$.
Finite-$N$ orthogonal bases and counting formulae were subsequently generalised
to multi-matrix systems
\cite{Kimura:2007wy,Brown:2007xh,Bhattacharyya:2008rb,Bhattacharyya:2008xy,Brown:2008ij}.
The restricted Schur and covariant bases also have a Young diagram label $R$ where finite-$N$ Hilbert spaces have $\ell(R) \leq N$ and the null states have $\ell(R) > N$.
Since then, this framework has been developed for quivers
\cite{Pasukonis:2013ts}, matrix systems with discrete permutation group $S_N$
symmetry
\cite{Barnes:2021tjp,Barnes:2022qli,OConnor:2023ptf,OConnor:2024udv,OConnor:2025fpj}
and tensor models \cite{BenGeloun:2017vwn,BenGeloun:2020yau,Barnes:2023rvv}.

The multi-matrix approaches exploit the fact that the linear extension of \eqref{PermsAndInvariants} for $\sigma \in \bb{C}[S_{m+n}]$ is invariant under conjugation by $\gamma \in S_m \times S_n$,
\begin{equation}\label{subconj}
\sigma \rightarrow \gamma \sigma \gamma^{-1}.
\end{equation}
The sums over the orbits in $S_{m+n}$ generated by \eqref{subconj} span the subalgebra of $\mathbb{C}[S_{m+n}]$ commuting with $S_m \times S_n$, called the \emph{permutation centraliser algebra} (PCA) $\cA(m,n)$ \cite{Mattioli:2016eyp}.
This associative algebra is semi-simple, allowing the structure theory of such algebras,
notably the Wedderburn-Artin decomposition, to be applied.
In the stable regime $N \ge L$, the dimension of $\cA(m,n)$ is equal to $\dim \cH_N(m,n)$, and the Wedderburn-Artin decomposition of
$\cA(m,n)$ as a direct sum of matrix algebras accounts for the counting formula
of the dimension as a sum of squares of Littlewood-Richardson coefficients. In
the unstable regime, the cut-off $\ell(R) \le N$ selects a subset of these
matrix blocks, defining a semisimple subalgebra $\cA_N(m,n)$ of $\cA(m,n)$
which accounts for the finite $N$ counting formula.

The PCA framework allows the explicit construction of a complete basis of
operators, compatible with finite-$N$ trace relations. The centre of
$\mathbb{C}[S_{m+n}]$, consisting of the elements commuting with all of $S_{m+n}$, is
contained in $\cA(m,n)$, and its elements are directly related to the Casimir
operators introduced in \cite{Kimura:2008ac}. Motivated by this observation,
several important properties of the centre were investigated in
\cite{Kemp:2019log}. Recently, these findings culminated in exact methods, based on integer eigensystems inside permutation algebras, for computing representation-theoretic finite-$N$ bases for multi-matrix invariants with $U(N)$ gauge group \cite{Padellaro:2024rld}.
Following these results, we will refer to a basis of $\cA(m,n)$ that
diagonalises the left action of central elements of $\mathbb{C}[S_{m+n}]$, and
identifies a Young diagram label $R$ which controls finite $N$ cutoffs, as a
\emph{stringy-exclusion-principle} (SEP) basis. The name is inspired by the cutoff on
the charges of bulk states found in \cite{Maldacena:1998bw}, of which the
condition $\ell(R) \le N$ is the counterpart here.

The two gauge groups, $U(N)$ and $SU(N)$, are equally natural choices in $\cN=4$ SYM, but they are significantly different at finite $N$.
The permutation framework has been used to study $\cH'_N(m,n)$, the space of $SU(N)$-invariant polynomials of traceless matrices $Z', W'$, extending earlier studies of half-BPS operators and correlators to the $SU(N)$ theory \cite{Corley:2002mj,deMelloKoch:2004crq,Brown:2007bb,Aprile:2020single}. 
The traceless condition removes permutations with fixed points. Thus, $\cH'_N(m,n)$ in the stable regime is represented by derangements, namely permutations without fixed points; for a single complex matrix, this description of the independent $SU(N)$ invariants was given in \cite{Corley:2002mj}. Derangements also arise in the study of tensor powers of adjoint $\mathfrak{sl}_N$ modules \cite{BenkartDoty2002,Doty:2007}.
Derangements do not form a group, so their use does not directly provide the representation-theoretic machinery available for $U(N)$. Moreover, projecting to derangements and projecting to the finite-$N$ subspace do not in general commute.
Our aim is to describe both $SU(N)$ gauge-invariant operators and their one-loop dilatation entirely within the PCA framework at any $N$, with the physical action related to the permutation action by the permutation construction of gauge invariant operators in \eqref{PermsAndInvariants}. Section \ref{sec:derangements} addresses this problem through an intersection conjecture and a separate proposal for the one-loop action.

We now state the three main results of this paper. The first is a finite-$N$ factorisation for the characteristic polynomial of the one-loop dilatation operator for $U(N)$. In the stable regime, the one-loop mixing matrix can be written as an operator $ H^{(1)}_{\gss}$ on $\cA(m,n)$ whose entries depend on the formal parameter $\gss$. Evaluating this operator at $\gss=1/N$ is meaningful even when $N<L$, although the matrix then acts on a space larger than the physical finite-$N$ space.
We prove that the orthogonal complement of $\cA_N(m,n)$ (the space of null states) is invariant under the one-loop dilatation.
It follows that the mixing matrix is block triangular by ordering a basis into physical and null sectors.
The characteristic polynomials exhibit the factorisation condition
\begin{equation}
\det_{\cH_N(m,n)}(x - \fD_N^{(1)}) = \frac{\det_{\cA(m,n)} (x -  H^{(1)}_{\gss} )\Big\vert_{\gss = \frac{1}{N}} }{Q^{(N)}_{m,n} (x)}
\label{eq:intro-finite-N-factorisation}
\end{equation}
where the left-hand side is the characteristic polynomial on the physical finite-$N$ space, and $Q^{(N)}_{m,n}$ comes from the null sector. Thus, the characteristic polynomial computed before imposing trace relations already contains the exact finite-$N$ spectrum as a factor.

The second result is a single-letter-trace factorisation relating the $U(N)$ and $SU(N)$ theories. 
By splitting $U(N)$ matrices into traceless and trace parts, we obtain an identity for the Hilbert series of $\cH_N$ and $\cH'_N$. 
This factorisation property extends to the eigenspace of the one-loop dilatation operator at finite $N$.
As a consequence, the $SU(N)$ characteristic polynomial is determined by a ratio of $U(N)$ characteristic polynomials at neighbouring bidegrees. The eigenspace factorisation provides an independent route to compute the $SU(N)$ characteristic polynomial at finite $N$, and a stringent check on the direct derangement construction.

Our third result is an intrinsic PCA description of $SU(N)$ states in terms of derangements, together with a proposal for their one-loop dilatation operator in the stable regime.
In the unstable regime, we propose that $SU(N)$-invariant operators at finite $N$ are described by the intersection
\begin{equation}
\cD_N(m,n):=\cA_N(m,n)\cap\mathbb{C}[D_{m+n}]  \ \cong \ \mathcal{H}'_N(m,n)
\label{eq:intro-SU-intersection}
\end{equation}
where $D_{m+n} \subset S_{m+n}$ is the set of derangements and $\cA_N(m,n)$ is the finite-$N$ subspace of PCA spanned by any SEP basis with $\ell(R) \leq N$.
We give computational evidence for this isomorphism by comparing the dimension of the intersection and the rank of its $SU(N)$ two-point-function Gram matrix with $\dim\mathcal{H}'_N(m,n)$. Separately, we propose a projected one-loop operator on the intersection and test its characteristic polynomial against the eigenspace factorisation formula. Together, these constructions give a candidate PCA description of $SU(N)$ operators and their one-loop dynamics in the unstable regime.
 
The intersection \eqref{eq:intro-SU-intersection} is also connected to the quarter-BPS states in $\cN=4$ SYM.
Early constructions of quarter-BPS operators were given in \cite{Ryzhov:2001bp}, and the number of quarter-BPS states at finite $N$ was found in \cite{Kinney:2005ej,Dolan:2007rq}. General weak-coupling quarter-BPS operators were constructed by using the intersection of a symmetrisation operation with the finite-$N$ constraint \cite{Lewis-Brown:2020nmg}. In the $\su2$ sector, the quarter-BPS states are characterised by the zero modes of the one-loop dilatation operator. 
Our proposal \eqref{eq:intro-SU-intersection} is a natural extension of the same strategy, where the symmetrisation operation is replaced by the derangement projection.

This paper is organised as follows. 
Section \ref{sec:UN factorisation} develops the PCA framework of the $U(N)$ invariants and the one-loop mixing problem, including a proof of the finite-$N$ factorisation \eqref{eq:intro-finite-N-factorisation}. 
Section \ref{sec:SU(N)} establishes the single-letter-trace factorisation by studying the space of $SU(N)$-invariant operators and the eigensystem of the one-loop dilatation operator. 
Section \ref{sec:derangements} develops the derangement description of traceless invariants, formulates \eqref{eq:intro-SU-intersection}, and tests the induced one-loop operator. 
The appendices collect conventions and computational data.
Our symmetric-group conventions follow \cite{Padellaro:2024rld}.

\section{One-loop finite-$N$ polynomial factorisation for $U(N)$}\label{sec:UN factorisation}

We consider the two-matrix sector of four-dimensional $\mathcal{N}=4$ super Yang-Mills theory with gauge group $U(N)$.
Let $\mathcal{H}_N(m,n)$ denote the space of $U(N)$ invariant two-matrix polynomials containing $m$ copies of $Z$ and $n$ copies of $W$, with $m,n \in \mathbb{Z}_{\geq 0}$.
We use the notation $m+n$ and $L \equiv m+n$ interchangeably.
The dimension of $\mathcal{H}_N(m,n)$ is independent of $N$ in the stable regime $N\geq L$. By contrast, the dimension of $\mathcal{H}_N(m,n)$ decreases in the unstable regime $N< L$, due to trace relations between invariant polynomials.
We now give an extended overview of the results in this section.

The first subsection describes the $U(N)$ invariants in terms of permutations. The permutation centraliser algebra $\cA (m,n)$ provides a unified description of the state space for all positive integer $N \ge L$.
The relation between permutations and $U(N)$ invariants is encoded by the permutations-to-invariants map (see Definition \ref{def:trace map})
\begin{equation}
    \swmap^{(N)}_{m,n}: \mathcal{A}(m,n) \longrightarrow \mathcal{H}_N(m,n).
    \label{eq:permutation-to-invariant-map-summary}
\end{equation}
In the stable regime this map is an isomorphism,
\begin{equation}
    \mathcal{H}_N(m,n) \cong \mathcal{A}(m,n), \qquad N\geq L.
\end{equation}
For $N<L$, the map has a non-trivial nullspace (a consequence of trace relations)
\begin{equation}
    \ker \swmap^{(N)}_{m,n} = \{\sigma \in \cA(m,n) \, \vert \, \swmap^{(N)}_{m,n}(\sigma) = 0\} \neq \{0\}, \qquad N < L
\end{equation}
The physical state space is represented by a subalgebra $\mathcal{A}_N(m,n)\subset\mathcal{A}(m,n)$, which we will define.
The decomposition of $\cA(m,n)$ into $\cA_N(m,n)$ and the nullspace 
\begin{equation}
    \cA(m,n) \cong \cA_N(m,n) \oplus \ker \swmap^{(N)}_{m,n}
\end{equation}
can be described using representation-theoretic bases, including the restricted-Schur \cite{Bhattacharyya:2008rb,Bhattacharyya:2008xy} and covariant bases \cite{Brown:2007xh,Brown:2008ij}. They are collectively called stringy-exclusion-principle (SEP) bases \cite{Lewis-Brown:2020nmg}. An eigenvalue construction of integer SEP bases was recently worked out in \cite{Padellaro:2024rld}, building on the observations in \cite{Kimura:2008ac, Kemp:2019log}.
Importantly, the restriction of the permutations-to-invariants map to $\cA_N(m,n)$ 
\begin{equation}
    \hat \swmap^{(N)}_{m,n}: \cA_N(m,n) \rightarrow \cH_N(m,n),
\end{equation}
gives an isomorphism
\begin{equation}
    \cA_N(m,n) \overset{\hat \swmap^{(N)}_{m,n}}{\cong} \mathcal{H}_N(m,n).
\end{equation}

We then consider the one-loop dilatation operator in the $\su2$ sector \cite{Beisert:2003tq},
\begin{equation}
\fD^{(1)}_N : \mathcal{H}_N(m,n) \rightarrow \mathcal{H}_N(m,n), \qquad
\fD^{(1)}_N = -\frac{1}{N}:\tr(\comm{Z}{W}\comm{\pdv{}{Z}}{\pdv{}{W}}): 
\end{equation}
where the colons denote normal ordering. 
Computing its action on operators labelled by an orbit basis $\sigma,\tau$ of $\cA(m,n)$ gives
\begin{equation}
\fD_N^{(1)}\,\swmap^{(N)}_{m,n}(\sigma) = \sum_{\tau}\Bigl\{(H^{(1)}_{\mathrm{p}})_{\tau\sigma}+\frac{1}{N}(H^{(1)}_{\mathrm{np}})_{\tau\sigma}\Bigr\}\,\swmap^{(N)}_{m,n}(\tau)
\end{equation}
where $H^{(1)}_{\mathrm{p}}$ and $H^{(1)}_{\mathrm{np}}$ are the planar and non-planar matrix elements, respectively. 
Promoting $1/N$ to a formal variable $\gss$, we define a new operator acting on $\cA(m,n)$ by
\begin{equation}
 H^{(1)}_{\gss} : \cA(m,n) \rightarrow \cA(m,n), \qquad
 H^{(1)}_{\gss}  (\sigma) = \sum_{\tau}
\Big\{  (  H^{(1)}_{\mathrm{p}} )_{\tau \sigma} + \gss \, (  H^{(1)}_{\mathrm{np}} )_{\tau \sigma} \Big\} \, \tau 
\end{equation}
and call it the one-loop permutation dilatation operator. Their explicit forms are given in Definition \ref{def: 1LPDO} \cite{Bellucci:2004ru,Pasukonis:2010rv}.
Proposition~\ref{prop:D1 intertwining} states that this action is compatible with the permutations-to-invariants map for every positive integer $N$, 
\begin{equation}
\fD_N^{(1)}\circ\swmap^{(N)}_{m,n} = \left.\swmap^{(N)}_{m,n}\circ H^{(1)}_{\gss}\right|_{\gss=1/N} \,, \qquad \forall N \in \mathbb{N}
\label{eq:intertwine}
\end{equation}
which we call the \emph{intertwining relation}.
This relation ensures that finite-$N$ trace relations are preserved on both sides, giving a triangularity property of the one-loop permutation dilatation operator in SEP bases.
Note that the one-loop permutation dilatation operator takes a block form in a SEP basis
\begin{equation}
     H_{\gss}^{(1)} \ \rightarrow \  \begin{pmatrix}  H_{11}^{(1)} &  H_{12}^{(1)} \\  H_{21}^{(1)} &  H_{22}^{(1)} \end{pmatrix}
    \quad \text{where} \quad \begin{aligned}
        & H_{11}^{(1)}: \cA_N(m,n) \rightarrow \cA_N(m,n), 
        && H_{12}^{(1)}: \ker \swmap^{(N)}_{m,n} \rightarrow \cA_N(m,n), \\
        & H_{21}^{(1)}: \cA_N(m,n) \rightarrow \ker \swmap^{(N)}_{m,n},
        && H_{22}^{(1)}: \ker \swmap^{(N)}_{m,n} \rightarrow \ker \swmap^{(N)}_{m,n}.
    \end{aligned}
\end{equation}
Using \eqref{eq:intertwine}, we show that the one-loop permutation dilatation operator is block triangular at $\gss = 1/N$
\begin{equation}
     H_{\gss}^{(1)} \Big\vert_{\gss=1/N} \ \rightarrow \  \begin{pmatrix}  H_{11}^{(1)}& 0 \\  H_{21}^{(1)} &  H_{22}^{(1)} \end{pmatrix}\,
\end{equation}
and the matrix elements of $ H_{11}^{(1)}$ can be chosen to be equal to the matrix elements of $ \fD_N^{(1)}$.
From this, the main theorem of this section follows.

\begin{maintheorem}[Polynomial factorisation]\label{prop:finite-N factorisation}
Consider the characteristic polynomials
    \begin{equation}
        \det_{\mathcal{H}_N(m,n)}(x -  \fD_N^{(1)}), \quad \det_{\cA(m,n)}(x -  H^{(1)}_{\gss}).
    \end{equation}
    For every positive integer $N$, the first polynomial divides the second after setting $\gss=N^{-1}$:
    \begin{equation}
        \det_{\cA(m,n)}(x -  H^{(1)}_{\gss})\Big\vert_{\gss=\frac{1}{N}} = Q^{(N)}_{m,n}(x)\det_{\cH_N(m,n)}(x - \fD_N^{(1)}) \, ,
    \end{equation}
    where $Q^{(N)}_{m,n}(x)$ is a polynomial factor. 
\end{maintheorem}

Note that the space $\mathcal{H}_N(m,n)$ involves finite-$N$ effects, in particular trace relations, for $N < L$.
This theorem shows that the continued characteristic polynomial in the stable regime, which a priori knows nothing about trace relations, contains the finite-$N$ spectrum in a precise sense.

\subsection{Permutation framework for $\mathcal{N}=4$ SYM with $U(N)$ gauge group}\label{sec:review UN}

Consider the polynomial ring $\mathbb{C}[Z,W]$ in $2N^2$ variables $Z_{ij}, W_{ij}$.
Given any $U \in U(N)$ we define the action on $Z,W$ by conjugation
\begin{equation}
    Z \mapsto UZU^\dagger, \quad 
    W \mapsto UWU^\dagger \,.
\end{equation}
This induces an action on the space of polynomials and the physical space of states in the $\su2$ sector of four-dimensional $\mathcal{N}=4$ SYM corresponds to the subspace of invariant polynomials.
\begin{definition}[Space of $U(N)$ invariants]
The subspace of $U(N)$ invariant polynomials is denoted by
\begin{equation}
    \mathcal{H}_N = \qty{ \poly (Z,W) \in \mathbb{C}[Z,W] \, \vert \, 
    \poly (UZU^\dagger , UWU^\dagger) = \poly (Z,W) \quad \forall U \in U(N)}.
\end{equation}
The space of invariants can be organised into subspaces of homogeneous bidegree, where $Z$ has degree $(1,0)$ and $W$ has degree $(0,1)$.
The subspaces of homogeneous degree $(m,n)$ are defined as
\begin{equation}
    \mathcal{H}_N(m,n) = \{ \poly(Z,W) \in \mathcal{H}_N \, \vert \, \poly(x Z, y W) = x^m y^n \poly(Z,W) \, \forall x,y \in \mathbb{C}\},
\end{equation}
and give a decomposition
\begin{equation}
    \mathcal{H}_N = \bigoplus_{m,n \geq 0} \mathcal{H}_N(m,n).
\end{equation}
\end{definition}
As a consequence of the first fundamental theorem of invariant theory for $GL(N)$ \cite{Weyl1939classical,GoodmanWallach2009}, the space $\mathcal{H}_N(m,n)$ is spanned by multi-traces.
Multi-traces form a basis in the stable regime, but have linear dependences in the unstable regime.
\begin{example}\label{ex: H_N(2,1)}
All elements in $\mathcal{H}_N(2,1)$ are linear combinations of 
\begin{equation}
\tr(ZZW), \quad \tr(ZZ)\tr(W), \quad \tr(ZW)\tr(Z), \quad \tr(Z)^2 \tr(W).
\end{equation}
In the stable regime $N \geq 3$, these elements form a basis. For $N < 3$ they are linearly dependent.
\end{example}

\subsubsection{Hilbert space at finite $N$}
Permutation centraliser algebras give a useful description of the space of invariants.
They form subalgebras of the group algebras $\mathbb{C}[S_L]$ of symmetric groups.
\begin{definition}[Symmetric group algebra]
Let $S_L$ be the symmetric group on $L$ letters.
The symmetric group algebra $\mathbb{C}[S_L]$ is the complex vector space whose basis is $S_L$
\begin{equation}
\mathbb{C}[S_L]
= \Big\{ \sum_{g\in S_L} c_g \, g \, \Big| \, c_g \in\mathbb{C} \Big\}.
\end{equation}
It becomes a finite-dimensional associative algebra when the group multiplication is extended bilinearly
\begin{equation}
\Big( \sum_{g \in S_L} c_{g}  g \Big)
\Big( \sum_{g' \in S_L} c_{g'} {g'} \Big)
= \sum_{g, g' \in S_L} c_{g}  \, c_{g'} \, g g' .
\end{equation}
\end{definition}
\begin{example}
    The product of elements $\sigma, \sigma' \in \mathbb{C}[S_3]$ given by
    \begin{equation}
        \sigma = (123) + (132), \quad \sigma' = (12) + (13)
    \end{equation}
    is given by
    \begin{equation}
        \sigma \sigma' = (123)(12) + (123)(13) + (132)(12) + (132)(13) = (12) + (13) + 2 \, (23).
    \end{equation}
We use the convention $(g_1g_2)(i)=g_2(g_1(i))$ following \cite{Padellaro:2024rld}.
\end{example}
The permutation centraliser algebras of interest are associated with subgroups $S_m \times S_n \subset S_{m+n}$ which permute the first $m$ objects among themselves and the last $n$ objects among themselves.
\begin{definition}[Permutation centraliser algebra \cite{Mattioli:2016eyp}]
The permutation centraliser algebra (PCA) at degree $(m,n)$ is the subspace of the group algebra made out of elements that commute with $S_m \times S_n$
\begin{equation}
\cA(m,n)
=
\{\sigma \in \mathbb{C}[S_{m+n}] \, \vert  \, \gamma \sigma \gamma^{-1}=\sigma, \quad \forall \gamma \in S_m \times S_n \}.
\label{def:PCA}
\end{equation}
\end{definition}
\begin{example}
    The group algebra element $\sigma = (13)(24) - (23)(14)$ is not in $\cA(2,2)$, but $\sigma = (13)(24) + (14)(23)$ is an element of $\cA(2,2)$.
\end{example}

The algebra $\cA(m,n)$ maps surjectively onto $\mathcal{H}_{N}(m,n)$.
Intuitively, this is because the following explicit map relates $\cA(m,n)$ to the space of all multi-trace invariants.
\begin{definition}[Permutations-to-invariants map]\label{def:trace map}
For $g\in S_{m+n}$\,, we define the multi-trace $U(N)$ invariant
\begin{equation}
\swmap^{(N)}_{m,n}(g) \equiv \sum_{i_1,\dots,i_{m+n}=1}^N
\prod_{p=1}^m Z_{i_p i_{g(p)}}
\prod_{q=m+1}^{m+n} W_{i_q i_{g(q)}} \,.
\label{eq:Og-definition}
\end{equation}
For a group-algebra element $\sigma=\sum_{g\in S_{m+n}}c_g(\sigma)g$, we linearly extend the above map,
\begin{equation}
\swmap^{(N)}_{m,n}(\sigma) \equiv
\sum_{g\in S_{m+n}}c_g(\sigma)\,\swmap^{(N)}_{m,n}(g) .
\label{eq:Osigma-definition}
\end{equation}
Since $\swmap^{(N)}_{m,n}(\sigma) = \swmap^{(N)}_{m,n}(h \sigma h^{-1})$ for any $h \in S_m \times S_n$\,, we obtain a map from permutation centraliser algebras to invariants
\begin{equation}
\swmap^{(N)}_{m,n}: \cA(m,n) \rightarrow \mathcal{H}_N(m,n).
\end{equation}
\end{definition}
\begin{example}
Consider $\sigma = (13)(24) + (14)(23)$ with $m=n=2$, then
    \begin{equation}
        \swmap^{(N)}_{2,2}(\sigma)
        = \sum_{i_1,i_2,i_3,i_4=1}^N \Big( Z_{i_1 i_3}Z_{i_2 i_4} W_{i_3 i_1} W_{i_4 i_2} + Z_{i_1 i_4}Z_{i_2 i_3} W_{i_3 i_2} W_{i_4 i_1} \Big)
        = 2 \tr(ZW)\tr(ZW).
    \end{equation}
\end{example}

The PCA $\cA(m,n)$ has a natural/combinatorial basis called the orbit basis.
For an element $g \in S_{m+n}$, the set 
\begin{equation}
    \{g\}_{m,n} =\{hgh^{-1}\mid h\in S_m \times S_n\},
\end{equation}
of elements related to $g$ by conjugation of the subgroup $S_m \times S_n$ is called the $S_m \times S_n$-orbit of $g$.
The group $S_{m+n}$ is partitioned into the set of distinct orbits.
Let $\{ g_1,\dots, g_l\}$ be a complete set of representatives of orbits and $l = \dim \cA(m,n)$, then
\begin{equation}
    S_L = \bigcup_{k=1}^{l} \{ g_k \}_{m,n} \,, \qquad
    \{ g_k \}_{m,n} \cap \{ g_{k'} \}_{m,n}= 
    \begin{cases}
        \{ g_k \}_{m,n}, &\text{ if $k = k'$}\\
        \emptyset, &\text{ if $k \neq k'$}
    \end{cases}
\end{equation}
Orbit representatives are used to define the orbit basis.
\begin{definition}[Orbit basis]
Let $\{ g_1,\dots, g_l\}$ be a complete set of representatives of $S_m \times S_n$ orbits of $S_{m+n}$.
For each representative $g_k$, define the following element in $\cA(m,n)$
\begin{equation}
    \orbitP{m,n}(g_k)=\sum_{h \in S_m \times S_n } hg_kh^{-1} \,.
\label{eq:orbit-sum}
\end{equation}
We extend $P_{m,n}$ linearly to $\mathbb{C}[S_{m+n}]$.
They form the orbit basis of $\cA(m,n)$
\begin{equation}
\cA(m,n)=
\Span_{\mathbb C}\left( \orbitP{m,n}(g_1),\dots, \orbitP{m,n}(g_l)\right).
\label{eq:orbit-basis}
\end{equation}
\end{definition}
\begin{example}
    \label{ex:cA22}
    The PCA $\cA(2,2)$ has dimension $l = 10$. A complete set of representatives is
    \begin{equation}
        \begin{aligned}
            &g_1 = (1)(2)(3)(4),\,
            &&g_2 = (1)(2)(34), \, 
            &&g_3 = (12)(3)(4), \,
            &&g_4 = (13)(2)(4), \,
            &&g_5 = (1)(234),   \,  \\
            &g_6 = (4)(123),    \,
            &&g_7 = (12)(34),   \,
            &&g_8 = (13)(24),   \,
            &&g_9 = (1234),     \,
            &&g_{10} = (1324)\, .
        \end{aligned}
    \end{equation}
    The choice of representatives is not unique, but the corresponding orbit sums are independent of the choice.
\end{example}

Permutation centraliser algebras have a natural Hilbert space (inner product) structure.
\begin{definition}[Orbit-basis inner product]
    \label{def:orbit-basis inner prod}
Define the linear function $\delta: \mathbb{C} [S_{m+n}] \rightarrow \mathbb{C}$ by
\begin{equation}
    \delta(g) = \begin{cases}
        1, &\text{ if $g=\mathrm{id}$}\\
        0, &\text{ if $g\neq \mathrm{id}$}.
    \end{cases}
\end{equation}
A natural inner product on $\cA(m,n)$ is given by
\begin{equation}
\langle \sigma, \sigma' \rangle = \delta(\sigma \cS(\sigma'))
\label{eq: inner prod}
\end{equation}
where $\cS$ is the involution
\begin{equation}
    \mathcal{S}\qty(\sum_{k=1}^l c_k \orbitP{m,n}(g_k)) = \sum_{k=1}^l c_k^* \orbitP{m,n}(g_k^{-1}).
\end{equation}
The orbit basis is orthogonal with respect to this inner product, which we call the orbit-basis inner product.
\end{definition}

\begin{definition}[Integer partitions]
The symbol $R \vdash L$ denotes an integer partition $R$ of $L$, or equivalently a Young diagram $R$ with $L$ boxes. The height of $R$ is the number of non-zero entries of $R$ denoted by $\ell(R)$. 
For example, the three integer partitions of $L=3$ have heights
\begin{equation}
\ell([3]) = 1, \qquad
\ell([2,1]) = 2, \qquad
\ell([1,1,1]) = 3.
\end{equation}
\end{definition}

The dimension of $\cA(m,n)$ is given by a sum of representation-theoretic numbers \cite{Collins:2008gc, Mattioli:2016eyp},
\begin{equation}
    \dim \cA(m,n) = \sum_{R \vdash L} g(R)
\label{eq: dimcA}
\end{equation}
where the sum runs over integer partitions (Young diagrams) $R \vdash L$, and
\begin{equation}
g(R) = \sum_{\substack{r_1 \vdash m \\ r_2 \vdash n}} g(r_1,r_2; R)^2 
= \sum_{\substack{\Lambda \vdash m+n \\ \ell(\Lambda) \leq 2}} C(R,R,\Lambda) K_{\Lambda, (m,n)} \,.
\label{eq: g(R)}
\end{equation}
Here $g(r_1,r_2; R)$ are Littlewood-Richardson coefficients, $C(R,R,\Lambda)$ are Kronecker coefficients and $K_{\Lambda, (m,n)}$ are Kostka numbers.

As mentioned, the permutations-to-invariants map is surjective but not injective for general $N$.
In other words, the map has a non-trivial kernel/nullspace, which we will describe using special eigenbases for $\cA(m,n)$.
These bases are called SEP bases and to define them we need to introduce the left action of the centre of the group algebra on $\cA(m,n)$.
\begin{definition}[{Centre of $\mathbb{C} [S_L]$}]
The centre of $\mathbb{C}[S_L]$ is the subspace of elements that commute with all of $\mathbb{C}[S_L]$
\begin{equation}
    \mathcal{Z} ( \mathbb{C}[S_L] ) = \{ \centz \in \mathbb{C}[S_L] \, \vert \, g \centz = \centz g \quad \forall g \in S_L \}.
\end{equation}
It is also equivalent to the permutation centraliser algebra $\cA(L,0)$.
\end{definition}
\begin{example}
    The relevant orbits for $\cA(L,0)$ are $S_L$-orbits of $S_L$, or equivalently, conjugacy classes of $S_L$.
    Therefore, the centre has an orbit basis labelled by conjugacy classes. For $L=3$ the orbit basis is given by
    \begin{equation}
        P_{L,0}(\mathrm{id}) \propto \mathrm{id}, \quad 
        P_{L,0}((12)) \propto (12) + (13) + (23), \quad 
        P_{L,0}((123)) \propto (123) + (132) \; .
    \end{equation}
\end{example}
\noindent
The centre acts on $\cA(m,n)$ by left multiplication.
\begin{definition}[Left centre action]\label{def:left center}
Let $\centz \in \mathcal{Z}(\mathbb{C}[S_L])$ have an expansion
\begin{equation}
    \centz = \sum_{g \in S_L} c_g(\centz) g.
\end{equation}
The left action of $\centz$ on $\sigma \in \cA(m,n)$ is given by
\begin{equation}
m^{\mathfrak{L}}: \mathcal{Z}(\mathbb{C}[S_{m+n}]) \otimes \cA(m,n) \rightarrow \cA(m,n) , \qquad
m^{\mathfrak{L}}[\centz ](\sigma) = \sum_{g \in S_L} c_g(\centz ) g \sigma.
\end{equation}
\end{definition}

For a central element satisfying $\cS(\centz)=\centz$, the left action is self-adjoint with respect to the orbit-basis inner product \eqref{eq: inner prod} \cite{Padellaro:2024rld}.
Since all central elements commute, it follows that there exists a simultaneous eigenbasis (of the centre) for $\cA(m,n)$.
\begin{definition}[SEP basis]
\label{def:finite-N basis}
A basis for $\cA(m,n)$ is called SEP
if it simultaneously diagonalises the left action of the centre.
Such bases can always be labelled by an integer partition $R \vdash L$ and $a \in \{1,\dots, g(R)\}$, where $g(R)$ is defined in equation \eqref{eq: g(R)}.
For example, write
\begin{equation}
    \mathcal{Q}^{R}_a = \sum_{k=1}^l \mathcal{Q}^{R}_a(g_k) \orbitP{m,n}(g_k)
\label{cQRa expand}
\end{equation}
then the coefficients $\mathcal{Q}^{R}_a(g_k)$ are constrained by demanding
\begin{equation}
    m^{\mathfrak{L}}[ \centz ](\mathcal{Q}^R_a) = \frac{\chi^R(\centz )}{d_R} \mathcal{Q}^R_a.
\label{left-center Q^R_a}
\end{equation}
Here $\chi^R (\sigma)$ is the $S_L$ character for $\sigma \in \bb{C}[S_L]$, $d_R$ is the dimension of the irreducible representation of $S_L$ labelled by $R$, and $\frac{\chi^R(\centz )}{d_R}$ is called a central character.
\end{definition}
\begin{example}
    The SEP bases for $L=m$ are (proportional to) central projectors. For example, at $L=3$ we have
    \begin{equation}
        \cQ^{[3]} \propto \sum_{g \in S_3} g, \quad \cQ^{[2,1]} \propto 2\, \mathrm{id} - (123) - (132), \quad \cQ^{[1,1,1]} \propto \sum_{g \in S_3} \mathrm{sgn}(g)  g.
    \end{equation}
\end{example}
\begin{remark}
\label{remark: multiplicity}
It is often useful to further refine/resolve the multiplicity label $a$ using operators that commute with the left action of the centre.
For example, one resolution involves picking out fixed irreducible representations $\Lambda$ of the $U(2)$ symmetry rotating $Z$ and $W$, which is a global symmetry of four-dimensional $\mathcal{N}=4$ SYM.
We will come back to this point at the end of this Subsection.
\end{remark}

Cutoffs on the height $\ell(R)$ of the Young diagram labels on SEP bases define subspaces of $\cA(m,n)$ that describe the image and kernel of the permutations-to-invariants map.
\begin{definition}[Finite-$N$ PCA subspace]
\label{def:cAN}
Let $\cQ^R_a$ be a SEP basis.
For each positive integer $N$, $\cA(m,n)$ has a subspace
\begin{equation}
    \cA_N(m,n) = \Span_\mathbb{C}(\cQ^{R}_a \, \vert \, R \vdash (m+n), \ell(R) \leq N, a \in \{1,\dots,g(R)\}) .
\end{equation}
\end{definition}
\noindent
As we will now prove, the finite-$N$ PCA subspace is isomorphic to the image of the permutations-to-invariants map and the kernel forms a complement \cite{Corley:2001zk, Brown:2007xh}.
We introduce notation for this isomorphism and state it as a proposition.
\begin{mainproposition}
\label{prop: AN iso HN}
Let the restriction of the permutations-to-invariants map to $\cA_N(m,n)$ be denoted
\begin{equation}
    \hat \swmap^{(N)}_{m,n}: \cA_N(m,n) \rightarrow \cH_N(m,n).
\end{equation}
It gives an isomorphism
\begin{equation}
    \cA_N(m,n) \cong \mathcal{H}_N(m,n).
\end{equation}
\end{mainproposition}

To prove this, we will use the so-called $U(N)$ inner product. 
\begin{definition}[$U(N)$ inner product]\label{prop:UN inner}
    The $U(N)$ inner product on $\cA(m,n)$ is defined as the pullback of the free two-point function of $\mathcal{N}=4$ SYM with $U(N)$ gauge group.
For $\sigma_1 , \sigma_2 \in \cA(m,n)$ 
\begin{equation} 
\begin{aligned}
\langle \sigma_1,\sigma_2\rangle_{U(N)} &\equiv 
\left\langle \swmap^{(N)}_{m,n}(\sigma_1) \, \swmap^{(N)}_{m,n}(\sigma_2)^\dagger \right\rangle_{\cN=4} ,
\\
\text{where } \swmap^{(N)}_{m,n}(g)^\dagger &= \sum_{i_1,\dots,i_L=1}^N
\prod_{p=1}^m Z^\dagger_{i_p i_{g^{-1}(p)}} \prod_{q=m+1}^L W^\dagger_{i_q i_{g^{-1} (q)}} \,,
\end{aligned}
\label{eq:UN-pullback} 
\end{equation}
and the free two-point function is computed from the elementary $U(N)$ Wick contractions
\begin{equation} 
\left\langle Z_{ij}Z^\dagger_{kl}\right\rangle_{\cN=4}
=\left\langle W_{ij}W^\dagger_{kl}\right\rangle_{\cN=4}
=\delta_{il}\delta_{jk}. 
\label{eq:UN-wick} 
\end{equation}
\end{definition}
 This two-point function is positive definite, because it agrees with the expectation values defined by a Gaussian integral on the space of two complex $N\times N$ matrices. Hence $\Vev{O O^\dagger}_{\cN=4}=0$ if and only if $O=0$.
The $U(N)$ inner product is related to the orbit-basis inner product \eqref{eq: inner prod} by
\begin{equation}
\langle \sigma_1, \sigma_2 \rangle_{U(N)} = m! n! \, \langle \Omega_N \, \sigma_1, \sigma_2 \rangle 
= \sum_{h \in S_m \times S_n}
\delta \Big( \Omega_N h \, \sigma_1 h^{-1} \cS (\sigma_2 ) \Big), \qquad
\Omega_N = \sum_{h \in S_L} N^{{\rm cycles} (h)} h .
\label{eq: U(N) inner prod}
\end{equation}
where ${\rm cycles}(h)$ is the number of permutation cycles in $h \in S_L$ (see \cite{Corley:2001zk}).
In other words, the orbit-basis inner product agrees, up to a normalisation factor, with the $N\rightarrow\infty$ limit of the $U(N)$ inner product \cite{Padellaro:2024rld}.
The element $\Omega_N$ belongs to the centre $\mathcal{Z}(\mathbb{C}[S_{m+n}])$ and is therefore diagonal in a SEP basis.
The eigenvalue of $\cQ^R_a$ is equal to the product of box weights
\begin{equation}
f^R =\prod_{(i,j)\in R}(N+j-i) , \qquad
\begin{cases}
f^R=0, & \text{ if $\ell(R)> N$} \\
f^R > 0, & \text{ if $\ell(R)\le N$}. \\
\end{cases}
\label{eq:fR-vanishing}
\end{equation}

The following result on the kernel of the permutations-to-invariants map was given in \cite{Brown:2007xh}, formulated using covariant bases.
For completeness we give a small extension of the proof using general SEP bases.
\begin{mainproposition}\label{thm:UN-inv-fin}
The kernel of the permutations-to-invariants map is spanned by the SEP-basis elements labelled by Young diagrams with more than $N$ rows
\begin{equation}
    \ker(\swmap^{(N)}_{m,n})=\Span_{\mathbb C}\bigl\{\cQ^R_a\mid\ell (R)>N\bigr\}. 
\label{eq:radical-kernel} 
\end{equation}
The restriction of $\langle\cdot,\cdot\rangle_{U(N)}$ to $\cA_N(m,n)$ is positive definite.
Moreover, the label $R$ in SEP bases gives an orthogonal decomposition with respect to the orbit-basis inner product:
\begin{equation}
    \cA(m,n) = \cA_N(m,n) \oplus \ker\bigl(\swmap^{(N)}_{m,n}\bigr).
\label{eq:PCA-orth-decomp}
\end{equation}
\end{mainproposition}
\begin{proof}
For the first claim, decompose an arbitrary element $\sigma\in\cA(m,n)$ into its $R$-sector components in the SEP basis:
\begin{equation}
\sigma=\sum_R\sigma_R, \qquad \sigma_R=\sum_a\sigma_{R,a} \, \cQ^R_a.
\end{equation}
Distinct $R$-sectors are orthogonal in the orbit-basis inner product, and $\Omega_N$ acts on each sector by the scalar $f^R$.
The $U(N)$ inner product in Definition \ref{prop:UN inner} therefore expresses the free two-point function as
\begin{equation}
    \left\langle \swmap^{(N)}_{m,n}(\sigma) \, \swmap^{(N)}_{m,n}(\sigma)^\dagger \right\rangle_{\cN=4}
    = \langle \sigma,\sigma\rangle_{U(N)}
    = m!n! \sum_R f^R \langle \sigma_R,\sigma_R\rangle .
\label{eq:SEP-norm}
\end{equation}
The orbit-basis inner product is positive definite, so $\langle \sigma_R,\sigma_R\rangle>0$ whenever $\sigma_R\neq0$.
Since $f^R>0$ for $\ell(R)\leq N$ and $f^R=0$ otherwise, the right-hand side vanishes if and only if $\sigma_R=0$ for every $R$ with $\ell(R)\leq N$.
The free two-point function is positive definite on $\cH_N(m,n)$, so the left-hand side vanishes if and only if $\sigma\in\ker(\swmap^{(N)}_{m,n})$.
Comparing these conditions proves \eqref{eq:radical-kernel}.
The same identity shows that the restriction of the $U(N)$ inner product to $\cA_N(m,n)$ is positive definite.
For the second claim, the two summands in \eqref{eq:PCA-orth-decomp} are spanned by the $R$-sectors with $\ell(R)\leq N$ and $\ell(R)>N$, respectively.
They exhaust $\cA(m,n)$ and are orthogonal in the orbit-basis inner product, proving the orthogonal decomposition.
\end{proof}

From basic linear algebra, the proposition we wanted to prove follows.
For completeness we outline the argument.
\begin{proof}[Proof of Proposition~\ref{prop: AN iso HN}]
The permutations-to-invariants map $\swmap^{(N)}_{m,n}: \cA(m,n) \rightarrow \cH_N(m,n)$ is a linear map between finite-dimensional vector spaces.
We use the result from linear algebra that says that the restricted map is an isomorphism to the image.
In the present case, this says that we have an isomorphism
\begin{equation}
    \hat \swmap^{(N)}_{m,n}: \cA_N(m,n) \rightarrow \im \swmap^{(N)}_{m,n}
\end{equation}
Since $\swmap^{(N)}_{m,n}$ is surjective, $\im \swmap^{(N)}_{m,n}  = \cH_N(m,n)$ and we get $\cA_N(m,n) \cong \cH_N(m,n)$.
In other words, the restricted permutations-to-invariants map is an isomorphism to the space of invariants.
\end{proof}

We now return to the Remark \ref{remark: multiplicity} on resolving the multiplicity labels of SEP basis elements.
The adjoint(conjugation) action of $S_L$ on $\bb{C}[S_L]$ is related to permutations of the field contents of the corresponding invariant. 
To see this, we introduce the vector $X^a_{ij} = (Z_{ij}, W_{ij})$. 
Note that for every $\gamma \in S_L$ we have the relabelling identity
\begin{equation}
\sum_{i_1, i_2 , \dots , i_L=1}^N 
X^{a_{\gamma(1)}}_{i_1 i_{g(1)}} \, X^{a_{\gamma(2)}}_{i_2 i_{g(2)}} \, \dots \, X^{a_{\gamma(L)}}_{i_L i_{g(L)}} 
=
\sum_{j_1, j_2 , \dots , j_L=1}^N 
X^{a_1}_{j_1 j_{h(1)}} \, X^{a_2}_{j_2 j_{h(2)}} \, \dots \, X^{a_L}_{j_L j_{h(L)}} \,, \qquad
h = \gamma^{-1} g \gamma\,.
\label{symm cov perm}
\end{equation}
The vector $X^a_{ij}$ forms the fundamental representation of $U(2)$ and the relabelling identity implies that the adjoint action is Schur-Weyl dual to the action of $U(2)$.
The covariant basis \cite{Brown:2007xh,Brown:2008ij} decomposes $\mathcal{H}_N$ into irreducible representations of $U(2)$.
In the language of PCAs, this decomposition corresponds to the diagonalisation of an adjoint action of the centre \cite{Kemp:2019log,Padellaro:2024rld} (see also \cite{Bornman:2016zwn} for the action of $U(2)$ on the restricted Schur basis).
\begin{definition}[Adjoint centre action]
The adjoint action of $\centz \in \mathcal{Z}(\mathbb{C}[S_{m+n}])$ on $\sigma \in \cA(m,n)$ is given by
    \begin{equation}
m^{\mathfrak{ad}}: \mathcal{Z}(\mathbb{C}[S_{m+n}]) \otimes \cA(m,n) \rightarrow \cA(m,n) , \qquad
m^{\mathfrak{ad}}[\centz](\sigma) = \sum_{g \in S_{m+n}} c_g(\centz) \, g \sigma g^{-1}.
    \end{equation}
For $\cS(\centz)=\centz$, the adjoint action is self-adjoint with respect to the orbit-basis inner product and commutes with the left action of the centre. 
The simultaneous eigenbases are labelled by $R \vdash (m+n)$ and $\Lambda \vdash (m+n)$ with $\ell(\Lambda) \leq 2$.
The partition $\Lambda = [\lambda_1,\lambda_2]$ labels an irreducible representation of $U(2)$.
Its restriction to $SU(2)$ has spin $(\lambda_1-\lambda_2)/2$, and its total $U(1)$ degree is $\lambda_1+\lambda_2$\,.
\end{definition}

Every irreducible representation $\Lambda$ of $U(2)$ has a highest-weight state (HWS) vector, unique up to scale.
For a fixed $\Lambda=[m,n]$ with $m\geq n\geq0$, these HWS vectors lie in $\mathcal{H}_N(m,n)$. 
Applying the lowering operator of $u(2)$
\begin{equation}
    J_{-} = \tr (W\pdv{}{Z})
\end{equation}
to the highest-weight vector in $\mathcal{H}_N(m,n)$ produces a state in $\mathcal{H}_N(m-1,n+1)$ if $m>n$, and a zero vector if $m=n$.
Applying the quadratic Casimir $\hat C_2$ of $u(2)$ leaves $(m,n)$ invariant and multiplies the state by an eigenvalue $C_2(\Lambda)$.
The $u(2)$ Casimir is related to the adjoint action of the central element \cite{Kimura:2008ac}
\begin{equation}
T_2 = \sum_{i<j}(ij) \in \mathcal{Z}(\bb{C}[S_{m+n}]), \qquad
m^{\mathfrak{ad}}[T_2](\sigma) = \sum_{1\leq i<j\leq L}(i,j)\,\sigma\,(i,j), \qquad \sigma\in\cA(m,n).
\label{eq:adjoint-T2}
\end{equation}
The eigenvalue of the adjoint action in the $\Lambda$ sector is the normalised character $\chi^\Lambda(T_2)/d_\Lambda$.
Therefore, the HWS subspace of matrix invariants in the irreducible representation $\Lambda = [m,n]$ is given by 
\begin{equation}
    \label{eq: cH HWS}
    \cH_N^{\mathrm{HWS}}(m,n) = \{ \poly(Z,W) \in \mathcal{H}_{N}(m,n) \, \vert \, \hat C_2 \poly(Z,W) = C_2([m,n]) \poly(Z,W) \}.
\end{equation}
Correspondingly, by the permutations-to-invariants map, it is the space
\begin{equation}
    \label{eq: cA HWS}
    \cA_N^{\mathrm{HWS}}(m,n) = \{ \sigma \in \cA_N(m,n) \, \vert \, m^{\mathfrak{ad}}[T_2](\sigma) = \frac{\chi^{[m,n]}(T_2)}{d_{[m,n]}} \sigma\}.
\end{equation}
We will also use the stable HWS spaces
\begin{equation}
    \label{eq: cA HWS stable}
    \cA^{\mathrm{HWS}}(m,n) = \{ \sigma \in \cA(m,n) \, \vert \, m^{\mathfrak{ad}}[T_2](\sigma) = \frac{\chi^{[m,n]}(T_2)}{d_{[m,n]}} \sigma\}.
\end{equation}
The dimension of $\cA_N^{\rm HWS} (m,n)$ follows from \eqref{eq: dimcA} and Definition \ref{def:cAN} as
\begin{equation}
\dim \cA_N^{\mathrm{HWS}}(m,n) = \sum_{\substack{R \vdash L \\ \ell(R) \le N}} C(R,R, [m,n]) 
\label{eq: dimcA HWS}
\end{equation}
where we used $K_{[m,n], (m,n)} = 1$ \cite{Brown:2007xh,Brown:2008ij}.

\subsubsection{One-loop dilatation operator}

We now move on to review the one-loop dilatation operator in permutation language.
The dilatation operator in the $\su2$ sector of $\mathcal N=4$ SYM can be expanded at small coupling as \cite{Beisert:2003tq}
\begin{equation}
     \fD_N =  \fD^{(0)}_N+\frac{N g_{\rm YM}^2}{8\pi^2} \,  \fD^{(1)}_N+\cdots 
\label{eq:dilatation-expansion}
\end{equation}
The tree-level (first) term counts the number of fields/letters 
\begin{equation}
     \fD^{(0)}_N = \tr\!\left(Z\frac{\partial}{\partial Z}+W\frac{\partial}{\partial W}\right), \qquad
\end{equation}
meaning that every $\poly \in \mathcal{H}_N(m,n)$ satisfies $ \fD^{(0)}_N \poly = L \, \poly$.
The one-loop (second) term is given by
\begin{equation}
     \fD^{(1)}_N = -\frac1N :\tr\!\left([Z,W]\left[\frac{\partial}{\partial Z},\frac{\partial}{\partial W}\right]\right): \ .
\end{equation}
The dilatation operator is Hermitian with respect to the two-point functions of $\mathcal{N}=4$ SYM.
Hence, it is diagonalisable and all eigenvalues are real. It also commutes with the action of $U(2)$ on $\mathcal{H}_N$ rotating $Z$ and $W$ into each other.

The one-loop dilatation can be formulated as a linear operator on permutation algebras, following \cite{Bellucci:2004ru,Pasukonis:2010rv}.
\begin{definition}[One-loop permutation dilatation operator]\label{def: 1LPDO}
For $h \in S_L$, define the intermediate map
\begin{equation}
    \bar H^{(1)}_{\gss}(h) \equiv \sum_{\substack{i=1\\ i\ne h(i)}}^L
\left[ h - (i,h(i))\, h\,(i,h(i)) \right]
 +{\gss} \sum_{i=1}^L \sum_{\substack{j=1\\ h(j)\ne i\\ j\ne i}}^L
\left[ h\, (i,h(j)) - (i,j)\, h\,(i,h(j))\,(i,j) \right].
\label{eq:D1-orbit-basis}
\end{equation}
This map is covariant under conjugation as shown in \eqref{app:D1 SL covariance},
\begin{equation}
\bar H^{(1)}_{\gss}(\gamma h\gamma^{-1}) = \gamma\bar H^{(1)}_{\gss}(h)\gamma^{-1}, \qquad \gamma\in S_L.
\label{eq:H1-conjugation}
\end{equation}
Consequently, $P_{m,n}\circ\bar H^{(1)}_{\gss}$ is independent of the representative chosen within each $S_m\times S_n$ orbit.
In the orbit basis, with representatives $\{g_1,\dots,g_l\}$, define the one-loop permutation dilatation operator by
\begin{equation}
    \begin{aligned}
        & H^{(1)}_{\gss}: \cA(m,n) \rightarrow \cA(m,n) \\
        & H^{(1)}_{\gss}(P_{m,n}(g_k)) \equiv P_{m,n} \circ \bar H^{(1)}_{\gss}(g_k).
    \end{aligned}
\end{equation}
The first term of \eqref{eq:D1-orbit-basis} is the planar contribution, while the second term is the non-planar contribution and carries an explicit factor of ${\gss}$.
Writing the full composition explicitly gives
\begin{equation}
    \begin{aligned}
    & H^{(1)}_{\gss}(\orbitP{m,n} (g_k))
    =
\sum_{\substack{i=1\\ i\ne g_k(i)}}^L
\left\{
\orbitP{m,n} (g_k) - \orbitP{m,n}\!\left((i,g_k(i))g_k(i,g_k(i))\right)
\right\}
\\
    &\quad\quad\quad\quad\quad\quad\quad\quad +{\gss}
\sum_{i=1}^L
\sum_{\substack{j=1\\ g_k(j)\ne i\\ j\ne i}}^L
\left\{
\orbitP{m,n} \!\left(g_k(i,g_k(j))\right)
- \orbitP{m,n} \!\left((i,j)g_k(i,g_k(j))(i,j)\right)
\right\}.
    \end{aligned}
\end{equation}
\end{definition}

\begin{mainproposition}\label{prop:D1 intertwining}
The one-loop permutation dilatation operator satisfies the intertwining relation via the permutations-to-invariants map,
\begin{equation}
\fD_N^{(1)}\!\left(\swmap^{(N)}_{m,n}(\sigma)\right) = \left.\swmap^{(N)}_{m,n}\!\left(H^{(1)}_{\gss}(\sigma)\right)\right|_{\gss=1/N}, \qquad \sigma\in\cA(m,n), \quad \forall N \in \mathbb{N}.
\label{eq:all-N intertwining}
\end{equation}
\end{mainproposition}
\begin{proof}
The direct calculation in Appendix \ref{app:induced action D1} gives
\begin{equation}
\fD_N^{(1)}\!\left(\swmap^{(N)}_{m,n}(g)\right) 
= \swmap^{(N)}_{m,n}\!\left( \bar H^{(1)}_{\xi} (g)\right) \Big|_{\xi=1/N} , \qquad g \in S_L\,, \quad \forall N \in \mathbb{N}
\label{eq:direct calc D1}
\end{equation}
where the right-hand side is given by \eqref{eq:D1-orbit-basis}.
Equation \eqref{eq:direct calc D1} is the pointwise identity for each $g \in S_L$.
By using the conjugation property of the permutations-to-invariants map,
\begin{equation}
\swmap^{(N)}_{m,n}\!\left(P_{m,n}(g)\right) = m!n!\,\swmap^{(N)}_{m,n}(g)
\end{equation}
and $\bar H^{(1)}_{\gss}$ in \eqref{eq:H1-conjugation}, we can rewrite \eqref{eq:direct calc D1} as the action on the orbit basis of $\cA(m,n)$.
\end{proof}

\subsection{Factorisation of characteristic polynomials at finite $N$}
In this subsection we prove the main theorem of this section (Theorem \ref{prop:finite-N factorisation}).
To make the proofs explicit, we will use SEP bases throughout this subsection.
In particular, let $\cQ^R_a$ be a SEP basis for $\cA(m,n)$ and define
\begin{equation}
    O^R_{a} = \hat \swmap^{(N)}_{m,n}(\cQ^R_a) \in \mathcal{H}_N(m,n).
\end{equation}
Recall that these elements form a basis for $\mathcal{H}_N(m,n)$
\begin{equation}
    \mathcal{H}_N(m,n) = \Span \Bigl( O^R_a \, \Big| \, R \vdash (m+n), \ell(R) \leq N, 1\leq a \leq g(R) \Bigr) .
\end{equation}
This choice of basis is very convenient because
\begin{equation}
    \label{eq: swmap SEP basis}
    \swmap^{(N)}_{m,n}(\cQ^R_a) =
    \begin{cases}
        O^R_a, &\text{ if $\ell(R) \leq N$}\\
        0, &\text{ if $\ell(R) > N$}.
    \end{cases}
\end{equation}
In the end, we will see that the statements we want to make are independent of the choice.

Let us first consider the relationship between $ \fD^{(1)}_N$ and $ H_{\gss}^{(1)}$ in these bases.
Proposition~\ref{prop:D1 intertwining} gives the intertwining relation
\begin{equation}
     \fD_N^{(1)}\swmap^{(N)}_{m,n}(\cQ^R_a)
    =
    \swmap^{(N)}_{m,n}( H_{\gss}^{(1)}(\cQ^R_a))
    \Big\vert_{\gss=\frac{1}{N}}\,.
\end{equation}
By \eqref{eq: swmap SEP basis}, the left-hand side vanishes for $\ell(R)>N$ and can be written as follows:
\begin{equation}
    \label{eq: Dside2}
     \fD_N^{(1)}\swmap^{(N)}_{m,n}(\cQ^R_a) =
    \begin{cases}
         \fD_N^{(1)} O^R_a, &\text{ if $\ell(R) \leq N$}\\
        0, &\text{ if $\ell(R) > N$.}
    \end{cases}
\end{equation}
For $\ell(R)\leq N$, the action of $ \fD_N^{(1)}$ on $O^R_a$ gives
\begin{equation}
    \label{eq: Dside}
     \fD_N^{(1)}  O^R_a  = \sum_{\substack{S \vdash (m+n) \\ \ell(S) \leq N}}\sum_{b=1}^{g(S)} 
     ( \fD_N^{(1)})_{ba}^{SR} O^{S}_b 
\end{equation}
where $( \fD_N^{(1)})_{ba}^{SR}$ are the matrix coefficients in the SEP basis.
On the right-hand side, expanding the action of $ H_{\gss}^{(1)}$ gives
\begin{equation}
    \label{eq: Hside}
    \swmap^{(N)}_{m,n}( H_{\gss}^{(1)}(\cQ^R_a))
    \Big\vert_{\gss=\frac{1}{N}}
    =
    \swmap^{(N)}_{m,n} \Big(
    \sum_{S \vdash (m+n) } \sum_{b=1}^{g(S)}
    (H_{\gss}^{(1)})_{ba}^{SR} \, \cQ^S_b \Big) \Big|_{\gss=\frac{1}{N}}
    =
    \sum_{\substack{S \vdash (m+n) \\ \ell(S)\leq N}}
    \sum_{b=1}^{g(S)}
    \left.( H_{\gss}^{(1)})_{ba}^{SR}\right|_{\gss=\frac{1}{N}}
    O^S_b ,
\end{equation}
where we used \eqref{eq: swmap SEP basis} to truncate the sum.

This implies a block triangular structure for $( H_{\gss}^{(1)})_{ba}^{SR}$ at $\gss=1/N$.
Examples of matrix elements are presented in Appendix \ref{app:mixing matrix data}.
\begin{lemma}
\label{cor:finite-N factorisation}
Let $( H_{\gss}^{(1)})_{ba}^{SR}$ be the matrix elements of the one-loop permutation dilatation operator in a SEP basis $\cQ^R_a$.
It has block triangular structure at $\gss = 1/N$
\begin{equation}
    ( H_{\gss}^{(1)})_{ba}^{SR}\Big \vert_{\gss=\frac{1}{N}} = \begin{pmatrix}  \fD_N^{(1)} & 0 \\ H_{21}^{(1)} & H_{22}^{(1)} \end{pmatrix}, 
\end{equation}
where the first and second blocks correspond to $\cA_N(m,n)$ and $\ker(\swmap^{(N)}_{m,n})$, respectively.
The zero block corresponds to $\ell(R)>N$ and $\ell(S)\leq N$.
The matrix in the top-left block is the operator $ \fD_N^{(1)}$ evaluated in the basis $O^R_a = \hat \swmap^{(N)}_{m,n}(\cQ^R_a)$.
\end{lemma}
\begin{proof}
    We consider two cases.
    The first case is $R$ with $\ell(R) \leq N$ such that $\cQ^R_a \in \cA_N(m,n)$.
    Extracting coefficients of $O^S_b$ in \eqref{eq: Hside} and \eqref{eq: Dside}  gives the equalities
    \begin{equation}
        ( H_{\gss}^{(1)})_{ba}^{SR}\Big \vert_{\gss=\frac{1}{N}} = ( \fD_N^{(1)})_{ba}^{SR}, \qquad \forall R,S \vdash (m+n) \text{ s.t. } \ell(R) \leq N, \ell(S) \leq N.
    \end{equation}
    This is the upper-left block in the lemma statement.
    The second case is $\ell(R) > N$ such that $\cQ^R_a \in \ker \swmap^{(N)}_{m,n}$.
    In this case, extracting the coefficient of \eqref{eq: Dside2} gives zero and therefore
    \begin{equation}
        ( H_{\gss}^{(1)})_{ba}^{SR}\Big \vert_{\gss=\frac{1}{N}} = 0 \qquad \forall R,S \vdash (m+n) \text{ s.t. } \ell(R) > N, \ell(S) \leq N.
    \end{equation}
\end{proof}

The main theorem, which we recall for convenience, follows from the lemma.
We give examples of the Theorem in subsection \ref{subsec:ex charpoly}.
\begin{remaintheorem}[Polynomial factorisation]{prop:finite-N factorisation}
Consider the characteristic polynomials
    \begin{equation}
        \det_{\mathcal{H}_N(m,n)}(x -  \fD_N^{(1)}), \quad \det_{\cA(m,n)}(x -  H^{(1)}_{\gss}).
    \end{equation}
    For every positive integer $N$, the first polynomial divides the second after setting $\gss=N^{-1}$:
    \begin{equation}
        \det_{\cA(m,n)}(x -  H^{(1)}_{\gss})\Big\vert_{\gss=\frac{1}{N}} = Q^{(N)}_{m,n}(x)\det_{\cH_N(m,n)}(x - \fD_N^{(1)}) \, ,
    \end{equation}
    where $Q^{(N)}_{m,n}(x)$ is a polynomial factor. 
\end{remaintheorem}
\begin{proof}
    By Lemma \ref{cor:finite-N factorisation}, at $\gss=N^{-1}$ the characteristic polynomial is the product of the characteristic polynomials of the diagonal blocks in a SEP basis. 
\end{proof}
\begin{remark}
In the orbit basis, the one-loop permutation dilatation operator has rational entries.
SEP bases have integer vector solutions \cite{Padellaro:2024rld} and therefore the one-loop dilatation operator in a SEP basis has rational coefficients as well.
It follows that both factors are polynomials with rational coefficients.
\end{remark}

Note that the block triangular structure implies that every $\cQ^R_a$ with $\ell(R) > N$ gets mapped back into $\ker \swmap^{(N)}_{m,n}$.
Consequently, any linear combination also stays within the kernel and we get the basis independent version of lemma \ref{cor:finite-N factorisation}
\begin{corollary}
    \label{lemma:finite-N factorisation}
    Let $\swmap^{(N)}_{m,n}$ be the map that takes elements of $\cA(m,n)$ to invariant polynomials in $\mathcal{H}_{N}(m,n)$ and let $\ker\swmap^{(N)}_{m,n}$  be the subspace of elements that are mapped to zero.
    The one-loop permutation dilatation operator $ H^{(1)}_{\gss}: \cA(m,n) \rightarrow \cA(m,n)$ leaves $\ker\swmap^{(N)}_{m,n}$ invariant for $\gss=\tfrac{1}{N}$.
\end{corollary}

It is also useful to have an independent statement about the top-left block.
For this, we introduce the orthogonal (with respect to the orbit-basis inner product) projector to $\cA_N(m,n)$
\begin{equation}
    P^{\cA_N(m,n)}: \cA(m,n) \rightarrow \cA_N(m,n).
\end{equation}
Restricting the one-loop permutation dilatation operator to act on $\cA_N(m,n)$ defines a composite map
\begin{equation}
    P^{\cA_N(m,n)} \circ  H^{(1)}_\gss: \cA_N(m,n) \rightarrow \cA_N(m,n).
\end{equation}
It is directly related to $ \fD^{(1)}_N$ by the restricted permutations-to-invariants map.
\begin{corollary}
\label{prop: unstable 1LPDO}
    The composite operator 
   \begin{equation}
       P^{\cA_N(m,n)} \circ  H^{(1)}_\gss\Big\vert_{\gss = 1/N}: \cA_N(m,n) \rightarrow \cA_N(m,n)
   \end{equation}
   is equivalent to the one-loop dilatation operator:
   \begin{equation}
       \label{eq: unstable 1LPDO prf}
       P^{\cA_N(m,n)} \circ  H^{(1)}_\gss \vert_{\gss=1/N} = (\hat \swmap^{(N)} _{m,n})^{-1} \circ \fD^{(1)}_N \circ \hat \swmap^{(N)}_{m,n} \, .
   \end{equation}
\end{corollary}
\begin{proof}
The restricted permutations-to-invariants map is an isomorphism and therefore has a unique inverse.
For consistency with the rest of this section, take a SEP basis $\cQ^R_a$ with $\ell(R) \leq N$ for $\cA_N(m,n)$ and $O^R_a = \hat \swmap^{(N)}_{m,n}(\cQ^R_a)$ a basis for $\cH_N(m,n)$.
This gives
\begin{equation}
    (\hat \swmap^{(N)} _{m,n})^{-1}(O^R_a) = \cQ^R_a \quad \text{and}  \quad 
    P^{\cA_N(m,n)}(\cQ^R_a) = 
    \begin{cases}
        \cQ^R_a, &\text{ if $\ell(R) \leq N$ }\\
        0, &\text{ if $\ell(R) > N$. }
    \end{cases}
\end{equation}
and the l.h.s. of \eqref{eq: unstable 1LPDO prf} equals
\begin{equation}
    \begin{aligned}
        P^{\cA_N(m,n)} \circ  H^{(1)}_\gss(\cQ^R_a)\Big\vert_{\gss = 1/N}&= 
    \sum_{S\vdash m+n} \sum_{b = 1}^{g(S)}  ( H^{(1)}_\gss)_{ba}^{SR}\Big\vert_{\gss=1/N} P^{\cA_N(m,n)}(\cQ^S_b) \\
    &=
    \sum_{\substack{S\vdash m+n \\ \ell(S) \leq N}} \sum_{b = 1}^{g(S)}  ( H^{(1)}_\gss)_{ba}^{SR}\Big\vert_{\gss=1/N} \cQ^S_b
    \end{aligned}
\end{equation}
while the r.h.s. equals
\begin{equation}
    \begin{aligned}
        ((\hat \swmap^{(N)} _{m,n})^{-1} \circ \fD^{(1)}_N \circ \hat \swmap^{(N)}_{m,n})(\cQ^R_a) 
        &= 
        ((\hat \swmap^{(N)} _{m,n})^{-1} \circ \fD^{(1)}_N)(O^R_a) \\
        &= \sum_{\substack{S\vdash m+n \\ \ell(S) \leq N}} \sum_{b = 1}^{g(S)}  ( \fD^{(1)}_N)_{ba}^{SR} (\hat \swmap^{(N)} _{m,n})^{-1}(O^S_b) \\
        &= \sum_{\substack{S\vdash m+n \\ \ell(S) \leq N}} \sum_{b = 1}^{g(S)}  ( \fD^{(1)}_N)_{ba}^{SR} \cQ^S_b.
    \end{aligned}
\end{equation}
From Lemma \ref{cor:finite-N factorisation} the matrix elements agree and the two sides are equal.
Lastly, since the equality holds for all basis elements it holds at the level of operators.
\end{proof}

This composite operator satisfies a compatibility condition with respect to the $U(N)$ inner product.
\begin{corollary}\label{prop: unstable 1LPDO-2}
For any $\sigma_a,\sigma_b\in\cA_N(m,n)$, the operators in Corollary \ref{prop: unstable 1LPDO} satisfy the equation
\begin{equation}
\left.\left\langle \sigma_b,
P^{\cA_N(m,n)}\circ  H_{\gss}^{(1)}(\sigma_a)
\right\rangle_{U(N)}\right|_{\gss=1/N}
=
\left\langle
\swmap^{(N)}_{m,n}(\sigma_b)\,
\left(\fD_N^{(1)}\swmap^{(N)}_{m,n}(\sigma_a)\right)^\dagger
\right\rangle_{\cN=4}.
\label{eq:finite-N physical matching}
\end{equation}
\end{corollary}
\begin{proof}
The following intertwining relation holds
\begin{equation}
\left.\swmap^{(N)}_{m,n}\!\left( H_{\gss}^{(1)}(\sigma)\right)\right|_{\gss=1/N}=
\left.\swmap^{(N)}_{m,n}\!\left(P^{\cA_N(m,n)}\circ H_{\gss}^{(1)}(\sigma)\right)\right|_{\gss=1/N}
= \fD_N^{(1)}\swmap^{(N)}_{m,n}(\sigma).
\label{eq:projected finite-N intertwining}
\end{equation}
Substituting \eqref{eq:projected finite-N intertwining} into the definition of the $U(N)$ inner product, Definition \ref{prop:UN inner}, gives \eqref{eq:finite-N physical matching}. 
\end{proof}

The factor $Q^{(N)}_{m,n}(x)$ in Theorem \ref{prop:finite-N factorisation} can be understood by introducing the orthogonal projector onto $\ker \swmap^{(N)}_{m,n}$
\begin{equation}
    [1 - P^{\cA_N(m,n)}]: \cA(m,n) \rightarrow \ker \swmap^{(N)}_{m,n}.
\end{equation}
and the composite operator
\begin{equation}
    [1 - P^{\cA_N(m,n)}] \circ  H^{(1)}_\gss: \ker \swmap^{(N)}_{m,n} \rightarrow
    \ker \swmap^{(N)}_{m,n}.
\end{equation}
By Lemma \ref{cor:finite-N factorisation} we have
\begin{equation}
    \det_{\cA(m,n)}(x -  H^{(1)}_\gss) \Big \vert_{\gss = 1/N} = 
    \det_{\cA_N(m,n)}(x - P^{\cA_N(m,n)} \circ  H^{(1)}_\gss)  \det_{\ker \swmap^{(N)}_{m,n}}(x - [1-P^{\cA_N(m,n)}] \circ  H^{(1)}_\gss) \Big \vert_{\gss = 1/N}.
\end{equation}
and
\begin{equation}
\label{eq: char poly equality}
    \begin{aligned}
        \det_{\cH_N(m,n)}(x -  \fD^{(1)}_N) &= \det_{\cA_N(m,n)}(x - P^{\cA_N(m,n)} \circ  H^{(1)}_\gss) \Big \vert_{\gss = 1/N}, \\
        Q^{(N)}_{m,n}(x) &= \det_{\ker \swmap^{(N)}_{m,n}}(x - [1-P^{\cA_N(m,n)}] \circ  H^{(1)}_\gss) \Big \vert_{\gss = 1/N}.
    \end{aligned}
\end{equation}

\subsection{Examples of characteristic polynomials}
\label{subsec:ex charpoly}
In this subsection we will show explicit examples of finite-$N$ factorisation of the one-loop permutation dilatation operator.
We compare the two characteristic polynomials appearing in Theorem \ref{prop:finite-N factorisation}.
Because the one-loop dilatation operator exhibits $U(2)$ symmetry rotating $Z$ into $W$, the spectrum contains many degenerate/redundant eigenvalues.
To reduce the amount of redundant information presented, we have restricted ourselves to the subspaces of highest-weight states under the $U(2)$ symmetry.
The conclusion of Theorem \ref{prop:finite-N factorisation} holds equally well within this subspace because the dilatation operator commutes with the quadratic Casimir $\hat C_2$ of $U(2)$.
The projector to $\cA_N(m,n)$ also commutes with the adjoint action of the centre of $\mathbb{C}[S_{m+n}]$.

The examples are computed from the one-parameter family obtained by restricting
$P^{\cA_N(m,n)}\circ H_{\gss}^{(1)}$ to $\cA_N(m,n)$. First, compute the planar and non-planar terms of the one-loop permutation dilatation operator in the orbit basis, using the convention
\begin{equation}
     H^{(1)}_{\gss}(\orbitP{m,n} (g_k)) = \sum_{k'=1}^{\dim \cA(m,n)} ( H^{(1)}_{\mathrm{p}})_{k'k} \orbitP{m,n}(g_{k'}) 
    + \gss \sum_{k'=1}^{\dim \cA(m,n)} ( H^{(1)}_{\mathrm{np}})_{k'k}  \orbitP{m,n}(g_{k'}).
\end{equation}
The matrix $ H^{(1)}_{\mathrm{p}}$ is the result of applying the planar part of \eqref{eq:D1-orbit-basis} and $ H^{(1)}_{\mathrm{np}}$ is the result of applying the non-planar part.
Next choose a SEP basis of the HWS subspace, with change-of-basis coefficients
\begin{equation}
\cQ^{R,\mathrm{HWS}}_a
=
\sum_{k=1}^{\dim \cA(m,n)}
\cQ^{R,\mathrm{HWS}}_{a,k}\orbitP{m,n}(g_k).
\end{equation}
Here $R\vdash L$ and $a=1,\dots,C(R,R,[m,n])$, as follows from \eqref{eq: dimcA HWS}. In this basis, the planar/non-planar decomposition of each matrix element is
\begin{equation}
\begin{aligned}
\sum_{S \vdash L } \sum_{b=1}^{C(S,S,[m,n])} ( H^{(1)}_{\mathrm{p}})^{SR}_{ba} \, \cQ^{S,\mathrm{HWS}}_{b,k} 
&= \sum_{k'=1}^{\dim \cA(m,n)} ( H^{(1)}_{\mathrm{p}})_{kk'} \cQ^{R,\mathrm{HWS}}_{a,k'} 
\\
\sum_{S \vdash L } \sum_{b=1}^{C(S,S,[m,n])} ( H^{(1)}_{\mathrm{np}})^{SR}_{ba} \, \cQ^{S,\mathrm{HWS}}_{b,k} 
&= \sum_{k'=1}^{\dim \cA(m,n)} ( H^{(1)}_{\mathrm{np}})_{kk'} \cQ^{R,\mathrm{HWS}}_{a,k'} \, .
\end{aligned}
\end{equation}
The full one-loop permutation dilatation operator in a SEP basis is given by the sum
\begin{equation}
\Big( P^{\cA_N(m,n)} \circ  H_{\gss}^{(1)} \Big)_{ba}^{SR}
= ( H^{(1)}_{\gss})^{SR}_{ba} 
= ( H^{(1)}_{\mathrm{p}})^{SR}_{ba} +\gss ( H^{(1)}_{\mathrm{np}})^{SR}_{ba}
\end{equation}
where we may omit $P^{\cA_N(m,n)}$ by choosing the finite-$N$ SEP label $S$ satisfying $\ell(S)\leq N$. 
Below we use the notation $\det_{\cA_N^{\mathrm{HWS}}(m,n)}(x-H_{\gss}^{(1)})$ without $P^{\cA_N(m,n)}$, assuming that both row and column labels are taken from a SEP basis of $\cA_N^{\mathrm{HWS}}(m,n)$.

At the physical value $\gss=N^{-1}$, the refinement of equation \eqref{eq: char poly equality} gives the equality of characteristic polynomials
\begin{equation}
\det_{\cH_N^{\mathrm{HWS}}(m,n)}(x-\fD_N^{(1)})
=
\det_{\cA_N^{\mathrm{HWS}}(m,n)}
\!\left(x- H_{\gss}^{(1)}\right)
\Big|_{\gss=1/N}.
\label{eq:HWS physical charpoly}
\end{equation}
We compare \eqref{eq:HWS physical charpoly} with the stable-space polynomial
$\det_{\cA^{\mathrm{HWS}}(m,n)}(x- H_{\gss}^{(1)})$, where no cutoff is imposed on the Young diagrams. Theorem \ref{prop:finite-N factorisation}, restricted to the HWS subspace, states that the former divides the latter after setting $\gss=N^{-1}$.

\subsubsection{$\cA(2,2)$}
The characteristic polynomial of the one-loop permutation dilatation operator in the stable regime of the HWS of $\Lambda = [2,2]$ is given by
\begin{equation}
    N\geq 4: \quad \det_{\cA^{\mathrm{HWS}}(2,2)}(x-  H^{(1)}_{\gss}) = x^2(x-6)
\end{equation}
for all $\gss$.
The characteristic polynomials in the unstable regimes $N=2,3$, computed through the actions on $\cA_N^{\mathrm{HWS}}(2,2)$, are equal to
\begin{equation}
    \begin{aligned}
        &N=2: \quad \det_{\mathcal{H}_2^{\mathrm{HWS}}(2,2)}(x-  \fD^{(1)}_{N=2})= x(x-6) \\ 
        &N=3: \quad \det_{\mathcal{H}_3^{\mathrm{HWS}}(2,2)}(x-  \fD^{(1)}_{N=3}) = x^2(x-6)
    \end{aligned}
\end{equation}
Clearly, the unstable polynomials divide the stable polynomial.

\subsubsection{$\cA(3,2)$}
The one-loop characteristic polynomial in the stable regime of the highest-weight space of $\Lambda=[3,2]$ is given by
\begin{equation}
    N\geq 5: \quad \det_{\cA^{\mathrm{HWS}}(3,2)}(x-  H^{(1)}_{\gss}) = x^4(x-4)(x-6)
\end{equation}
for all $\gss$.
In the unstable regimes $N=2,3,4$ they are equal to 
\begin{equation}
    \begin{aligned}
        &N=2:\quad \det_{\cH_2^{\mathrm{HWS}}(3,2)}(x-  \fD^{(1)}_{N=2})= x(x-6) \\
        &N=3:\quad \det_{\cH_3^{\mathrm{HWS}}(3,2)}(x-  \fD^{(1)}_{N=3})= x^3(x-4)(x-6) \\
        &N=4:\quad \det_{\cH_4^{\mathrm{HWS}}(3,2)}(x-  \fD^{(1)}_{N=4}) = x^4(x-4)(x-6) \\
    \end{aligned}
\end{equation}
Again, the unstable polynomials divide the stable polynomial, as expected.

\subsubsection{$\cA(4,2)$}
The HWS of $\Lambda=[4,2]$ is the first case where the spectrum in the stable regime depends on $\gss$.
The polynomial is given by
\begin{equation}
\begin{gathered}
N\geq 6: \quad \det_{\cA^{\mathrm{HWS}}(4,2)}(x-  H^{(1)}_{\gss}) = x^{10}(x-4)(x-6) q_\gss (x), \\
q_\gss (x) \equiv x^3-16x^2+(80-40\xi^2)x-120+80\xi^2.
\end{gathered}
\label{eq:charpoly A42}
\end{equation}
In the unstable regimes $N=2,3,4,5$ they are equal to
\begin{equation}
    \begin{aligned}
        &N=2:\quad \det_{\cH_2^{\mathrm{HWS}}(4,2)}(x-  \fD^{(1)}_{N=2})=x^2(x-6)(x-10)\\
        &N=3:\quad \det_{\cH_3^{\mathrm{HWS}}(4,2)}(x-  \fD^{(1)}_{N=3})=x^6(x-4)(x-6)(x^2-\frac{38}{3} x +\frac{100}{3}) \\
        &N=4:\quad \det_{\cH_4^{\mathrm{HWS}}(4,2)}(x-  \fD^{(1)}_{N=4})=x^9(x-4)(x-6)(x^3 - 16x^2 + \frac{155}{2}x - 115)\\
        &N=5:\quad \det_{\cH_5^{\mathrm{HWS}}(4,2)}(x-  \fD^{(1)}_{N=5})=x^{10}(x-4)(x-6)(x^3 - 16x^2 + \frac{392}{5}x - \frac{584}{5}) 
    \end{aligned}
\end{equation}
which divide \eqref{eq:charpoly A42} when $\gss=\tfrac{1}{N}$ since
\begin{equation}
    \begin{aligned}
        q_{\gss = 1/2} (x) &= (x^3 - 16x^2 + \frac{320-40}{4}x + \frac{80-480}{4}) = (x-10)(x^2-6x+10), \\
        q_{\gss = 1/3} (x) &=
        (x^3 - 16x^2 + \frac{720-40}{9}x + \frac{80-1080}{9})  =
        (x-\frac{10}{3})(x^2-\frac{38}{3}x+\frac{100}{3}), \\
        q_{\gss = 1/4} (x) &=
        (x^3 - 16x^2 + \frac{1280-40}{16}x + \frac{80-1920}{16})  =
        (x^3-16x^2 + \frac{155}{2}x- 115), \\
        q_{\gss = 1/5} (x) &=
        (x^3 - 16x^2 + \frac{2000-40}{25}x + \frac{80-3000}{25})  =
        (x^3-16x^2 + \frac{392}{5}x- \frac{584}{5}) \,.
    \end{aligned}
\end{equation}

\section{$SU(N)$ factorisation of $U(N)$ state spaces}\label{sec:SU(N)}

The two complex scalars of $\cN=4$ SYM with $U(N)$ gauge group can be re-organised into traceless matrices $Z',W'$ and a trace part $z,w$.
We introduce lower-case symbols for single-letter traces
\begin{equation}
z \equiv \frac{\tr(Z)}{N} = \frac{1}{N} \sum_{i=1}^N Z_{ii} \,, \qquad
w \equiv \frac{\tr(W)}{N} = \frac{1}{N} \sum_{i=1}^N W_{ii} \,.
\label{def:sl-trace}
\end{equation}
The variables $Z_{ij},W_{ij}$ and $Z'_{ij},W'_{ij}$ are related by
\begin{equation}
    \begin{aligned}
       &Z_{ij}= Z_{ij}'+\delta_{ij} \, z,\\
       &W_{ij}= W_{ij}'+\delta_{ij} \, w,
    \end{aligned}
    \quad \quad 
    \text{or equivalently}
    \quad \quad 
    \begin{aligned}
       &Z = Z' + I z\\
       &W = W' + I w
    \end{aligned}
\label{eq:traceless-variable}
\end{equation}
where $I$ is the identity matrix and $\tr(Z')=\tr(W')=0$. 
The multi-trace operators in the $U(N)$ gauge theory can be expressed in terms of $Z',W'$ and $z,w$.

We focus on the subspace of $SU(N)$-invariant polynomials in $\mathbb{C}[Z',W']$ that do not contain the single-letter traces $z,w$. 
We denote this subspace by
\begin{equation}
    \mathcal{H}'_N = \{ \poly (Z',W') \in \mathbb{C}[Z',W'] \, \vert \, \poly (U Z'U^\dagger, U W'U^\dagger ) = \poly (Z',W')  \quad \forall U \in SU(N)\}.
\label{def:cH'mnN}
\end{equation}
It is bi-graded:
\begin{equation}
\mathcal{H}'_N= \bigoplus_{m,n \geq 0} \mathcal{H}'_N(m,n), \qquad \cH'_N(1,0) = \cH'_N(0,1) = 0.
\end{equation}
There are no invariants of degree $(1,0)$ or $(0,1)$ because $\tr(Z') = \tr(W') = 0$.
For example, $\mathcal{H}_N'(2,1)$ is at most one-dimensional, spanned by
\begin{equation}
        \tr(Z'Z'W') \, .
\end{equation}
This should be compared to Example \ref{ex: H_N(2,1)}.

In this section we will show how the $U(N)$ Hilbert space $\mathcal{H}_N$ factorises into the tensor product of the polynomial ring $\mathbb{C}[z,w]$ and the Hilbert space $\mathcal{H}'_N$ of invariants constructed from traceless matrices:
\begin{equation}
\mathcal{H}_N \cong \mathcal{H}'_N \otimes \mathbb{C}[z,w].
\label{eq:hilb factorisation}
\end{equation}
We also show that this factorisation extends to the one-loop level. 
The one-loop dilatation operator commutes with multiplication by single-letter traces, and both $\mathcal{H}_N$ and $\mathcal{H}'_N$ are closed under its action.
We use the symbol $\fD_N^{(1)}$ for the operator on $\mathcal{H}_N$ and $\fD_N^{\prime(1)}$ for the operator on $\mathcal{H}'_N$\,.
For a fixed eigenvalue $e$, we define the corresponding eigenspaces in the $U(N)$ and $SU(N)$ theories by
\begin{equation}
\mathcal{H}_N(e) = \{ \poly \in \mathcal{H}_N \, \vert \,  \fD^{(1)}_N \poly = e \poly \}, \qquad
\mathcal{H}'_N(e) = \{ \poly \in \mathcal{H}'_N \, \vert \,  \fD^{\prime(1)}_N \poly = e \poly \}\,.
\label{def:eigen fD1}
\end{equation}
The factorisation in \eqref{eq:hilb factorisation} also holds for each one-loop eigenspace:
\begin{equation}
\mathcal{H}_N(e) \cong \mathcal{H}'_N(e) \otimes \mathbb{C}[z,w].
\label{eq:eigenspace factorisation}
\end{equation}

An important consequence of the factorisation is the relation between the characteristic polynomial of $U(N)$ and $SU(N)$ theories in the highest-weight sector of the global $U(2)$ symmetry:
\begin{remainproposition}[One-loop factorisation formula]{cor:suN vs uN HWS fixed e}
\begin{equation}
\det_{\mathcal{H}_N'^{\mathrm{HWS}}(m,n)}(x -  \fD^{\prime(1)}_N) 
=
\frac{ \det_{\mathcal{H}_N^\mathrm{HWS}(m,n)}(x- \fD^{(1)}_N)
\det_{\mathcal{H}_N^\mathrm{HWS}(m-1,n-1)}(x- \fD^{(1)}_N) }
{ [\det_{\mathcal{H}_N^\mathrm{HWS}(m-1,n)}(x- \fD^{(1)}_N)]^{(1-\delta_{m,n})}
\det_{\mathcal{H}_N^\mathrm{HWS}(m,n-1)}(x- \fD^{(1)}_N) } 
\end{equation}
which is valid for any $N$. 
\end{remainproposition}
\noindent
Proposition~\ref{cor:suN vs uN HWS fixed e} will be used to test our conjecture in Section \ref{sec:derangements}.

\subsection{Hilbert series factorisation}

At the level of counting, the factorisation of the Hilbert series is readily derived using Molien-Weyl's theorem.
For a basic review of the theorem for physicists see \cite{OConnor:2023ptf} and for a textbook treatment see \cite{cox1998using}.

The Hilbert series for $\mathcal{H}_N$ is the generating function
\begin{equation}
\cZ_{\text{free}}^{u(N)}(x,y)
=
\sum_{m,n=0}^\infty x^m y^n \, \dim \mathcal{H}_N({m,n}) \, .
\label{eq:uN Hilbert series}
\end{equation}
The Molien-Weyl theorem says that this is equivalently computed from a group integral with Haar measure on $U(N)$
\begin{equation}
    \cZ_{\text{free}}^{u(N)}(x,y) = \int_{U(N)} \frac{\mathrm{d}u }{\det(1 - x D^{Z}(u))\det(1 - y D^{W}(u))}.
\end{equation}
Here $D^Z(u)$ and $D^W(u)$ are two copies of the adjoint representation of $U(N)$, acting by $X\mapsto uXu^\dagger$ for $X = Z, W$, respectively.
Readers familiar with this integrand may recognise this as the following plethystic exponential:
\begin{equation}
    \cZ_{\text{free}}^{u(N)}(x,y) = \int_{U(N)} \mathrm{d}u \; \mathrm{exp}
    \qty( \sum_{k=1}^\infty \frac{x^k + y^k}{k} \, \Bigl[ {\rm tr} (u^k ) \, {\rm tr} (u^\dagger{}^k) \Bigl] ) 
\end{equation}
where $\tr$ is the trace in the fundamental representation.
As we have explained, these representations decompose into the space of traceless matrices and their trace (which transforms trivially for all $U \in U(N)$)
\begin{equation}
    D^{Z} \cong D^{Z'} \oplus D^{z}, \quad
    D^{W} \cong D^{W'} \oplus D^{w} \, .
\end{equation}
It follows that
\begin{equation}
    \begin{aligned}
        \det(1 - x D^{Z}(u)) &= \det(1 - x D^{Z'}(u))\det(1 - x D^{z}(u)) = \det(1 - x D^{Z'}(u)) (1-x) \\
        \det(1 - y D^{W}(u)) &= \det(1 - y D^{W'}(u))\det(1 - y D^{w}(u)) = \det(1 - y D^{W'}(u)) (1-y) 
    \end{aligned}
\end{equation}
and therefore
\begin{equation}
    \cZ_{\text{free}}^{u(N)}(x,y) =\frac{1}{(1-x)(1-y)} \int_{U(N)} \mathrm{d}u \frac{1}{\det(1 - x D^{Z'}(u))\det(1 - y D^{W'}(u))}.
\end{equation}
The remaining matrix integral is, by Molien-Weyl's theorem, the Hilbert series for $\mathcal{H}'_N$
\begin{equation}
\cZ_{\text{free}}^{su(N)}(x,y)
=
\sum_{m,n=0}^\infty x^m y^n \, \dim \mathcal{H}_N'({m,n}) 
=
\int_{SU(N)} \mathrm{d}u \frac{1}{\det(1 - x D^{Z'}(u))\det(1 - y D^{W'}(u))} .
\label{eq:suN Hilbert series}
\end{equation}
Here we integrated the $U(1)$ part of the $U(N)$ integral, which acts trivially on $Z', W'$ in the adjoint representation of $SU(N)$.
In other words,
\begin{equation}
    \cZ_{\text{free}}^{u(N)}(x,y) = \frac{\cZ_{\text{free}}^{su(N)}(x,y)}{(1-x)(1-y)} \,.
\label{eq: ZUN to ZSUN}
\end{equation}
Since $\mathbb{C}[z,w]$ has Hilbert series $(1-x)^{-1}(1-y)^{-1}$, this corresponds to \eqref{eq:hilb factorisation} at the level of Hilbert series.
As an immediate consequence, the fixed degree counting of traceless invariants becomes
\begin{equation}
    \label{eq:H' dim diff formula}
    \dim \mathcal{H}'_N({m,n})  = 
    \dim \mathcal{H}_N(m,n) -
    \dim \mathcal{H}_N(m-1,n) -
    \dim \mathcal{H}_N(m,n-1) +
    \dim \mathcal{H}_N(m-1,n-1)  \, .
\end{equation}
Some examples of these dimensions are given in Table \ref{tbl:intsec dims}.

Molien-Weyl's theorem can also be used to prove a similar factorisation of the generating function for (counting of) quarter-BPS operators.
It is expected that $\cZ_{\mathrm{BPS}}^{u(N)}(x,y)$ is given by the Hilbert series for the ring of $U(N)$ invariant polynomials in a pair of commuting matrices $\comm{Z}{W}=0$ \cite{Dolan:2007rq}.
A generic pair of commuting matrices $(Z,W)$ is simultaneously diagonalisable by a similarity transformation,
\begin{equation}
    (Z,W) \sim (\diag(z_1,\dots,z_N),  \diag(w_1,\dots,w_N))
\end{equation}
and since $U(N)$ invariant polynomials are also invariant under similarity transformations, an invariant polynomial on the space of commuting matrices is determined by its restriction to diagonal pairs. 
The remaining invariance is the subgroup $S_N \subset U(N)$ of permutation matrices simultaneously permuting the eigenvalues $(z_i,w_i) \mapsto (z_{s(i)}, w_{s(i)})$ for $s \in S_N$.
This is the direct sum of two defining/natural representations $V_N$ of $S_N$.
The natural representation decomposes into the trivial representation spanned by the identity matrix
\begin{equation}
    V_0 = \{ \lambda I \, \vert \, \lambda \in \mathbb{C} \}
\end{equation}
and the diagonal matrices/vectors
\begin{equation}
    V_H = \{ \diag(x_1,\dots,x_N) \, \vert \, \sum_{i=1}^N x_i = 0 \}.
\end{equation}
In terms of representation matrices
\begin{equation}
    D^{V_N} \cong D^{V_H} \oplus D^{V_0} \,.
\end{equation}
The Molien-Weyl formula in this case gives
\begin{align}
\cZ_{\mathrm{BPS}}^{u(N)}(x,y)
&= \frac{1}{N!}\sum_{g \in S_N} \frac{1}{\det(1-xD^{V_N}(g))\det(1-yD^{V_N}(g))} 
\notag \\
&= \frac{1}{N!}\sum_{g \in S_N} \frac{1}{\det(1-xD^{V_H}(g))\det(1-xD^{V_0}(g))\det(1-yD^{V_H}(g))\det(1-yD^{V_0}(g))} 
\notag \\
&= \frac{1}{(1-x)(1-y)}\cZ_{\mathrm{BPS}}^{su(N)}(x,y)
\end{align}
where
\begin{equation}
    \cZ_{\mathrm{BPS}}^{su(N)}(x,y) = \frac{1}{N!}\sum_{g \in S_N} \frac{1}{\det(1-xD^{V_H}(g))\det(1-yD^{V_H}(g))} 
\end{equation}
is the Hilbert series for invariants on pairs of commuting traceless matrices.
This is the counting equivalent of \eqref{eq:eigenspace factorisation} for the special case of $e = 0$.

\subsection{Hilbert space factorisation}
Given the intuition coming from the counting formulas, we will show that the same factorisation holds at the level of spaces of states.
\begin{mainproposition}[Hilbert space factorisation]
\label{prop:hilb factorisation}
Let $\mathcal{H}_N$ be the space of invariants in traceful matrices and $\mathcal{H}'_N$ the space of invariants in traceless matrices.
The two spaces are related by a factorisation
\begin{equation}
    \mathcal{H}_N \cong \mathcal{H}'_N \otimes \mathbb{C}[z,w],
\label{eq:poly-factorisation}
\end{equation}
where $\mathbb{C}[z,w]$ is a polynomial ring in two variables.
We will see that this ring is generated by single-letter traces.
\end{mainproposition}

We prove the proposition by giving two maps and showing that they are inverse.
Since elements $\polt \in \mathcal{H}'_N \otimes \mathbb{C}[z,w]$ are functions on the space of pairs of traceless matrices $Z',W'$ and two complex values $z,w$, we write $\polt(Z',W',z,w)$ for their value on points in this space.
\begin{definition}[Substitution maps]\label{def:substitution maps}
Using the identity \eqref{eq:traceless-variable}, we introduce two substitution maps.
The map $\Psi : \mathcal{H}_N \, \to \, \mathcal{H}'_N \otimes \mathbb{C}[z,w]$ is
\begin{equation}
{\Psi}[ \poly ] (Z',W',z,w) \equiv \poly (Z' + I z, W' + I w) \ \in \ \mathcal{H}'_N \otimes \mathbb{C}[z,w] 
\label{def:Psi map}
\end{equation}
and the map $\Phi : \mathcal{H}'_N \otimes \mathbb{C}[z,w] \, \to \, \mathcal{H}_N$ is
\begin{equation}
\begin{aligned}
{\Phi}[\poly' \otimes \polt ] (Z, W) 
&\equiv (\poly' \otimes \polt) \Big(Z - I \, \frac{\tr(Z)}{N}, W - I \, \frac{\tr(W)}{N} , \frac{\tr(Z)}{N}, \frac{\tr(W)}{N} \Big) 
\\[1mm]
&= 
\poly ' \Big(Z - I \, \frac{\tr(Z)}{N}, W - I \, \frac{\tr(W)}{N} \Big) \, 
\polt \Big( \frac{\tr(Z)}{N}, \frac{\tr(W)}{N} \Big) 
\ \in \ \mathcal{H}_N
\end{aligned}
\label{def:Phi map}
\end{equation}
together with linear extension.
\end{definition}
\begin{example}
For instance, the invariant $\tr(ZW)$ satisfies
\begin{equation}
\Psi[\tr(ZW)] = \tr(Z'W') + Nzw , \qquad
\Phi[\tr(Z'W')] = \tr(ZW) -  N z w.
\end{equation}
\end{example}

\begin{proof}[Proof of Proposition~\ref{prop:hilb factorisation}]
To show that $\Phi$ is a left inverse of $\Psi$ we compute
\begin{equation}
\begin{aligned}
\Phi[\Psi[\poly]](Z,W) &= \Psi[\poly] \Big(Z - I \, \frac{\tr(Z)}{N}, W - I \, \frac{\tr(W)}{N} , 
\frac{\tr(Z)}{N} , \frac{\tr(W)}{N} \Big)
\\[1mm]
&= \poly \Big( Z - I \, \frac{\tr(Z)}{N} + I \, \frac{\tr(Z)}{N} ,
W - I \, \frac{\tr(W)}{N} + I \, \frac{\tr(W)}{N} \Big) = \poly(Z,W).
\end{aligned}
\end{equation}

To show that $\Phi$ is a right inverse of $\Psi$ we compute
\begin{equation}
\begin{aligned}
&\Psi[\Phi[\poly' \otimes \polt]](Z',W',z,w) \\
&= \Phi[\poly' \otimes \polt](Z'+Iz,W'+Iw)  \\
&= (\poly' \otimes \polt) \Big(Z'+Iz-I \, \frac{\tr (Z'+Iz)}{N}, W'+Iw-I \, \frac{\tr(W'+Iw)}{N} , 
\frac{\tr(Z'+Iz)}{N}, \frac{\tr(W'+Iw)}{N} \Big) \\
&= (\poly' \otimes \polt)(Z',W',z,w)
\end{aligned}
\end{equation}
where we used $\tr (X') = 0, \ \tr(I) = N$ and
\begin{equation}
\frac{\tr(X' + Ix)}{N} = \frac{\tr(X') + x \, \tr (I)}{N} = x, \qquad (X,x) \in \{  (Z,z),  (W,w) \}
\end{equation}
in the last equality.
\end{proof}

\noindent
Notably, $\mathcal{H}'_N$ corresponds to the subspace of $\mathcal{H}_N$ where the degree of $z,w$ is zero.
This will be a guiding fact when we construct finite-$N$ bases for invariants of traceless matrices.

\subsection{Eigenspace factorisation}\label{sec:eigenspace factorisation}

We first define derivative operators in traceless directions.
\begin{definition}[Projected traceless derivatives]
For $(X,x) \in \{  (Z,z),  (W,w) \}$, we introduce projected traceless derivatives by
\begin{equation}
\projD^{X'}_{ji} \equiv \pdv{}{X_{ij}}-\frac{\delta_{ij}}{N}\sum_{k=1}^N\pdv{}{X_{kk}} 
\label{def:proj-tr-der}
\end{equation}
which satisfy
\begin{equation}
\projD^{X'}_{ji} ( x ) = 0, \qquad
\projD^{X'}_{ji} ( X'_{kl} ) =\delta_{ik}\delta_{jl}-\frac{1}{N}\delta_{ij}\delta_{kl} \,.
\label{eq:proj-tr-der}
\end{equation}
\end{definition}

The factorisation in Proposition~\ref{prop:hilb factorisation} can be refined to eigenspaces of the dilatation operator.
\begin{maintheorem}[Eigenspace factorisation]\label{prop:eigenspace factorisation}
The two eigenspaces of the one-loop dilatation operator in \eqref{def:eigen fD1} are related by the factorisation
\begin{equation}
    \mathcal{H}_N(e) \cong \mathcal{H}'_N(e) \otimes \mathbb{C}[z,w]
\end{equation}
where $z,w$ are single-letter traces as in \eqref{def:sl-trace}.
\end{maintheorem}

\begin{proof}
In Lemma~\ref{lemma:oneloop factorisation} (see Appendix~\ref{app:trace factorisation}), we show that the one-loop dilatation operator can be rewritten using the projected traceless derivatives defined in \eqref{def:proj-tr-der}:
\begin{equation}
\fD_N^{\prime(1)} = -\frac{1}{N}:\tr\qty(\comm{Z'}{W'}\comm{\projD^{Z'}}{\projD^{W'}}): \,.
\label{eq: SUN deriv dil}
\end{equation}

As a consequence (Lemma \ref{lemma: weyl dil commute}), the one-loop dilatation operator commutes with $z,w, \pdv{z}, \pdv{w}$, where
\begin{equation}
    \pdv{x} = \sum_{i=1}^N \pdv{X_{ii}}, \quad (X,x) \in \{(Z,z),(W,w)\}.
\end{equation}
By Proposition~\ref{prop:hilb factorisation}, every invariant has a unique finite expansion
\begin{equation}
\poly(Z,W)=\sum_{a,b\geq0}z^aw^b\poly_{a,b}(Z',W'),
\end{equation}
and Lemma \ref{lemma: coefficient eigenspace factorisation} shows that $\fD_N^{(1)}\poly=e\poly$ holds if and only if
$\fD_N^{\prime(1)}\poly_{a,b}=e\poly_{a,b}$ for every $a,b \geq 0$.
\end{proof}

\begin{example}
The traceless invariant $\tr(Z'^2)$ is a half-BPS operator with $\fD_N^{\prime(1)}$ eigenvalue zero. The following operators
\begin{equation}
z^a w^b\tr(Z'^2) = z^a w^b \Big( \tr (ZZ) - N \, z^2 \Big), \qquad a,b\in\mathbb{Z}_{\geq 0}
\end{equation}
are $U(N)$ eigenstates with the same eigenvalue zero.
\end{example}

Theorem \ref{prop:eigenspace factorisation} has several important consequences.
The first consequence involves dimensions of eigenspaces.
Define the fixed degree eigenspaces
\begin{equation}
\begin{aligned}
\mathcal{H}_N(m,n,e) &= \{ \poly \in \mathcal{H}_N(m,n) \, \vert \,  \fD^{(1)}_N \poly = e \poly \} ,\\
\mathcal{H}'_N(m,n,e) &= \{ \poly \in \mathcal{H}_N'(m,n) \, \vert \,  \fD^{\prime(1)}_N \poly = e \poly \} \, .
\end{aligned}
\end{equation}
\begin{lemma}\label{lemma: dims eigenspaces relations}
The dimensions of fixed degree eigenspaces are related by
\begin{equation}
    \label{eq: m n e factorisation dim}
    \dim \mathcal{H}'_N(m,n,e) =
    \begin{aligned}[t]
    &\dim \mathcal{H}_N(m,n,e)
    +\dim \mathcal{H}_N(m-1,n-1,e) \\
    &-\dim \mathcal{H}_N(m-1,n,e)
    -\dim \mathcal{H}_N(m,n-1,e) \,.
    \end{aligned}
\end{equation}
\end{lemma}
\begin{proof}
Define the generating functions
\begin{equation}
    Z^{u(N)}_e (x,y) = \sum_{m,n \geq 0} x^m y^n \dim \mathcal{H}_N(m,n,e), \quad 
    Z^{su(N)}_e (x,y) = \sum_{m,n \geq 0} x^m y^n \dim \mathcal{H}'_N(m,n,e)\,.
\end{equation}
From Theorem \ref{prop:eigenspace factorisation} we have
\begin{equation}
    Z_e^{su(N)}(x,y) =  (1-x)(1-y)Z^{u(N)}_e(x,y) \,
\end{equation}
and extracting the coefficients of $x^{m}y^{n}$ gives the lemma.
\end{proof}

From the above lemma, we get a simple relationship between the characteristic polynomials in $SU(N)$ gauge theory and $U(N)$ gauge theory.
\begin{corollary}\label{cor:suN charpoly}
The $SU(N)$ characteristic polynomials are computable in terms of $U(N)$ characteristic polynomials:
\begin{equation}
    \begin{aligned}
    \det_{\mathcal{H}_N'(m,n)}(x -  \fD^{\prime(1)}_N) 
    &=
    \frac{
    \det_{\mathcal{H}_N(m,n)}(x- \fD^{(1)}_N)
    \det_{\mathcal{H}_N(m-1,n-1)}(x- \fD^{(1)}_N)
    }
    {
    \det_{\mathcal{H}_N(m-1,n)}(x- \fD^{(1)}_N)
    \det_{\mathcal{H}_N(m,n-1)}(x- \fD^{(1)}_N)
    }
    \end{aligned}
\end{equation}
where we adopt the convention
\begin{equation}
\cH_N(a,b)=0 \qquad {\rm and} \qquad
\det_{\mathcal H_N(a,b)} (x- \fD^{(1)}_N)=1 \qquad
{\rm if} \ a<0 \ {\rm or} \ b<0.
\end{equation}
\end{corollary}
\begin{proof}
Let $E$ be the union of the spectra of $\fD_N^{(1)}$ on the four spaces $\mathcal{H}_N(m,n)$, $\mathcal{H}_N(m-1,n-1)$, $\mathcal{H}_N(m-1,n)$ and $\mathcal{H}_N(m,n-1)$.
Lemma~\ref{lemma: dims eigenspaces relations} implies that $E$ contains the spectrum of $\fD_N^{\prime(1)}$ on $\mathcal{H}'_N(m,n)$.
The characteristic polynomial therefore factors over the real numbers as
\begin{equation}
    \det_{\mathcal{H}_N'(m,n)}(x -  \fD^{\prime(1)}_N) = \prod_{e \in E} (x-e)^{\dim \mathcal{H}'_N(m,n,e)} .
\end{equation}
From Lemma \ref{lemma: dims eigenspaces relations} we get
\begin{equation}
    \begin{aligned}
    \prod_{e \in E} (x-e)^{\dim \mathcal{H}'_N(m,n,e)}
    &=
    \prod_{e \in E} 
    \frac{
    (x-e)^{\dim \mathcal{H}_N(m,n,e)}
    (x-e)^{\dim \mathcal{H}_N(m-1,n-1,e)}
    }
    {
    (x-e)^{\dim \mathcal{H}_N(m-1,n,e)}
    (x-e)^{\dim \mathcal{H}_N(m,n-1,e)}
    } \\
    &=
    \frac{
    \det_{\mathcal{H}_N(m,n)}(x- \fD^{(1)}_N)
    \det_{\mathcal{H}_N(m-1,n-1)}(x- \fD^{(1)}_N)
    }
    {
    \det_{\mathcal{H}_N(m-1,n)}(x- \fD^{(1)}_N)
    \det_{\mathcal{H}_N(m,n-1)}(x- \fD^{(1)}_N) 
    } \,.
    \end{aligned}
\end{equation}
In particular, for any $e\in E$ absent from the spectrum of $\mathcal{H}_N'(m,n)$, Lemma~\ref{lemma: dims eigenspaces relations} ensures that the corresponding factors cancel between the numerator and denominator.
\end{proof}

The above results generalise to the highest-weight space of a fixed irreducible representation $\Lambda = [\lambda_1, \lambda_2]$ of $U(2)$.
Let $\hat C_2$ be the quadratic Casimir, we define the highest-weight eigenspaces of the one-loop dilatation operator
\begin{equation}
    \begin{aligned}
        \mathcal{H}_N^\mathrm{HWS}(m,n,e)      & = \{\poly(Z,W) \in \mathcal{H}_N(m,n,e) \, \vert \, \hat C_2 \poly(Z,W) = C_2([m,n])\poly(Z,W)\} \\
        \mathcal{H}_N^{'\mathrm{HWS}}(m,n,e)   & = \{\poly(Z',W') \in \mathcal{H}'_N(m,n,e) \, \vert \, \hat C_2 \poly(Z',W') = C_2([m,n])\poly(Z',W')  \}.
    \end{aligned}
\end{equation}
\begin{lemma}\label{prop:HWS dim comp}
The dimensions of highest-weight eigenspaces are related by
\begin{equation}
    \dim \mathcal{H}_N'^{\mathrm{HWS}}(m,n,e) = 
    \begin{aligned}[t]
        &\dim \mathcal{H}_N^{\mathrm{HWS}}(m,n,e) -
        (1-\delta_{m,n})\dim \mathcal{H}_N^{\mathrm{HWS}}(m-1,n,e) \\
        &-\dim \mathcal{H}_N^{\mathrm{HWS}}(m,n-1,e)+
        \dim \mathcal{H}_N^{\mathrm{HWS}}(m-1,n-1,e)
    \end{aligned}
\end{equation}
where the factor $1-\delta_{m,n}$ removes the $(m-1,m)$ term, which is not a valid partition.
\end{lemma}
\begin{proof}
Expanding the generating functions in terms of irreducible characters of $U(2)$ gives
\begin{equation}
    \sum_{L=0}^\infty \sum_{\substack{\Lambda \vdash L \\ \ell(\Lambda) \leq 2}} m^{su(N),e}_\Lambda \chi^\Lambda(x,y) = (1-x)(1-y)\sum_{L=0}^\infty \sum_{\substack{\Lambda \vdash L \\ \ell(\Lambda) \leq 2}} m^{u(N),e}_\Lambda \chi^\Lambda(x,y) .
    \label{eq: cor proof Z fixed e1}
\end{equation}
Here $m^{u(N),e}_\Lambda$ and $m^{su(N),e}_\Lambda$ are the multiplicities of the representation $\Lambda$ in $\mathcal{H}_N(e)$ and $\mathcal{H}'_N(e)$, respectively, and
\begin{equation}
\chi^{\Lambda } ( x , y ) = \frac{x^{\lambda_1+1} y^{\lambda_2} - x^{\lambda_2} y^{\lambda_1+1}}{x-y} \,.
\label{def:U2 char}
\end{equation}
Now use the fact that
\begin{equation}
    (1-x)(1-y) = 1 - \chi^{[1]}(x,y) + \chi^{[1,1]}(x,y) 
\end{equation}
and Pieri's formula which gives
\begin{equation}
    \begin{aligned}
    \chi^{[1]}(x,y)\chi^{[\lambda_1,\lambda_2]}(x,y) &= 
    \chi^{[\lambda_1+1,\lambda_2]}(x,y) + 
    \chi^{[\lambda_1,\lambda_2+1]}(x,y), \\
    \chi^{[1,1]}(x,y)\chi^{[\lambda_1,\lambda_2]}(x,y) &= 
    \chi^{[\lambda_1+1,\lambda_2+1]}(x,y)\,
    \end{aligned}
\end{equation}
where a character vanishes, by definition, if $\lambda_2 > \lambda_1$.
Substituting this into \eqref{eq: cor proof Z fixed e1} with $\Lambda=[\lambda_1,\lambda_2]$, and suppressing the $x,y$ dependence, gives
\begin{equation}
\sum_{L=0}^\infty \sum_{\substack{\Lambda \vdash L \\ \ell(\Lambda) \leq 2}} m^{su(N),e}_\Lambda \chi^\Lambda 
= \sum_{L=0}^\infty \sum_{\substack{\Lambda \vdash L \\ \ell(\Lambda) \leq 2}} 
m^{u(N),e}_{[\lambda_1,\lambda_2]} \, \Big( 
\chi^{[\lambda_1,\lambda_2]} - \chi^{[\lambda_1+1,\lambda_2]} - \chi^{[\lambda_1,\lambda_2+1]} + \chi^{[\lambda_1+1,\lambda_2+1]} 
\Big) .
\end{equation}
Here multiplicities with negative entries or entries that do not form a partition are understood to vanish.
Comparing the coefficients of $\chi^\Lambda$ gives
\begin{equation}
    m^{su(N),e}_{[\lambda_1,\lambda_2]} = 
    m^{u(N),e}_{[\lambda_1,\lambda_2]} -
    (1-\delta_{\lambda_1,\lambda_2})m^{u(N),e}_{[\lambda_1-1,\lambda_2]} -
    m^{u(N),e}_{[\lambda_1,\lambda_2-1]} +
    m^{u(N),e}_{[\lambda_1-1,\lambda_2-1]} \, .
\end{equation}
Since each irreducible representation associated with $\Lambda$ has a one-dimensional highest-weight subspace, the multiplicity of the character $\chi^\Lambda$ is equal to the dimension of the highest-weight space.
In other words, we can identify
\begin{equation}
    m^{su(N),e}_{[m,n]} = 
    \dim \mathcal{H}_N'^{\mathrm{HWS}}(m,n,e), \quad 
    m^{u(N),e}_{[m,n]} =
    \dim \mathcal{H}_N^{\mathrm{HWS}}(m,n,e).
\end{equation}
\end{proof}
Using the same proof strategy as in Corollary \ref{cor:suN charpoly}, the above lemma implies the following.
\begin{mainproposition}[One-loop factorisation formula]\label{cor:suN vs uN HWS fixed e}
The characteristic polynomials in the highest-weight subspaces factorise as
\begin{equation}
    \begin{aligned}
    \det_{\mathcal{H}_N'^{\mathrm{HWS}}(m,n)}(x -  \fD^{\prime(1)}_N) 
    =
    \frac{
    \det_{\mathcal{H}_N^\mathrm{HWS}(m,n)}(x- \fD^{(1)}_N)
    \det_{\mathcal{H}_N^\mathrm{HWS}(m-1,n-1)}(x- \fD^{(1)}_N)
    }
    {
        [\det_{\mathcal{H}_N^\mathrm{HWS}(m-1,n)}(x- \fD^{(1)}_N)]^{(1-\delta_{m,n})}
    \det_{\mathcal{H}_N^\mathrm{HWS}(m,n-1)}(x- \fD^{(1)}_N)
    } 
    \end{aligned}
\end{equation}
where the Kronecker delta is used to set first factor in the denominator equal to $1$ for $m=n$, where $(m-1,m)$ is not a valid integer partition.
\end{mainproposition}

\section{Permutation description of $SU(N)$ theory}\label{sec:derangements}

We now address the central problem of describing $SU(N)$ gauge-invariant operators and their one-loop dilatation within the PCA framework at any $N$.
In Section \ref{sec:SU(N)} we determined the $SU(N)$ characteristic polynomial from $U(N)$ data; here we seek a direct permutation description of the states and the one-loop dilatation operator.
We will see that traceless invariants in the stable regime are described by a derangement subspace $\cD(m,n) \subseteq \cA(m,n)$, namely the subspace spanned by permutations without fixed points.

Define the derangement projection $\pi_D$ on $\mathbb{C}[S_{m+n}]$ by linear extension of
\begin{equation}
\pi_D(g)=\begin{cases}
g, & (g\in D_{m+n}) \\
0, & (g\notin D_{m+n}).
\end{cases}
\end{equation}
This projection does not commute in general with the left central actions $m^{\mathfrak{L}}[\centz]$ that determine the finite-$N$ cutoff $\ell(R)\leq N$ in the $U(N)$ theory,
\begin{equation}
    \Big[ m^{\mathfrak{L}}[\centz], \pi_D \Big] \neq 0 .
\end{equation}
We propose a description of traceless invariants as the intersection of the derangement and finite-$N$ subspaces,
\begin{equation}
\cD_N(m,n) = \cA_N(m,n) \cap \cD(m,n) .
\label{DN proposal}
\end{equation}
We also propose that the $SU(N)$ one-loop dilatation operator in \eqref{eq: SUN deriv dil} is determined by the composition of an orthogonal projector $P^{\cD_N(m,n)}$ from $\cA(m,n)$ to $\cD_N(m,n)$ with the one-loop permutation dilatation operator $H_{\gss}^{(1)}$ defined in Definition \ref{def: 1LPDO}.
This proposal constitutes a PCA framework for the $SU(N)$ dilatation operator problem at finite $N$.

In the first subsection, we introduce derangement descriptions for the space of $SU(N)$-invariant polynomials $\cH'_N(m,n)$.
In the second subsection, we conjecture that the intersection \eqref{DN proposal} is mapped isomorphically to $\cH'_N(m,n)$ in the unstable regime.
We give numerical checks of the conjecture by comparing the dimensions of the Hilbert space, the dimensions of the intersection, and the rank of the Gram matrix. 
In the third subsection, we explain how to define the one-loop dilatation operator on the intersection.
In the final subsection, we present numerical checks of this proposal.
We check that the characteristic polynomials of the one-loop dilatation operator in the $SU(N)$ theory at finite $N$ agree with the factorisation formula in Proposition~\ref{cor:suN vs uN HWS fixed e}.

\subsection{Derangement orbit subspace}

Permutations also form a useful description of the space of invariants of traceless matrices. 
For $g \in S_{m+n}$, consider the map
\begin{equation}
    g \mapsto 
    \sum_{i_1,\dots,i_L=1}^N
    \prod_{p=1}^{m} Z'_{i_p i_{g(p)}}
    \prod_{q=m+1}^L W'_{i_q i_{g(q)}} \in \mathcal{H}'_N(m,n).
\end{equation}
If $g$ has fixed points, this map produces $\tr(Z')$ or $\tr(W')$ which vanishes.
Clearly, we should restrict ourselves to permutations without fixed points.
\begin{definition}[Derangements]
A permutation $g \in S_{m+n}$ is called a derangement if it has no fixed points. The subset
\begin{equation}
D_{m+n} = \{ g \in S_{m+n} \, \vert \, g(i) \neq i \quad \forall i \in \{1,\dots,m+n\} \}
\end{equation}
is called the set of derangements in $S_{m+n}$\,. 
Note that $D_{m+n}$ does not form a subgroup, since it does not contain the identity permutation.
\end{definition}

The permutation centraliser algebra $\cA(m,n)$ has a subspace spanned by orbits labelled by derangements.
\begin{definition}[Derangement orbit space]
    Let $g_1,\dots,g_{r} \in D_{m+n}$ be a complete set of representatives of orbits of $S_m \times S_n$ acting on $D_{m+n}$ by conjugation.
    Since cycle structure is preserved by conjugation, every element in the orbit of a derangement is a derangement.
    The representatives label an orbit basis for the \emph{derangement orbit space} (DOS)
\begin{equation}
\cD(m,n) = \Span_{\mathbb{C}}( P_{m,n}(g_p) \, \vert\, 1 \leq p \leq r  ) = \cA(m,n) \cap \mathbb{C}[D_{m+n}] \,.
\end{equation}
Here $\mathbb{C}[D_{m+n}]$ denotes the vector subspace of $\mathbb{C}[S_{m+n}]$ spanned by $D_{m+n}$, and the inclusion $D_{m+n} \hookrightarrow S_{m+n}$ induces the natural embedding $\cD(m,n) \hookrightarrow \cA(m,n)$.
\end{definition}
\begin{example}
    The DOS $\cD(2,2)$ has dimension $r = 4$. A complete set of representatives is
    \begin{equation}
        \begin{aligned}
            &&g_1 = (12)(34),   \,
            &&g_2 = (13)(24),   \,
            &&g_3 = (1234),     \,
            &&g_{4} = (1324)\, .
        \end{aligned}
        \label{eq: cD22 orbits}
    \end{equation}
    This should be compared to $\cA(2,2)$ in Example \ref{ex:cA22} which has dimension ten.
    A subset of its orbit basis elements corresponds to Equation \eqref{eq: cD22 orbits}.
\end{example}

The DOS $\cD(m,n)$ maps onto $\mathcal{H}'_N(m,n)$ through the derangements-to-invariants map.
\begin{definition}[Derangements-to-invariants map]
\label{eq:Osigma-traceless-definition}
    Let $g\in D_{m+n}$ and fix $m,n \in \mathbb{Z}_{\geq 0}$ and $N\geq 1$, with $L=m+n$. Define the multi-trace invariant
\begin{equation}
    \begin{aligned}
    &\swmap'^{(N)}_{m,n}: \cD(m,n) \rightarrow \mathcal{H}'_N(m,n)\\
    &\swmap'^{(N)}_{m,n}(g)
\equiv
\sum_{i_1,\dots,i_L=1}^N
\prod_{p=1}^m Z'_{i_p i_{g(p)}}
\prod_{q=m+1}^L W'_{i_q i_{g(q)}} \,.
    \end{aligned}
\end{equation}
For an element $\sigma=\sum_{g\in D_{m+n}}c_g(\sigma)g \in \cD(m,n)$ we linearly extend the above map,
\begin{equation}
    \swmap'^{(N)}_{m,n}(\sigma)
\equiv
\sum_{g\in D_{m+n}}c_g(\sigma)\,\swmap'^{(N)}_{m,n}(g)
\end{equation}
where $c_g(\sigma) \in \mathbb{C}$.
\end{definition}
\noindent

The derangements-to-invariants map is surjective for all $N$ from the following argument. For any $\poly' (Z', W') \in \cH'_N(m,n)$ we write $\poly' (Z', W') = \poly' (Z-I z, W- I w) \in \cH_N(m,n)$, which is labelled by a linear combination of permutations by the Schur-Weyl duality for $U(N)$. The permutations with fixed points correspond to single-letter traces, which vanish by $\tr(Z')=\tr(W')=0$. Thus $\swmap'^{(N)}_{m,n}$ on derangements is surjective to $\mathcal{H}'_N(m,n)$; see also \cite{PROCESI1976306,BenkartDoty2002,Corley:2002mj} for related discussions. 
In short, the DOS gives a faithful description of the space of traceless invariants in the stable regime.

\begin{mainproposition}
    In the stable regime $N \geq m+n$, the derangements-to-invariants map $\swmap'^{(N)}_{m,n}$ is bijective
    \begin{equation}
        \dim \cD(m,n) = 
        \dim \cH_N'(m,n)  \, .
    \end{equation}
    This is the analogue of the bijective correspondence between the PCA $\cA(m,n)$ and $\mathcal{H}_N(m,n)$ in the stable regime, for traceless matrices and derangement orbit spaces.
\end{mainproposition}
\begin{proof}
We give a simple argument based on counting using the inclusion-exclusion principle.
Note that $\dim \cD(m,n)$ is equal to the number of orbit basis elements of $\cA(m,n)$ with no fixed points.
The number of orbit basis elements of $\cA(m,n)$ with at least one fixed point in $\{1,\dots,m\}$ is given by $\dim \cA(m-1,n)$.
Similarly the number of orbits with at least one fixed point in $\{m+1,\dots,m+n\}$ is given by $\dim \cA(m,n-1)$.
However, this overcounts by the overlap -- the orbits of permutations with at least one fixed point in each set -- which have cardinality $\dim \cA(m-1,n-1)$.
The inclusion-exclusion principle gives
\begin{equation}
    \dim \cD(m,n) = \dim \cA(m,n) - \dim \cA(m,n-1) - \dim \cA(m-1,n) + \dim \cA(m-1,n-1) \,.
\end{equation}
For $N \geq m+n$ we have $\dim \cA(m,n) = \dim \mathcal{H}_N(m,n)$ and together with the factorisation in Proposition~\ref{prop:hilb factorisation} we get
\begin{equation}
    \begin{aligned}
        \dim \cD(m,n) &=
        \dim \cH_N(m,n) - \dim \cH_N(m,n-1) - \dim \cH_N(m-1,n) + \dim \cH_N(m-1,n-1) \\
        &= \dim \cH_N'(m,n) \,.
    \end{aligned}
\end{equation}
Since the map is surjective and the dimensions agree, the kernel is trivial.
\end{proof}
\noindent

The DOS inherits the orbit-basis inner product from the PCA (Definition~\ref{def:orbit-basis inner prod}).
\begin{definition}[Orbit-basis inner product for DOS]\label{def:orbit-basis inner prod DOS}
The orbit-basis inner product in $\cD(m,n)$ is given by
\begin{equation}
\langle \sigma, \sigma' \rangle = \delta(\sigma \cS(\sigma')) .
\label{eq: inner prod DOS}
\end{equation}
The Gram matrix in the orbit basis for $\cD(m,n)$ is defined by
\begin{equation}
    ( {G}_{\dos} (m,n) )_{pq} = \langle \orbitP{m,n}(g_p) , \orbitP{m,n}(g_q) \rangle , \qquad (p, q \in \{ 1, \dots, \dim \cD(m,n) \}).
\label{eq: Gram matrix orbit-DOS}
\end{equation}
\end{definition}

\subsection{Derangement framework of $SU(N)$ Hilbert space at finite $N$}\label{subsec:derangement hilb}
In this subsection we will propose a construction of finite-$N$ bases for $\mathcal{H}'_N(m,n)$ using derangements.
\begin{mainconjecture}[$SU(N)$ finite-$N$ bases]
    \label{conj:suN bases}
    Define the intersection of two subspaces $\cA_N(m,n)$ and $\cD(m,n)$ of the PCA $\cA(m,n)$ by
    \begin{equation}
        \cD_N(m,n) = \cA_N(m,n) \cap \cD(m,n)
    \end{equation}
    and denote the restriction of the derangements-to-invariants map to this intersection by $\hat\swmap'^{(N)}_{m,n} \equiv \swmap'^{(N)}_{m,n} \big|_{\cD_N(m,n)}$\,. 
    We conjecture that this map gives an isomorphism
    \begin{equation}
        \hat \swmap'^{(N)}_{m,n}: \cD_N(m,n) \ \overset{\sim}{\rightarrow} \ \mathcal{H}'_N(m,n), \qquad \forall N \in \bb{N} .
    \label{def:DI map finite-N}
    \end{equation}
\end{mainconjecture}

We test the conjecture by comparing the dimension of the intersection with the dimension of the physical space, and by checking that the image of the restricted map fills the physical space:
\begin{equation}
\begin{aligned}
\dim\cD_N(m,n)&=\dim\mathcal{H}'_N(m,n),\\
\dim\im\hat\swmap'^{(N)}_{m,n}&=\dim\mathcal{H}'_N(m,n) \qquad \text{(Surjective)} .
\end{aligned}
\label{eq:consistency check}
\end{equation}
The second condition establishes surjectivity. Together with the rank-nullity theorem
\begin{equation}
\dim\im\hat\swmap'^{(N)}_{m,n}+\dim\ker\hat\swmap'^{(N)}_{m,n}=\dim\cD_N(m,n)
\label{eq:ranknullity}
\end{equation}
the two conditions establish injectivity by
\begin{equation}
\dim \ker \hat\swmap'^{(N)}_{m,n} = 0 \qquad \text{(Injective)}.
\end{equation}
We implement these two checks using an intersection basis and its $SU(N)$ two-point functions, as follows.

\noindent\textbf{Computational Check 1: Dimension of the intersection.}\par
The dimension $\dim\mathcal{H}'_N(m,n)$ is obtained from \eqref{eq:H' dim diff formula}. 
The dimension $\dim\cD_N(m,n)$ is obtained by constructing a basis of $\cD_N(m,n)$ following  standard methods of  linear algebra. 

The construction of a basis of $\cD_N(m,n)$ is described as follows.
Let $\{ e_1,\dots, e_l \}$ be the orbit basis elements of the PCA with $l=\dim\cA(m,n)$, and $\{ b_1,\dots, b_r \}$ be those of the DOS with $r=\dim\cD(m,n)$.
Let $\{\hat e_1,\dots,\hat e_{l_N}\}$ be a basis of the finite-$N$ PCA subspace with $l_N=\dim\cA_N(m,n)$, and $\{\hat b_1,\dots,\hat b_{r_N}\}$ a basis of the intersection, with $r_N=\dim\cD_N(m,n)$. 
The four bases are related by rectangular matrices $B$, $B_N$ and ${B}_{\dos}$ as
\begin{equation}
\hat b_q = \sum_{a=1}^l B_{a q} \, e_a 
= \sum_{i=1}^{l_N} (B_N)_{i q} \, \hat e_i 
= \sum_{p=1}^r ({B}_{\dos})_{p q} \, b_p \,, \qquad q \in\{1,\dots, r_N\}.
\label{def:BtB}
\end{equation}
Thus, $B$, $B_N$ and $B_{\dos}$ embed the coordinates of the intersection into those of the PCA, the finite\,-$N$ PCA subspace and the DOS, respectively.
The construction of these rectangular matrices is reviewed in Appendix \ref{apx: null space algo}.

\noindent\textbf{Computational Check 2: Dimension of the image.}\par
To test the second condition in \eqref{eq:consistency check}, we evaluate the free two-point functions of the polynomials $\swmap'^{(N)}_{m,n}(\hat b_q)(Z',W')$ using the $SU(N)$ Wick contractions with the spacetime dependence removed:
\begin{equation}
\langle Z'_{ij}Z'^\dagger_{kl}\rangle_{SU(N)} =\langle W'_{ij}W'^\dagger_{kl}\rangle_{SU(N)} 
=\delta_{il}\delta_{jk}-\frac{1}{N}\delta_{ij}\delta_{kl} \,.
\label{eq:suN wick}
\end{equation}
They define a positive-definite inner product on the space $\mathcal{H}'_N(m,n)$, generalising Definition \ref{prop:UN inner}. 
The Gram matrix of the images of our intersection basis is
\begin{equation}
    (\mathcal{G}_{SU(N)}(m,n))_{pq}=\left\langle\swmap'^{(N)}_{m,n}(\hat b_p)(Z',W')\,\swmap'^{(N)}_{m,n}(\cS(\hat b_q))(Z'^\dagger,W'^\dagger)\right\rangle_{SU(N)}, \qquad p,q\in\{1,\dots,r_N\}.
\label{def:SUN Gram bb}
\end{equation}
This Gram matrix is computed using the permutation formula in Appendix \ref{app:Gram diagram}. 
Since the physical inner product on $\mathcal{H}'_N(m,n)$ is positive definite, a linear combination of $\swmap'^{(N)}_{m,n}(\hat b)$ has zero norm if and only if it is the zero polynomial. 
Consequently,
\begin{equation}
    \rank \mathcal{G}_{SU(N)}(m,n)=\dim\im\hat\swmap'^{(N)}_{m,n}.
\end{equation}
The second line of \eqref{eq:consistency check} becomes
\begin{equation}
    \rank \mathcal{G}_{SU(N)}(m,n)=\dim\mathcal{H}'_N(m,n)
\end{equation}
which is our second check.

Table \ref{tbl:intsec dims} is a summary of the dimension of $\mathcal{H}'_N(m,n)$, the dimension of $\cD_N(m,n)$ and the rank of the Gram matrix $\mathcal{G}_{SU(N)}(m,n)$. Their agreement verifies both equations in \eqref{eq:consistency check}.

\begin{table}[ht!]
    \centering
    $
    \begin{array}{l  l | l  l l l }
        m & n &  N &   \dim \cH'_N(m,n) & \dim \cD_N(m,n) &\rank\mathcal{G}_{SU(N)}(m,n) \\
        \hline
        2 & 2 &  >3 & 4 &4 &4 \\
          &   &  3  & 3 &3 &3 \\
          &   &  2  & 2 &2 &2 \\
          \hline
        3 & 2 &  >4 & 5 &5 &5 \\
          &   &  4  & 4 &4 &4 \\
          &   &  3  & 3 &3 &3 \\
          &   &  2  & 0 &0 &0 \\ 
          \hline
        4 & 2 &  >5 & 11&11&11\\ 
          &   &  5  & 10&10&10\\ 
          &   &  4  & 9 &9 &9 \\ 
          &   &  3  & 5 &5 &5 \\ 
          &   &  2  & 2 &2 &2 \\
          \hline
        3 & 3 &  >5 & 12&12&12\\ 
          &   &  5  & 11&11&11\\ 
          &   &  4  & 10&10&10\\ 
          &   &  3  & 6 &6 &6 \\ 
          &   &  2  & 2 &2 &2 \\
          \hline
        4 & 3 &  >6 & 20&20&20\\
          &   &  6  & 19&19&19\\ 
          &   &  5  & 18&18&18\\
          &   &  4  & 13&13&13\\ 
          &   &  3  & 6 &6 &6 \\
          &   &  2  & 0 &0 &0   
    \end{array}
    $
\caption{A table of (1) dimensions of the space of finite-$N$ traceless matrix invariants computed using \eqref{eq:H' dim diff formula} 
(2) dimensions of $\cD_N(m,n)$ computed using the intersection algorithm outlined in Appendix \ref{apx: null space algo} 
(3) rank of the Gram matrix $\mathcal{G}_{SU(N)}(m,n)$ computed using Appendix~\ref{app:Gram diagram}. The three quantities agree in every displayed case, supporting Conjecture~\ref{conj:suN bases}.
In each $(m,n)$ block, the first row is in the stable regime and the remaining rows are in the unstable regime.}
\label{tbl:intsec dims}
\end{table}

\subsection{One-loop derangement dilatation operator at finite $N$}

Conjecture \ref{conj:suN bases} identifies the intersection $\cD_N(m,n)$ as a candidate permutation description of the physical $SU(N)$ state space. We now propose an operator on this intersection that represents the physical one-loop dilatation operator. Its characteristic polynomial will be tested in the next subsection against the factorisation formula in Proposition~\ref{cor:suN vs uN HWS fixed e}.

Our proposal is closely related to the $U(N)$ case in Corollary \ref{prop: unstable 1LPDO}.
The restricted permutations-to-invariants map intertwines $\fD^{(1)}_N$ with the permutation dilatation operator $H^{(1)}_{\gss = 1/N}$ followed by an orthogonal projection onto $\cA_N(m,n)$ with respect to the orbit-basis inner product. This projection removes elements of $\ker \swmap^{(N)}_{m,n}$ generated by the action of the permutation dilatation operator.

For $SU(N)$, the physical one-loop dilatation operator acting on $\mathcal{H}'_N$ is given by $\fD_N^{\prime(1)}$ in \eqref{eq: SUN deriv dil}.
Then, using the restricted derangements-to-invariants map $\hat\swmap'^{(N)}_{m,n}$ in \eqref{def:DI map finite-N}, we conjecture that the action of the physical one-loop dilatation is represented by the permutation action followed by orthogonal projection onto $\cD_N(m,n)$. 

\begin{mainconjecture}[One-loop derangement dilatation operator]
Let $P^{\cD_N(m,n)}:\cA(m,n)\rightarrow\cD_N(m,n)$ be the orthogonal projector with respect to the orbit-basis inner product in Definition \ref{def:orbit-basis inner prod}. Define the one-loop derangement dilatation operator by the composition
\begin{equation}
P^{\cD_N(m,n)}\circ\left.H^{(1)}_{\gss}\right|_{\cD_N(m,n)}:\cD_N(m,n)\rightarrow\cD_N(m,n).
\label{def:PDNH}
\end{equation}
Assuming Conjecture~\ref{conj:suN bases}, we propose that $\hat\swmap'^{(N)}_{m,n}$ intertwines this operator at $\gss=1/N$
with $\fD_N^{\prime(1)}$ on $\mathcal{H}'_N(m,n)$:
\begin{equation}
P^{\cD_N(m,n)} \circ\left.H^{(1)}_{\gss=1/N}\right|_{\cD_N(m,n)}
= (\hat\swmap'^{(N)}_{m,n})^{-1} \circ\fD_N^{\prime(1)} \circ\hat\swmap'^{(N)}_{m,n} \,.
\end{equation}
\end{mainconjecture}

It is useful to construct another family of Gram matrices from the orbit-basis inner product.
Let $G(m,n)_{kk'} = \langle \orbitP{m,n}(g_k) , \orbitP{m,n}(g_{k'}) \rangle$ be the orbit-basis Gram matrix on the PCA $\cA(m,n)$ from Definition \ref{def:orbit-basis inner prod}.
Its restriction to the DOS is $G_{\dos}(m,n)$ in \eqref{eq: Gram matrix orbit-DOS}. The orbit-basis Gram matrix on the intersection $\cD_N(m,n)$ is therefore
\begin{equation}
    {G}_{\intr}(m,n) \equiv B^tG(m,n)B = B_{\dos}^tG_{\dos}(m,n){B}_{\dos}
\label{def:mathcal GramN}
\end{equation}
where $B^t$ is the transpose of $B$, assuming that the matrices have real coefficients in the ambient orbit basis.
Since the orbit-basis inner product is positive definite, and since the columns of $B$ are independent, ${G}_{\intr}(m,n)$ is positive definite and has rank $r_N$.
In particular, ${G}_{\intr}(m,n)$ is an invertible matrix.

The matrix elements of the one-loop derangement dilatation operator in a basis of $\cD_N(m,n)$ can be computed from the matrices $B$ and $G(m,n)$ introduced in the previous subsection. Recall that $B$ has $l$ rows and $r_N$ independent columns, and that $G(m,n)$ represents the positive-definite orbit inner product on $\cA(m,n)$. From \eqref{def:mathcal GramN}, the orbit-basis Gram matrix on the intersection ${G}_{\intr}(m,n) = B^tG(m,n)B$ is invertible.
We then obtain an $l \times l$ matrix representation of the orthogonal projector $P^{\cD_N(m,n)}$ in the orbit basis as
\begin{equation}
    P=B[{G}_{\intr}(m,n)]^{-1}B^tG(m,n)=B(B^tG(m,n)B)^{-1}B^tG(m,n).
\label{eq:cDN def}
\end{equation}
This projector maps ambient coordinate vectors orthogonally onto the column space of $B$, which represents $\cD_N(m,n)$.

Let $H$ be the $l\times l$ matrix representation of $H^{(1)}_{\gss}$ at $\gss=1/N$ in the same ambient basis. 
The matrix $(PHB)$ has $l$ rows and $r_N$ columns representing the image of $P^{\cD_N(m,n)}$.
Since $P$ and $B$ have the same column space,  there exists a unique $r_N\times r_N$ matrix $X$ satisfying
\begin{equation}
BX=PHB.
\label{eq:tilde H in B}
\end{equation}
The matrix $X$ represents the derangement dilatation operator \eqref{def:PDNH} in the basis $\{\hat b_1,\dots,\hat b_{r_N}\}$.

Multiplying \eqref{eq:tilde H in B} by $B^tG(m,n)$ gives
\begin{equation}
(B^tG(m,n)B)X = B^t G(m,n) B (B^tG(m,n)B)^{-1}B^tG(m,n) HB = B^tG(m,n)HB.
\end{equation}
Therefore,
\begin{equation}
X=(B^tG(m,n)B)^{-1}B^tG(m,n)HB.
\end{equation}
This equation proves the uniqueness of $X$.

\subsection{Examples of characteristic polynomials}
\label{sec: ex charpoly UN}
In this subsection we will show explicit examples of characteristic polynomials of the one-loop derangement dilatation operator.
To make the formulas more readable we will use the notation 
\begin{equation}
        \widetilde{H}^{(1)}_{\gss=\frac{1}{N}}  = \left. P^{\cD_N(m,n)}\circ H^{(1)}_{\gss=\frac{1}{N}}\right \vert_{\cD_N(m,n)} \,.
\end{equation}
We give evidence that this operator correctly produces the characteristic polynomial of the $SU(N)$ one-loop dilatation operator by comparing to Proposition~\ref{cor:suN vs uN HWS fixed e}:
\begin{equation}
    \det_{\cD_N^{\mathrm{HWS}}(m,n)}(x- \widetilde{H}^{(1)}_{\gss=\frac{1}{N}} )
\stackrel{?}{=}
\frac{\det_{\mathcal{H}_N^\mathrm{HWS}(m,n)}(x- \fD^{(1)}_N)
\det_{\mathcal{H}_N^\mathrm{HWS}(m-1,n-1)}(x- \fD^{(1)}_N)
}{
[\det_{\mathcal{H}_N^\mathrm{HWS}(m-1,n)}(x- \fD^{(1)}_N)]^{(1-\delta_{m,n})}
\det_{\mathcal{H}_N^\mathrm{HWS}(m,n-1)}(x- \fD^{(1)}_N)
} \,.
\label{eq:check SUN spec}
\end{equation}
As in Subsection \ref{subsec:ex charpoly} on $U(N)$ characteristic polynomials, we have restricted ourselves to the subspaces
\begin{equation}
    \cD^{\mathrm{HWS}}(m,n) = \cA^{\mathrm{HWS}}(m,n) \cap \mathbb{C}[D_{m+n}], \qquad 
    \cD_N^{\mathrm{HWS}}(m,n)  = \cA_N^{\mathrm{HWS}}(m,n) \cap \mathbb{C}[D_{m+n}]
\end{equation}
of highest-weight states under the $U(2)$ symmetry.
We give explicit data for the bidegrees $(2,2)$, $(3,2)$, $(4,2)$, and $(4,3)$. The spectral comparison has been checked for every bidegree with $m+n\leq 14$ and every relevant finite rank $N\geq 2$ including the stable regime.
The displayed results are summarised in Tables \ref{tbl: D22 char polys}, \ref{tbl: D32 char polys}, \ref{tbl: D42 char polys}, and \ref{tbl: D43 char polys}. 
We will return to physical aspects of these results in \cite{PRS_code_2026}.

\subsubsection{$\cD(2,2)$}
The characteristic polynomials of the one-loop derangement dilatation operator acting on ${\cD_N^{\mathrm{HWS}}(2,2)}$ are given by
\begin{align}
    \label{eq:cD222 char poly}
        N=2: \det_{\cD_2^{\mathrm{HWS}}(2,2)}(x- \widetilde H^{(1)}_{\gss = \frac{1}{2}}) &= (x-6)\\
    \label{eq:cD322 char poly}
        N=3: \det_{\cD_3^{\mathrm{HWS}}(2,2)}(x- \widetilde H^{(1)}_{\gss=\frac{1}{3}}) &= x(x-6)
\end{align}
in the unstable regimes.
In the stable regime $N \geq 4$, the characteristic polynomial can be computed for all $\gss$, giving
\begin{equation}
    \det_{\cD^{\mathrm{HWS}}(2,2)}(x- \widetilde H^{(1)}_{\gss}) = x(x-6) \,.
\end{equation}
Note that the unstable polynomials divide the stable polynomial.

The consistency check involves comparing these results to Proposition~\ref{cor:suN vs uN HWS fixed e}.
The characteristic polynomials entering the ratio in Proposition~\ref{cor:suN vs uN HWS fixed e}, at $N=2$, are
\begin{equation}
    \begin{aligned}
        &\det_{\mathcal{H}^{\mathrm{HWS}}_{2}(2,2)}(x-  \fD^{(1)}_{2}) = x(x-6), \quad
        &&\det_{\mathcal{H}^{\mathrm{HWS}}_{2}(1,1)}(x-  \fD^{(1)}_{2}) = 1 \\
        &\det_{\mathcal{H}^{\mathrm{HWS}}_{2}(2,1)}(x-  \fD^{(1)}_{2}) = x, \quad 
        &&
    \end{aligned}
\end{equation}
which gives
\begin{equation}
    \det_{\mathcal{H}'^{\mathrm{HWS}}_2(2,2)}(x-  \fD^{(1)}_{2}) = \frac{x(x-6) \times 1 }{x \times 1} = (x-6)
\end{equation}
and agrees with \eqref{eq:cD222 char poly}.
Similarly, for $N=3$ the relevant characteristic polynomials are
\begin{equation}
    \begin{aligned}
        &\det_{\mathcal{H}^{\mathrm{HWS}}_{3}(2,2)}(x-  \fD^{(1)}_{3}) = x^2(x-6), \quad
        &&\det_{\mathcal{H}^{\mathrm{HWS}}_{3}(1,1)}(x-  \fD^{(1)}_{3}) = 1 \\
        &\det_{\mathcal{H}^{\mathrm{HWS}}_{3}(2,1)}(x-  \fD^{(1)}_{3}) = x, \quad 
        &&
    \end{aligned}
\end{equation}
which gives
\begin{equation}
    \det_{\mathcal{H}'^{\mathrm{HWS}}_3(2,2)}(x-  \fD^{(1)}_{3}) = \frac{x^2(x-6) \times 1 }{x \times 1} = x(x-6)
\end{equation}
and agrees with \eqref{eq:cD322 char poly}.
We summarise these results in Table \ref{tbl: D22 char polys}.
\begin{table}[ht!]
    \centering
    $
    \renewcommand{\arraystretch}{2}
    \begin{array}{c | c | c | c}
        {\displaystyle N} & { \displaystyle \det_{\cD_N^{\mathrm{HWS}}(2,2)}(x- \widetilde H^{(1)}_{\gss}) }& { \displaystyle \det_{\cA_N^{\mathrm{HWS}}(2,2)}(x-  H^{(1)}_{\gss})\det_{\cA_N^{\mathrm{HWS}}(1,1)}(x-  H^{(1)}_{\gss}) } & { \displaystyle \det_{\cA_N^{\mathrm{HWS}}(2,1)}(x-  H^{(1)}_{\gss}) }\\
        \hline
        2  & (x-6) & x(x-6) & x \\
        >2  & x(x-6) & x^2(x-6) & x
    \end{array}
    $
    \caption{The finite-$N$ characteristic polynomials of the $SU(N)$ one-loop dilatation operator in the highest-weight space of the $\Lambda=[2,2]$ representation of $U(2)$, computed using the PCA framework. The second column is computed using derangement orbit spaces $\cD_N^{\mathrm{HWS}}(2,2)$. The third and fourth columns are computed using permutation centraliser spaces $\cA_N^{\mathrm{HWS}}(m,n)$. The second column is equal to the ratio of the third and fourth columns as predicted by Proposition~\ref{cor:suN vs uN HWS fixed e}.}
    \label{tbl: D22 char polys}
\end{table}

\subsubsection{$\cD(3,2)$}
The characteristic polynomials of the one-loop derangement dilatation operator acting on $\cD_N^{\mathrm{HWS}}(3,2)$ in the highest-weight space of $\Lambda=[3,2]$ are given by
\begin{align}
    \label{eq:cD232 char poly}
        N=2: \det_{\cD_2^{\mathrm{HWS}}(3,2)}(x- \widetilde H^{(1)}_{\gss = \frac{1}{2}}) &= 1\\
    \label{eq:cD332 char poly}
        N=3: \det_{\cD_3^{\mathrm{HWS}}(3,2)}(x- \widetilde H^{(1)}_{\gss=\frac{1}{3}}) &= (x-4) \\
    \label{eq:cD432 char poly}
        N=4: \det_{\cD_4^{\mathrm{HWS}}(3,2)}(x- \widetilde H^{(1)}_{\gss=\frac{1}{4}}) &= x(x-4)
\end{align}
in the unstable regimes.
In the stable regime $N \geq 5$, the characteristic polynomial can be computed for all $\gss$, giving
\begin{equation}
    \det_{\cD^{\mathrm{HWS}}(3,2)}(x- \widetilde H^{(1)}_{\gss}) = x(x-4) \,.
\end{equation}
The unstable polynomials divide the stable polynomial here as well.

To compare with Proposition~\ref{cor:suN vs uN HWS fixed e} at $N=2$, the relevant characteristic polynomials are
\begin{equation}
    \begin{aligned}
        &\det_{\mathcal{H}^{\mathrm{HWS}}_{2}(3,2)}(x-  \fD^{(1)}_{2}) = x(x-6), \quad
        &&\det_{\mathcal{H}^{\mathrm{HWS}}_{2}(2,1)}(x-  \fD^{(1)}_{2}) = x \\
        &\det_{\mathcal{H}^{\mathrm{HWS}}_{2}(2,2)}(x-  \fD^{(1)}_{2}) = x(x-6), \quad 
        &&\det_{\mathcal{H}^{\mathrm{HWS}}_{2}(3,1)}(x-  \fD^{(1)}_{2}) = x
    \end{aligned}
\end{equation}
which agrees with \eqref{eq:cD232 char poly} since then
\begin{equation}
    \det_{\mathcal{H}'^{\mathrm{HWS}}_2(3,2)}(x-  \fD^{(1)}_{2}) = 1 \,.
\end{equation}
For $N=3$ the polynomials are
\begin{equation}
    \begin{aligned}
        &\det_{\mathcal{H}^{\mathrm{HWS}}_{3}(3,2)}(x-  \fD^{(1)}_{3}) = x^3(x-4)(x-6), \quad
        &&\det_{\mathcal{H}^{\mathrm{HWS}}_{3}(2,1)}(x-  \fD^{(1)}_{3}) = x \\
        &\det_{\mathcal{H}^{\mathrm{HWS}}_{3}(2,2)}(x-  \fD^{(1)}_{3}) = x^2(x-6), \quad 
        &&\det_{\mathcal{H}^{\mathrm{HWS}}_{3}(3,1)}(x-  \fD^{(1)}_{3}) = x^2
    \end{aligned}
\end{equation}
which agrees with \eqref{eq:cD332 char poly} since then
\begin{equation}
    \det_{\mathcal{H}'^{\mathrm{HWS}}_3(3,2)}(x-  \fD^{(1)}_{3}) =  \frac{x^3(x-4)(x-6) \times x}{x^2(x-6)\times x^2} = (x-4)\,.
\end{equation}
Lastly, for $N=4$ we have
\begin{equation}
    \begin{aligned}
        &\det_{\mathcal{H}^{\mathrm{HWS}}_{4}(3,2)}(x-  \fD^{(1)}_{4}) = x^4(x-4)(x-6), \quad
        &&\det_{\mathcal{H}^{\mathrm{HWS}}_{4}(2,1)}(x-  \fD^{(1)}_{4}) = x \\
        &\det_{\mathcal{H}^{\mathrm{HWS}}_{4}(2,2)}(x-  \fD^{(1)}_{4}) = x^2(x-6), \quad 
        &&\det_{\mathcal{H}^{\mathrm{HWS}}_{4}(3,1)}(x-  \fD^{(1)}_{4}) = x^2
    \end{aligned}
\end{equation}
which agrees with \eqref{eq:cD432 char poly} since
\begin{equation}
    \det_{\mathcal{H}'^{\mathrm{HWS}}_4(3,2)}(x-  \fD^{(1)}_{4}) =  \frac{x^4(x-4)(x-6) \times x}{x^2(x-6)\times x^2} = x(x-4)\,.
\end{equation}
We summarise these results in Table \ref{tbl: D32 char polys}.
\begin{table}[ht!]
    \centering
    $
    \renewcommand{\arraystretch}{2}
    \begin{array}{c | c | c | c}
N & \displaystyle \det_{\cD_N^{\mathrm{HWS}}(3,2)} (x-\widetilde H^{(1)}_{\gss}) 
& \begin{aligned}[t]
  &\det_{\cA_N^{\mathrm{HWS}}(3,2)} (x-H^{(1)}_{\gss})\times \\ &\qquad \det_{\cA_N^{\mathrm{HWS}}(2,1)} (x-H^{(1)}_{\gss}) \end{aligned}
& \begin{aligned}[t] &\det_{\cA_N^{\mathrm{HWS}}(2,2)} (x-H^{(1)}_{\gss})\times \\ &\qquad \det_{\cA_N^{\mathrm{HWS}}(3,1)} (x-H^{(1)}_{\gss}) \end{aligned}
\\
        \hline
        2  & 1 & x^2(x-6) & x^2(x-6) \\
        3  & (x-4) & x^4(x-4)(x-6) & x^4(x-6) \\
        >3  & x(x-4) & x^5(x-4)(x-6) & x^4(x-6) \\
    \end{array}
    $
    \caption{The finite-$N$ characteristic polynomials of the $SU(N)$ one-loop dilatation operator in the highest-weight space of the $\Lambda=[3,2]$ representation of $U(2)$, computed using the PCA framework. The second column is computed using derangement orbit space $\cD_N^{\mathrm{HWS}}(3,2)$. The third and fourth columns are computed using permutation centraliser spaces $\cA_N^{\mathrm{HWS}}(m,n)$. The second column is equal to the ratio of the third and fourth columns as predicted by Proposition~\ref{cor:suN vs uN HWS fixed e}.}
    \label{tbl: D32 char polys}
\end{table}

\subsubsection{$\cD(4,2)$}
The characteristic polynomials of the one-loop derangement dilatation operator acting on ${\cD_N^{\mathrm{HWS}}(4,2)}$ are given by
\begin{align}
    \label{eq:cD242 char poly}
        N=2: \det_{\cD_2^{\mathrm{HWS}}(4,2)}(x- \widetilde H^{(1)}_{\gss = \frac{1}{2}}) &= (x-10)
        \\
    \label{eq:cD342 char poly}
        N=3: \det_{\cD_3^{\mathrm{HWS}}(4,2)}(x- \widetilde H^{(1)}_{\gss=\frac{1}{3}}) &= x(x^2-\frac{38}{3}x+\frac{100}{3}) 
        \\
    \label{eq:cD442 char poly}
        N=4: \det_{\cD_4^{\mathrm{HWS}}(4,2)}(x- \widetilde H^{(1)}_{\gss=\frac{1}{4}}) &= x^2(x^3 - 16x^2 + \frac{155}{2}x - 115)
        \\
    \label{eq:cD542 char poly}
        N=5: \det_{\cD_5^{\mathrm{HWS}}(4,2)}(x- \widetilde H^{(1)}_{\gss=\frac{1}{5}}) &= x^3(x^3-16x^2+\frac{392}{5}x-\frac{584}{5})
\end{align}
in the unstable regimes.
In the stable regime $N \geq 6$, the characteristic polynomial can be computed for all $\gss$, giving
\begin{equation}
    \det_{\cD^{\mathrm{HWS}}(4,2)}(x- \widetilde H^{(1)}_{\gss}) 
    = x^3(x^3 - 16x^2 +(80-40 \, \gss^2) x+  80 \, \gss^2 -120  ) .
\end{equation}
The same elementary computations as in Section \ref{subsec:ex charpoly} show that the polynomials in the unstable regime divide the polynomial in the stable regime at $\gss = N^{-1}$.

To compare with Proposition~\ref{cor:suN vs uN HWS fixed e}, the relevant characteristic polynomials at $N=2$ are
\begin{equation}
    \begin{aligned}
        &\det_{\mathcal{H}^{\mathrm{HWS}}_{2}(4,2)}(x-  \fD^{(1)}_{2}) = 
        x^2(x-6)(x-10), \quad
        &&\det_{\mathcal{H}^{\mathrm{HWS}}_{2}(3,1)}(x-  \fD^{(1)}_{2}) =
        x \\
        &\det_{\mathcal{H}^{\mathrm{HWS}}_{2}(3,2)}(x-  \fD^{(1)}_{2}) =
        x(x-6), \quad 
        &&\det_{\mathcal{H}^{\mathrm{HWS}}_{2}(4,1)}(x-  \fD^{(1)}_{2}) =
        x^2
    \end{aligned}
\end{equation}
which agrees with \eqref{eq:cD242 char poly} since then
\begin{equation}
    \det_{\mathcal{H}'^{\mathrm{HWS}}_2(4,2)}(x-  \fD^{(1)}_{2}) = (x-10)\,.
\end{equation}
For $N=3$ the polynomials are
\begin{equation}
    \begin{aligned}
        &\det_{\mathcal{H}^{\mathrm{HWS}}_{3}(4,2)}(x-  \fD^{(1)}_{3}) =
        x^6(x-4)(x-6)(x^2-\frac{38}{3}x+\frac{100}{3}), \quad
        &&\det_{\mathcal{H}^{\mathrm{HWS}}_{3}(3,1)}(x-  \fD^{(1)}_{3}) =
        x^2 \\
        &\det_{\mathcal{H}^{\mathrm{HWS}}_{3}(3,2)}(x-  \fD^{(1)}_{3}) =
        x^3(x-4)(x-6), \quad 
        &&\det_{\mathcal{H}^{\mathrm{HWS}}_{3}(4,1)}(x-  \fD^{(1)}_{3}) =
        x^4
    \end{aligned}
\end{equation}
which agrees with \eqref{eq:cD342 char poly} since then
\begin{equation}
    \det_{\mathcal{H}'^{\mathrm{HWS}}_3(4,2)}(x-  \fD^{(1)}_{3}) =  \frac{x^6(x-4)(x-6)(x^2-\frac{38}{3}x+\frac{100}{3}) \times x^2}{x^3(x-4)(x-6)\times x^4} = x(x^2-\frac{38}{3}x+\frac{100}{3})\,.
\end{equation}
For $N=4$ we have
\begin{equation}
    \begin{aligned}
        &\det_{\mathcal{H}^{\mathrm{HWS}}_{4}(4,2)}(x-  \fD^{(1)}_{4}) =
        x^9(x-4)(x-6)(x^3-16x^2 + \frac{155}{2}x - 115), \quad
        &&\det_{\mathcal{H}^{\mathrm{HWS}}_{4}(3,1)}(x-  \fD^{(1)}_{4}) =
        x^2 \\
        &\det_{\mathcal{H}^{\mathrm{HWS}}_{4}(3,2)}(x-  \fD^{(1)}_{4}) =
        x^4(x-4)(x-6), \quad 
        &&\det_{\mathcal{H}^{\mathrm{HWS}}_{4}(4,1)}(x-  \fD^{(1)}_{4}) =
        x^5
    \end{aligned}
\end{equation}
which agrees with \eqref{eq:cD442 char poly} since
\begin{equation}
    \det_{\mathcal{H}'^{\mathrm{HWS}}_4(4,2)}(x-  \fD^{(1)}_{4}) =
    \frac{x^9(x-4)(x-6)(x^3-16x^2 + \frac{155}{2}x - 115)  x^2}{x^4(x-4)(x-6)x^5} =
    x^2(x^3-16x^2 + \frac{155}{2}x - 115)
    \,.
\end{equation}
For $N=5$ we have
\begin{equation}
    \begin{aligned}
        &\det_{\mathcal{H}^{\mathrm{HWS}}_{5}(4,2)}(x-  \fD^{(1)}_{5}) =
        x^{10}(x-4)(x-6)(x^3-16x^2 + \frac{392}{5}x - \frac{584}{5}), \quad
        &&\det_{\mathcal{H}^{\mathrm{HWS}}_{5}(3,1)}(x-  \fD^{(1)}_{5}) =
        x^2 \\
        &\det_{\mathcal{H}^{\mathrm{HWS}}_{5}(3,2)}(x-  \fD^{(1)}_{5}) =
        x^4(x-4)(x-6), \quad 
        &&\det_{\mathcal{H}^{\mathrm{HWS}}_{5}(4,1)}(x-  \fD^{(1)}_{5}) =
        x^5
    \end{aligned}
\end{equation}
which agrees with \eqref{eq:cD542 char poly} since it gives

\begin{equation}
    \det_{\mathcal{H}'^{\mathrm{HWS}}_5(4,2)}(x-  \fD^{(1)}_{5}) =
    x^3(x^3-16x^2 + \frac{392}{5}x - \frac{584}{5})
    \,.
\end{equation}
We summarise these results in Table \ref{tbl: D42 char polys}.
\begin{table}[ht!]
    \centering
    \renewcommand{\arraystretch}{1.5}
$\begin{array}{c | c | c | c}
        N & \displaystyle \det_{\cD_N^{\mathrm{HWS}}(4,2)}(x-\widetilde H^{(1)}_{\gss}) & \begin{aligned}[t] &\det_{\cA_N^{\mathrm{HWS}}(4,2)}(x- H^{(1)}_{\gss})\times \\
        &\qquad \det_{\cA_N^{\mathrm{HWS}}(3,1)}(x- H^{(1)}_{\gss}) \end{aligned} & \begin{aligned}[t] &\det_{\cA_N^{\mathrm{HWS}}(3,2)}(x- H^{(1)}_{\gss})\times \\
        &\qquad \det_{\cA_N^{\mathrm{HWS}}(4,1)}(x- H^{(1)}_{\gss}) \end{aligned} \\
        \hline
        2 & (x-10) & x^3(x-6)(x-10) & x^3(x-6) \\
        3 & x (x^2-\frac{38}{3}x+\frac{100}{3}) & \begin{aligned}[t] &x^8(x-4)(x-6)\times \\
        &\quad (x^2-\frac{38}{3}x+\frac{100}{3}) \end{aligned} & x^7(x-4)(x-6) \\
        4 & \begin{aligned}[t] &x^2\times \\
        &\quad (x^3-16x^2+\frac{155x}{2}-115) \end{aligned} & \begin{aligned}[t] &x^{11}(x-4)(x-6)\times \\
        &\quad (x^3-16x^2+\frac{155x}{2}-115) \end{aligned} & x^9(x-4)(x-6) \\
        5 & \begin{aligned}[t] &x^3\times \\
        &\quad (x^3-16x^2+\frac{392x}{5}-\frac{584}{5}) \end{aligned} & \begin{aligned}[t] &x^{12}(x-4)(x-6)\times \\
        &\quad (x^3-16x^2+\frac{392x}{5}-\frac{584}{5}) \end{aligned} & x^9(x-4)(x-6) \\
        >5 & x^3 q_{\gss}(x) &x^{12}(x-4)(x-6) q_{\gss}(x) & x^9(x-4)(x-6) \\
\end{array}$
    \caption{The finite-$N$ characteristic polynomials of the $SU(N)$ one-loop dilatation operator in the highest-weight space of the $\Lambda=[4,2]$ representation of $U(2)$, computed using the PCA framework. The second column is computed using derangement orbit space $\cD_N^{\mathrm{HWS}}(4,2)$. The third and fourth columns are computed using permutation centraliser spaces $\cA_N^{\mathrm{HWS}}(m,n)$. The second column is equal to the ratio of the third and fourth columns as predicted by Proposition~\ref{cor:suN vs uN HWS fixed e}. The polynomial $q_{\gss}(x) = (x^3-16x^2+(80-40\gss^2)x+80\gss^2-120)$ with $\xi=1/N$ is defined to save space.}
    \label{tbl: D42 char polys}
\end{table}

\begin{table}[ht!]
    \centering
    \renewcommand{\arraystretch}{2}
    $
    \begin{array}{c | c | c | c}
        N & \displaystyle \det_{\cD_N^{\mathrm{HWS}}(4,3)}(x-\widetilde H^{(1)}_{\gss}) & \begin{aligned}[t] &\det_{\cA_N^{\mathrm{HWS}}(4,3)}(x- H^{(1)}_{\gss})\times \\
        &\qquad \det_{\cA_N^{\mathrm{HWS}}(3,2)}(x- H^{(1)}_{\gss}) \end{aligned} & \begin{aligned}[t] &\det_{\cA_N^{\mathrm{HWS}}(3,3)}(x- H^{(1)}_{\gss})\times \\
        &\qquad \det_{\cA_N^{\mathrm{HWS}}(4,2)}(x- H^{(1)}_{\gss}) \end{aligned} \\
        \hline
        2 & 1 & x^2(x-6)(x-10) & x^2(x-6)(x-10) \\
        3 & (x-\frac{20}{3}) & \begin{aligned}[t] &x^7(x-4)^2(x-6)^2 (x-\frac{20}{3}) \times \\
        &\quad (x^2-\frac{38}{3}x+\frac{100}{3}) \end{aligned} & \begin{aligned}[t] &x^7(x-4)^2(x-6)^2\times \\
        &\quad (x^2-\frac{38}{3}x+\frac{100}{3}) \end{aligned} \\
        4 & \begin{aligned}[t] &x(x-5)\times \\
        &\quad (x^2-9x+\frac{35}{2}) \end{aligned} & \begin{aligned}[t] &x^{12}(x-4)^2(x-6)^2\times \\
        &\quad (x-5)(x^2-9x+\frac{35}{2})\times \\
        &\quad (x^3-16x^2+\frac{155}{2}x-115) \end{aligned} & \begin{aligned}[t] &x^{11}(x-4)^2(x-6)^2\times \\
        &\quad (x^3-16x^2+\frac{155}{2}x-115) \end{aligned} \\
        5 & \begin{aligned}[t] &x^2(x-5)\times \\
        &\quad (x^2-9x+\frac{92}{5}) \end{aligned} & \begin{aligned}[t] &x^{14}(x-4)^2 (x-5) (x-6)^2\times \\
        &\quad (x^2-9x+\frac{92}{5})\times \\
        &\quad (x^3-16x^2+\frac{392}{5}x-\frac{584}{5}) \end{aligned} & \begin{aligned}[t] &x^{12}(x-4)^2(x-6)^2\times \\
        &\quad (x^3-16x^2+\frac{392}{5}x-\frac{584}{5}) \end{aligned} \\
        >5 & \begin{aligned}[t] &x^2(x-5)\times \\
        &\quad (x^2-9x-40\gss^2+20) \end{aligned} & \begin{aligned}[t] &x^{14}(x-4)^2 (x-5) (x-6)^2 \, q_{\gss}(x)\times \\
        &\quad (x^2-9x-40\gss^2+20) \end{aligned} & x^{12}(x-4)^2(x-6)^2 \, q_{\gss}(x) 
    \end{array}
    $
\caption{The finite-$N$ characteristic polynomials of the $SU(N)$ one-loop dilatation operator in the highest-weight space of the $\Lambda=[4,3]$ representation of $U(2)$, computed using the PCA framework. The second column is computed using derangement orbit space $\cD_N^{\mathrm{HWS}}(4,3)$. The third and fourth columns are computed using permutation centraliser spaces $\cA_N^{\mathrm{HWS}}(m,n)$. The second column is equal to the ratio of the third and fourth columns as predicted by Proposition~\ref{cor:suN vs uN HWS fixed e}.
The polynomial $q_{\gss}(x) = (x^3-16x^2+(80-40\gss^2)x+80\gss^2-120)$ with $\xi=1/N$ is defined to save space.
Although $N=6$ lies in the unstable regime $N < 7$, the result coincides with the stable-regime polynomial after setting $\gss=1/6$.}
    \label{tbl: D43 char polys}
\end{table}

\section{Conclusion}\label{sec:conclusion}

We summarise the content of the paper and outline some future research directions. 

\subsection{Summary}

In this paper, we addressed the problem of calculating the exact finite $N$ dependence of the spectrum of the one-loop dilatation operator in the $\su2$ sector of $\cN=4$ SYM with gauge groups $U(N)$ and $SU(N)$, for fixed classical dimension of the gauge invariant operators. Multi-trace operators of bidegree $(m,n)$ are linearly independent in the stable regime $N\geq L=m+n$. In this regime, the dimension of the space $\cH_N(m,n)$ of $U(N)$-invariant states is independent of $N$, and is described by the permutation centraliser algebra $\cA(m,n)$ \cite{Mattioli:2016eyp}. In the unstable regime $N<L$, finite-$N$ trace relations remove states in a non-trivial way. This is the field-theory manifestation of the stringy exclusion principle, which is non-perturbative in the $1/N$ expansion. We now summarise our main results.

In Section \ref{sec:UN factorisation}, we defined the permutations-to-invariants map
\begin{equation}
\swmap^{(N)}_{m,n}:\cA(m,n)\longrightarrow\cH_N(m,n),
\end{equation}
which is an isomorphism in the stable regime. 
In the unstable region, the permutations-to-invariants map $\swmap^{(N)}_{m,n}$ has a non-trivial kernel.
Next, we defined bases for $\cA(m,n)$, labelled by Young diagrams $R\vdash L$, which we called SEP bases. Such bases diagonalise the left action of the centre of the group algebra $\mathbb{C}[S_{m+n}]$; the restricted Schur and covariant bases give two realisations \cite{Brown:2007xh,Bhattacharyya:2008rb,Bhattacharyya:2008xy,Brown:2008ij,Collins:2008gc}.
A SEP basis splits $\cA(m,n)$ into two subspaces. The sector with $\ell(R)>N$ is the kernel of the permutations-to-invariants map (Proposition~\ref{thm:UN-inv-fin}), whereas the complementary sector with $\ell(R)\leq N$, denoted by $\cA_N(m,n)$, is isomorphic to the space of $U(N)$ invariants for every $N$.

We then introduced a one-parameter family of one-loop permutation dilatation operators $ H^{(1)}_{\gss}$ \cite{Bellucci:2004ru,Pasukonis:2010rv}. 
We constructed $ H^{(1)}_{\gss}$ so that its action on $\cA_N(m,n)$ evaluated at $\gss = N^{-1}$ reproduces the action of the physical one-loop dilatation operator $ \fD^{(1)}_N$ on $\cH_N(m,n)$ for any positive integer $N$. This statement is expressed by the intertwining relation with $\swmap^{(N)}_{m,n}$ (Proposition~\ref{prop:D1 intertwining}).
It follows that the kernel of $\swmap^{(N)}_{m,n}$ is invariant under the one-loop permutation dilatation operator, and hence the matrix of $ H^{(1)}_{\gss = 1/N}$ is block triangular (Lemma \ref{cor:finite-N factorisation}). Our first main theorem then states that the characteristic polynomial computed from $ H^{(1)}_{\gss} $, evaluated at $\gss=1/N$, contains the exact finite-$N$ characteristic polynomial as a factor (Theorem \ref{prop:finite-N factorisation}):
\begin{equation}
 \left.  \det_{\cA(m,n)}\!\left(x- H^{(1)}_{\gss}\right)  \right|_{\gss=1/N}
 = Q^{(N)}_{m,n}(x)\, \det_{\cH_N(m,n)}\!\left(x-\fD_N^{(1)}\right).
\label{eq:conclusion-UN-factorisation}
\end{equation}
Thus, the physical spectrum in the unstable regime is a factor of a characteristic polynomial computed on the stable, $N$-independent space.

In Section \ref{sec:SU(N)}, we studied a second factorisation, associated with the Lie-algebra decomposition of $u(N)$ into $su(N) \oplus u(1)$. Decomposing the matrices into traceless and trace parts, $Z=Z'+zI$ and $W=W'+wI$, gives a direct proof of
$\cH_N\cong\cH'_N\otimes\mathbb{C}[z,w]$, where $\cH'_N$ is the space of $SU(N)$-invariant polynomials in $Z'$ and $W'$ (Proposition~\ref{prop:hilb factorisation}; see \cite{PROCESI1976306}).
We then proved that this factorisation holds for each eigenvalue of the one-loop dilatation operator (Theorem \ref{prop:eigenspace factorisation}). 
Therefore every $SU(N)$ eigenstate of eigenvalue $e$ generates an infinite tower of $U(N)$ eigenstates with the same anomalous dimension:
\begin{equation}
 \cH_N(e) \cong \cH'_N(e)\otimes\mathbb{C}[z,w].
\label{eq:conclusion-eigenspace-factorisation}
\end{equation}
This is stronger than a relation between the Hilbert series of the $SU(N)$ and $U(N)$ theories in the $\su2$ sector, and exhibits the decoupling of the free $U(1)$ multiplet at one loop. We expect the same factorisation to hold for the anomalous part of the perturbative dilatation operator in the $\su2$ sector, since trace parts do not enter commutator interactions. As a consequence, for every $N$, the $SU(N)$ characteristic polynomial is a ratio of $U(N)$ characteristic polynomials at neighbouring bidegrees (Corollary \ref{cor:suN charpoly} and Proposition~\ref{cor:suN vs uN HWS fixed e}). 
The same argument applies to the zero-eigenspaces, and hence to the counting of quarter-BPS states.

In Section \ref{sec:derangements}, we gave an intrinsic permutation description of the $SU(N)$ theory that does not rely on the $U(N)$ factorisation. For traceless matrices, a permutation with a fixed point produces a factor of $\tr(Z')$ or $\tr(W')$, and hence a vanishing invariant. This motivates labelling $SU(N)$ invariants by derangements, i.e. permutations with no fixed points. 
We proved that, in the stable regime, the derangement orbit space $\cD(m,n)$ is isomorphic to $\cH'_N(m,n)$ through the derangements-to-invariants map. In the unstable regime, we conjectured that the restriction of this map to $\cD_N(m,n)=\cA_N(m,n)\cap\cD(m,n)$ is an isomorphism onto $\cH'_N(m,n)$ for every $N$ (Conjecture \ref{conj:suN bases}):
\begin{equation}
\hat \swmap'^{(N)}_{m,n} \, : \, \cD_N(m,n) \ \stackrel{\sim}{\longrightarrow} \ \mathcal{H}'_N(m,n).
\label{eq: conj iso}
\end{equation}
This conjecture is supported by three independently computed quantities: the dimension of $\mathcal{H}'_N(m,n)$, the dimension of $\cD_N(m,n)$,
and the rank of the $SU(N)$ two-point-function Gram matrix. We verified their agreement for $L\leq7$.
We further proposed that the one-loop derangement dilatation operator
\begin{equation}
    P^{\cD_N(m,n)}\circ  H^{(1)}_{\gss = \frac{1}{N}}\big|_{\cD_N(m,n)}
\label{eq:conclusion-derangement-H}
\end{equation}
has the characteristic polynomial predicted by Proposition~\ref{cor:suN vs uN HWS fixed e}.
Therefore, the derangement construction also reproduces the correct eigenvalues, and the proposal was checked up to $L\leq 14$. These agreements provide evidence for both the intersection conjecture and the proposed dynamics. 
Assuming these conjectures, we obtain a novel self-contained permutation framework for $SU(N)$ invariants, free correlators and one-loop mixing in the $\su2$ sector at finite $N$.

We expect that the methods developed in this paper suffice to prove the intersection conjecture (Equation \eqref{eq: conj iso}) and 
to establish the corresponding intertwining property of $\widetilde H^{(1)}_{\gss=1/N}$.
Since this is of less immediate interest for physical applications, we will return to it in the future.
We will investigate the physical aspects of the spectral data in upcoming work \cite{PRS_code_2026}.
The present methods may also provide an avenue for constructing an $SU(N)$ analogue of a SEP basis, with $\cD_N(m,n)$ arising as the finite-$N$ truncation of a basis of $\cD(m,n)$.

\subsection{Open problems and related directions}

We close with a discussion  of a few related topics and future directions.

An important open problem is to identify the physical factors in the polynomial factorisation \eqref{eq:conclusion-UN-factorisation} directly from the $N$-independent spectral data. Theorem \ref{prop:finite-N factorisation} does not identify which roots survive when we specialise $ \xi $ to $1/N$. A selection rule that identifies the physical roots using only the $N$-independent data could reveal additional representation-theoretic structure in the one-loop dilatation operator, associated to the chain of subgroups  $U(1) \subset U(2) \subset \dots \subset U(N)$.

The relations among the gauge invariants at finite $N$ are the subject of the elegant second fundamental
theorem for matrix invariants, formulated in terms of traces and permutations
by Procesi \cite{PROCESI1976306} and obtained independently by Razmyslov
\cite{Razmyslov1974}. A cutoff on Young diagrams with more than $N$ rows plays a role there as well.
These mathematical works offer a complementary perspective on the description of invariants by focusing on the generators and
relations of the invariant ring.   In the classical invariant-ring approach, one fixes $N$ and studies generators of the trace relations together with their higher syzygies. 
Closed-form Hilbert series for two-matrix invariants — which capture the relations, and the higher-order relations between relations — are known only for small $N$ \cite{djokovic2006poincare}, in part because the degree of the Hilbert-series numerator grows with $N$. 
In the permutation approach used here and in \cite{Padellaro:2024rld}, one instead fixes $(m,n)$, works in the semisimple algebra $\cA(m,n)$, and removes the Young-diagram sectors with $\ell(R)>N$ which are closed under the one-loop dilatation operator. This description is uniform in $N$. Connecting it more directly to classical invariant theory should clarify how generators and syzygies are encoded by the Young-diagram cutoff. Related stable-regime counting results arise from symmetric pairs \cite{Willenbring2007stable} and from Hilbert series for invariants of several matrices under simultaneous conjugation \cite{HarrisWillenbring2014}.

The single-letter traces $\tr Z$ and $\tr W$ are scalar components of the free multiplet associated with the central $U(1)$ of the $U(N)$ gauge group. This multiplet describes the centre-of-mass motion of the coincident D3-branes, and is known as the singleton or doubleton representation of $psu(2,2|4)$ \cite{Gunaydin:1998sw}. The singleton appears as pure gauge modes in the bulk spectrum of type IIB supergravity on \AdSxS\ \cite{Kim:1985ez}, while its boundary degrees of freedom have been studied using topological field theory and its extensions \cite{Witten:1998wy,Maldacena:2001ss,Belov:2004ht,Bergman:2025isp}.

Theorem \ref{prop:eigenspace factorisation} states that, at every $N$, each one-loop eigenspace in the $\su2$ sector of the $U(N)$ theory is the corresponding $SU(N)$ eigenspace tensored with $\mathbb{C}[z,w]$. In the planar regime, the distinction is readily implemented by separating the single-letter traces from the single-trace operators dual to propagating string states.
For giant gravitons and LLM geometries, the familiar constructions use the $U(N)$ half-BPS sector \cite{Corley:2001zk,Lin:2004nb}, while large dimension  operators and correlators have also been constructed directly in the $SU(N)$ theory \cite{Corley:2002mj,deMelloKoch:2004crq,Brown:2007bb}.
Thus, the gauge theory provides both $U(N)$ and $SU(N)$ descriptions for large operators, but how these two descriptions are related on the gravity side remains to be understood.
A very interesting question motivated by our work is  whether two formulations of string theory on \AdSxS, each complete in its own terms, can reproduce the $SU(N)$ and $U(N)$ spectra and exhibit the corresponding factorisation of the state space at every level and any $N$.
Such a construction would clarify how the boundary singleton sector remains decoupled from the degrees of freedom describing branes and geometries, even when the number of excitations in the singleton sector is comparable to or larger than $N$.

Finally, the framework is algorithmic and suggests two concrete extensions. 
The first is computational: one can push to larger $L$ (and general $N$), where $N$ is much smaller than $L$, giving access to a strongly unstable regime in which quantum corrections to anomalous dimensions interplay with finite-$N$ trace relations. 
 The second is to enlarge the field content. Incorporating anticommuting, Grassmann-valued matrices into the PCA construction would extend the method beyond the $\su2$ sector toward the $\tfrac{1}{16}$-BPS sector.
 This sector includes microstates of supersymmetric AdS$_5$ black holes, whose SUSY-protected counting is encoded in the superconformal index \cite{Kinney:2005ej,Dolan:2007rq}.
 Finite-$N$ invariant rings have also attracted renewed attention in connection with non-perturbative states \cite{deMelloKoch:2025ngs} and the fortuity mechanism for BPS black-hole microstates \cite{Chang:2022mjp,Chang:2024zqi}. Relating the Young-diagram description of the kernel and its one-loop dynamics to these developments would  provide spectral information precisely in the non-planar regime where the bulk description is least understood.

\subsubsection*{Acknowledgements}
We thank Joseph Ben Geloun, Robert de Mello Koch, Yang Lei, Hai Lin, Adolfo Holguin, Sam van Leuven, Denjoe O'Connor, Pranav Pandit and Micha\l{} Studzi\'nski for useful discussions.
AP is supported by the Department of Atomic Energy, Government of India, under project no. RTI14019.
SR is supported by the Science and Technology Facilities Council (STFC) Consolidated Grant ST/T000686/1, Amplitudes, Strings and Duality and the Royal Society International Exchanges Grant \verb|IEC\NSFC\242376|. 
The ideas, arguments and writing in this paper came from the authors. AI models were used for assistance with aspects of the analysis and editing.

\appendix

\section{Matrix elements of one-loop dilatation operator}\label{app:induced action D1}

A linear map $H$ on the gauge-invariant operators in $\mathcal{H}_{N}(m,n)$ in the stable regime induces a matrix $H_{k'k}$ in the orbit basis of $\cA(m,n)$, denoted by
\begin{equation}
H \cdot \orbitP{m,n} (g_k) \equiv \sum_{k'=1}^{\dim \cA(m,n)} H_{k'k} \, \orbitP{m,n} (g_{k'}) , \qquad N \ge L .
\label{app:induced action F2}
\end{equation}
In this appendix, we compute the matrix elements of the one-loop dilatation operator in terms of permutations \cite{Bellucci:2004ru,Pasukonis:2010rv}.
Colour contractions produce a factor of $N$, and we can treat $N$ as an arbitrary parameter.
This calculation is used in the proof of Proposition~\ref{prop:D1 intertwining}.

Let $(\partial/\partial Z)$ act on the $p$-th position and $(\partial/\partial W)$ act on the $q$-th position, where $p \in \{1,\dots,m\}$ and $q \in \{m+1,\dots,L\}$.
For convenience, we write
\begin{equation}
X_a=\begin{cases}
Z & (1\leq a\leq m),\\
W & (m+1\leq a\leq L).
\end{cases}
\end{equation}
A direct differentiation gives
\begin{multline}
    N\, \fD^{(1)}_N \cdot \swmap^{(N)}_{m,n}(g)
= \sum_{i_1,\dots,i_L=1}^N
\sum_{p=1}^m\sum_{q=m+1}^L
\prod_{\substack{a=1\\a\neq p,q}}^L (X_a)_{i_a i_{g(a)}}
\times
\\
\sum_{\nu=1}^N
\Bigl\{
-\delta_{i_p i_{g(q)}} Z_{i_q\nu}W_{\nu i_{g(p)}}
+\delta_{i_p i_{g(q)}} W_{i_q\nu}Z_{\nu i_{g(p)}}
+Z_{i_p\nu}W_{\nu i_{g(q)}}\delta_{i_q i_{g(p)}}
-W_{i_p\nu}Z_{\nu i_{g(q)}}\delta_{i_q i_{g(p)}}
\Bigr\} .
\label{app:D1 four terms SL}
\end{multline}
The four terms come respectively from $-ZW$, $+WZ$, $+ZW$ and $-WZ$ in the two possible contractions of the commutator. We now identify the permutation associated with each term.

\paragraph{1) First term.}

We begin with the non-planar case $p\neq g(q)$. The delta function sets $i_p=i_{g(q)}$. We eliminate $i_p$ and define a new set of summation indices by
\begin{equation}
j_p=i_q , \qquad j_q=\nu , \qquad
j_a=i_a \quad \bigl(a\notin\{p,q\}\bigr) .
\label{app:D1 relabel first}
\end{equation}
The two inserted fields are then placed at the $p$-th and $q$-th positions,
\begin{equation}
Z_{i_q\nu}=Z_{j_pj_q} , \qquad
W_{\nu i_{g(p)}}=W_{j_qj_{h_1(q)}} .
\end{equation}
Together with
\begin{equation}
(X_a)_{i_ai_{g(a)}}=(X_a)_{j_aj_{h_1(a)}}
\qquad \bigl(a\notin\{p,q\}\bigr) ,
\end{equation}
after the substitution $i_p=i_{g(q)}$, these equations uniquely determine
\begin{equation}
h_1=(p,q)\,g\,(p,g(q))\,(p,q) .
\label{app:D1 h1}
\end{equation}
Indeed, $h_1(p)=q$ and $j_{h_1(q)}=i_{g(p)}$, which verify the two inserted factors above; its action on every $a\notin\{p,q\}$ verifies the remaining factors.
Thus, without including its minus sign in \eqref{app:D1 four terms SL}, the first term is
\begin{equation}
    \swmap^{(N)}_{m,n}\Bigl((p,q)\,g\,(p,g(q))\,(p,q)\Bigr) .
\end{equation}

Consider next the planar case $p=g(q)$. The delta function becomes $\sum_{i_p} \delta_{i_pi_p} = N$. The relabelling \eqref{app:D1 relabel first} then gives
\begin{equation}
    N\,\swmap^{(N)}_{m,n}\Bigl((p,q)g(p,q)\Bigr) .
\label{app:D1 first planar}
\end{equation}

\paragraph{2) Second term.}

We first consider $p\neq g(q)$. The same delta function again sets $i_p=i_{g(q)}$, but the order of the two inserted fields is reversed. We now relabel
\begin{equation}
j_p=\nu , \qquad j_q=i_q , \qquad
j_a=i_a \quad \bigl(a\notin\{p,q\}\bigr) .
\label{app:D1 relabel second}
\end{equation}
The inserted fields satisfy
\begin{equation}
Z_{\nu i_{g(p)}}=Z_{j_pj_{h_2(p)}} , \qquad
W_{i_q\nu}=W_{j_qj_p} .
\end{equation}
Matching the remaining factors in the same way gives
\begin{equation}
h_2=g\,(p,g(q)) .
\label{app:D1 h2}
\end{equation}
In particular, $h_2(q)=p$ and $j_{h_2(p)}=i_{g(p)}$, as required by the two inserted factors.
The second term is therefore
\begin{equation}
    \swmap^{(N)}_{m,n}\Bigl(g\,(p,g(q))\Bigr) .
\end{equation}

For the planar case $p=g(q)$, the sum over the free index $i_p$ again gives $N$. Using \eqref{app:D1 relabel second}, the second term becomes
\begin{equation}
    N\,\swmap^{(N)}_{m,n}(g) .
\label{app:D1 second planar}
\end{equation}
Combining the first two terms with their signs in \eqref{app:D1 four terms SL}, their non-planar and planar results can be written together as
\begin{equation}
    \swmap^{(N)}_{m,n}\Bigl(g\,\spar{p,g(q)}
-(p,q)\,g\,\spar{p,g(q)}\,(p,q)
\Bigr) 
\label{app:D1 p double bracket}
\end{equation}
where the double bracket is defined by
\begin{equation}
\spar{i,j} \equiv N \, \delta_{ij} + (i,j) \bigl(1-\delta_{ij}\bigr) .
\label{def:double bracket}
\end{equation}

\paragraph{3) Third term.}

We begin with the non-planar case $q\neq g(p)$. In this case the delta function sets $i_q=i_{g(p)}$. We eliminate $i_q$ and relabel
\begin{equation}
j_p=i_p , \qquad j_q=\nu , \qquad
j_a=i_a \quad \bigl(a\notin\{p,q\}\bigr) .
\label{app:D1 relabel third}
\end{equation}
The two inserted fields become
\begin{equation}
Z_{i_p\nu}=Z_{j_pj_q} , \qquad
W_{\nu i_{g(q)}}=W_{j_qj_{h_3(q)}} .
\end{equation}
Matching all the factors gives
\begin{equation}
h_3=g\,(q,g(p)) .
\label{app:D1 h3}
\end{equation}
Here $h_3(p)=q$ and $j_{h_3(q)}=i_{g(q)}$.
Hence, the third term is
\begin{equation}
    \swmap^{(N)}_{m,n}\Bigl(g\,(q,g(p))\Bigr) .
\end{equation}

For the planar case $q=g(p)$, the delta function becomes $\delta_{i_qi_q}$. Since $p$ is the unique preimage of $q$ under $g$, the index $i_q$ appears nowhere else in the third term and its sum gives $N$. The relabelling \eqref{app:D1 relabel third} gives
\begin{equation}
    N\,\swmap^{(N)}_{m,n}(g) .
\label{app:D1 third planar}
\end{equation}

\paragraph{4) Fourth term.}

We first consider $q\neq g(p)$. The delta function is the same as in the third term, while the order of the inserted fields is reversed. We relabel
\begin{equation}
j_p=\nu , \qquad j_q=i_p , \qquad
j_a=i_a \quad \bigl(a\notin\{p,q\}\bigr) .
\label{app:D1 relabel fourth}
\end{equation}
Now
\begin{equation}
Z_{\nu i_{g(q)}}=Z_{j_pj_{h_4(p)}} , \qquad
W_{i_p\nu}=W_{j_qj_p} .
\end{equation}
The solution of the index-matching equations is
\begin{equation}
h_4=(p,q)\,g\,(q,g(p))\,(p,q) .
\label{app:D1 h4}
\end{equation}
Indeed, $h_4(q)=p$ and $j_{h_4(p)}=i_{g(q)}$.
The fourth term, without its minus sign, is therefore
\begin{equation}
    \swmap^{(N)}_{m,n}\Bigl((p,q)\,g\,(q,g(p))\,(p,q)\Bigr) .
\end{equation}

For the planar case $q=g(p)$, the sum over the free index $i_q$ gives $N$. The relabelling \eqref{app:D1 relabel fourth} gives
\begin{equation}
    N\,\swmap^{(N)}_{m,n}\Bigl((p,q)g(p,q)\Bigr) .
\label{app:D1 fourth planar}
\end{equation}
Combining the third and fourth terms with their signs, their non-planar and planar results are
\begin{equation}
    \swmap^{(N)}_{m,n}\Bigl(
g\,\spar{q,g(p)}
-(p,q)\,g\,\spar{q,g(p)}\,(p,q)
\Bigr) .
\label{app:D1 q double bracket}
\end{equation}

\paragraph{5) Summary.}

For each $(p,q)$, the first two terms give \eqref{app:D1 p double bracket}, while the last two terms independently give \eqref{app:D1 q double bracket}. By summing over $p, q$, we obtain
\begin{equation}
\begin{aligned}
    N\, \fD^{(1)}_N \cdot \swmap^{(N)}_{m,n}(g)
&=\sum_{p=1}^m\sum_{q=m+1}^L
\swmap^{(N)}_{m,n}\Bigl(
g\,\spar{p,g(q)}-(p,q)g\,\spar{p,g(q)}(p,q)
\\
&\hspace{50mm}
+g\,\spar{q,g(p)}-(p,q)g\,\spar{q,g(p)}(p,q) \Bigr) 
\end{aligned}
\label{app:D1 pq double bracket}
\end{equation}
We can rewrite this equation as
\begin{equation}
    N\, \fD^{(1)}_N \cdot \swmap^{(N)}_{m,n}(g) = \swmap^{(N)}_{m,n}\bigl( {\sf H} (g)\bigr) , \qquad
    {\sf H} (g) = \sum_{i,j=1}^L \Bigl\{ g\,\spar{i,g(j)} - (i,j)\,g\,\spar{i,g(j)}\,(i,j) \Bigr\} 
\label{app:D1 induced SL}
\end{equation}
in agreement with \cite{Bellucci:2004ru,Pasukonis:2010rv}.

This equation \eqref{app:D1 induced SL} is derived as follows.
The first line of \eqref{app:D1 pq double bracket} corresponds to $(i,j)=(p,q)$, while the second corresponds to $(i,j)=(q,p)$. 
If both $i$ and $j$ are in the $Z$ position, or in the $W$ position, the two terms in \eqref{app:D1 induced SL} are conjugate under $S_m\times S_n$\,, and the sum vanishes after application of $\swmap^{(N)}_{m,n}$ or $\orbitP{m,n}$.
For example
\begin{equation}
    \sum_{i,j=1}^m \orbitP{m,n}\Big(g\,\spar{i,g(j)}\Big) - \orbitP{m,n}\Big((i,j)\,g\,\spar{i,g(j)}\,(i,j)\Big) =\sum_{i,j=1}^m \orbitP{m,n}\Big(g\,\spar{i,g(j)}\Big) - \orbitP{m,n}\Big(g\,\spar{i,g(j)}\Big) = 0
\end{equation}
and similarly when $i,j \in \{m+1,\dots,L\}$.
Thus, the sum over $p,q$ extends to the sums over $i,j\in\{1,\dots,L\}$.

The extension from $(p,q)$ to $(i,j)$ makes the induced action \eqref{app:D1 induced SL} manifestly covariant under $S_L$-conjugation. 
Indeed, for any $\gamma\in S_L$, the double bracket satisfies\footnote{Our convention is $(g_1 \, g_2 ) \, (i) = g_2 (g_1(i))$ as in \cite{Padellaro:2024rld}.}
\begin{equation}
\gamma^{-1} \,\spar{i,g(j)}\,\gamma
=\spar{\gamma(i),(\gamma^{-1} g\gamma)(\gamma(j))} .
\end{equation}
Conjugating \eqref{app:D1 induced SL} by $\gamma$ and relabelling the two summation indices therefore gives
\begin{equation}
{\sf H}(\gamma^{-1} g\gamma)=\gamma^{-1} {\sf H}(g)\,\gamma , \qquad \gamma\in S_L 
\label{app:D1 SL covariance}
\end{equation}
which proves \eqref{eq:H1-conjugation} in the main text.

This conjugation property guarantees that $\fD^{(1)}_N$ commutes with the global $\su2$ Casimir operator through its relation to the adjoint $T_2$ action \eqref{eq:adjoint-T2}.
We extend ${\sf H}$ linearly to $\mathbb C[S_L]$ and sum \eqref{app:D1 SL covariance} over transpositions to obtain
\begin{equation}
{\sf H}\!\left(m^{\mathfrak{ad}}[T_2](\sigma)\right) = \sum_{i<j}(i,j)\,{\sf H}(\sigma)\,(i,j)
= m^{\mathfrak{ad}}[T_2]\!\left({\sf H}(\sigma)\right).
\label{app:D1 adjoint commutation}
\end{equation}
Thus ${\sf H}$ commutes with $m^{\mathfrak{ad}}[T_2]$ on $\mathbb C[S_L]$, and hence on $\cA(m,n)$.
Together with \eqref{app:D1 induced SL}, this establishes the claimed commutation on $\mathcal H_N(m,n)$.

Finally, expanding the double bracket in \eqref{app:D1 induced SL} using \eqref{def:double bracket}, we obtain
\begin{equation}
{\sf H} (g) = N \, \sum_{\substack{i=1\\ i\ne g(i)}}^L 
\Big\{ g-(i,g(i))\,g\,(i,g(i)) \Big\} 
+ \sum_{i=1}^L\sum_{\substack{j=1\\ g(j)\ne i}}^L 
\Big\{ g\,(i,g(j))-(i,j)\,g\,(i,g(j))\,(i,j) \Big\}
\label{app:eq:D1-orbit-basis}
\end{equation}
which agrees with the intermediate map \eqref{eq:D1-orbit-basis} in Definition~\ref{def: 1LPDO} through
\begin{equation}
\left.\bar H^{(1)}_{\gss}(g)\right|_{\gss=1/N} = \frac1N\,{\sf H}(g).
\end{equation}

\section{Trace derivatives and eigenspace factorisation}\label{app:trace factorisation}
In this appendix we prove several derivative identities underlying the eigenspace factorisation in Section \ref{sec:eigenspace factorisation}.
We will show that multiplication by $z, w$ and the action of $\pdv{z},\pdv{w}$ commute with the action of the one-loop dilatation operator.
This involves re-writing the one-loop dilatation operator in terms of derivatives of traceless matrices.

The decomposition of traceful matrices defined in \eqref{eq:traceless-variable} gives a corresponding decomposition of matrix derivatives.
For $(X,x)\in\{(Z,z),(W,w)\}$, we first compute the trace derivative at fixed $X'$.
Since varying $x$ shifts all diagonal entries of $X$ equally, this derivative takes the form
\begin{equation}
    \pdv{}{x} = c \sum_{a=1}^N \pdv{}{X_{aa}}.
\end{equation}
From $x=\tr(X)/N$ and the condition $\pdv{}{x}x=1$, we obtain $c=1$, giving
\begin{equation}
    \label{eq:sl-deriv}
    \pdv{}{x} = \sum_{a=1}^N \pdv{}{X_{aa}}.
\end{equation}
Next we define the traceless derivative by subtracting the trace component from $\pdv{}{X_{ij}}$:
\begin{equation}
\projD^{X'}_{ji} = \pdv{}{X_{ij}} - c'\,\frac{\delta_{ij}}{N} \sum_{a=1}^N \pdv{}{X_{aa}},
\end{equation}
with an unknown constant $c'$. This derivative satisfies
\begin{alignat}{9}
\projD^{X'}_{ji} (x)
&= \left( \pdv{}{X_{ij}} - c'\,\frac{\delta_{ij}}{N} \sum_{a=1}^N \pdv{}{X_{aa}} \right) \frac{\tr(X)}{N} 
& &= \frac{1}{N} \, \delta_{ij}(1-c')
\label{eq: traceless D on trace} \\
\projD^{X'}_{ji} ( X'_{kl} ) 
&= \Big( \pdv{}{X_{ij}} - c' \, \frac{\delta_{ij}}{N}  \sum_{a=1}^N \pdv{}{X_{aa}} \Big) \Big( X_{kl} - \delta_{kl} \, \frac{\tr(X)}{N} \Big)
& &= \delta_{ik}\delta_{jl}-\frac{1}{N}\delta_{ij}\delta_{kl} \,.
\label{eq: traceless D on traceless} 
\end{alignat}
By requiring that the derivative preserves the surface $\tr(X)=0$, we obtain $c'=1$:
\begin{equation}
\label{eq: Lprime deriv}
\projD^{X'}_{ji} \equiv \pdv{}{X_{ij}} - \frac{\delta_{ij}}{N}\pdv{}{x}
\end{equation}
which we call the projected traceless derivative.

Using the projected traceless derivative, we can formulate a rewriting lemma for the one-loop dilatation operator.
\begin{lemma}\label{lemma:oneloop factorisation}
The derivatives in the one-loop dilatation operator on $\mathcal{H}_N$ can be replaced with projected traceless derivatives
\begin{equation}
\fD^{(1)}_N
= -\frac{1}{N}:\tr \qty(\comm{Z}{W} \comm{\pdv{}{Z}}{\pdv{}{W}}): \ \ 
= -\frac{1}{N}:\tr \qty(\comm{Z'}{W'} \comm{ \projD^{Z'} }{ \projD^{W'} }): \,.
\end{equation}
\end{lemma}
\begin{proof}
    Consider the first commutator
    \begin{equation}
        \comm{Z}{W}_{ij} = \comm{(Z'+zI)}{(W'+wI)}_{ij} = \comm{Z'}{W'}_{ij}
    \end{equation}
    the last equality uses the fact that the identity matrix commutes with all matrices.
    For the second commutator
\begin{equation}
\comm{\pdv{}{Z}}{\pdv{}{W}}_{ij} 
= \comm{ \Big( \projD^{Z'} + \frac{I}{N}\pdv{}{z} \Big) }{ \Big( \projD^{W'} + \frac{I}{N}\pdv{}{w} \Big)}_{ij}
\end{equation}
    all but the first term vanish: 
\begin{align}
\comm{ \projD^{Z'} }{\frac{I}{N}\pdv{w}}_{ij}
&=\sum_{k=1}^N \qty( \projD^{Z'}_{ik} \, \frac{\delta_{kj}}{N} \pdv{}{w} - \frac{\delta_{ik}}{N}\pdv{}{w} \projD^{Z'}_{kj} ) 
= \frac{1}{N} \qty( \projD^{Z'}_{ij} \pdv{}{w} - \pdv{}{w} \projD^{Z'}_{ij} ) =0 , \\
\comm{\frac{I}{N}\pdv{z}}{ \projD^{W'} }_{ij}
&=\sum_{k=1}^N \qty(\frac{\delta_{ik}}{N} \pdv{}{z} \projD^{W'}_{kj} - \projD^{W'}_{ik} \frac{\delta_{kj}}{N}\pdv{}{z}) 
= \frac{1}{N} \qty(\pdv{z} \projD^{W'}_{ij} - \projD^{W'}_{ij} \pdv{}{z} ) =0 , \\
\comm{\frac{I}{N}\pdv{z}}{\frac{I}{N}\pdv{w}}_{ij}
&=\sum_{k=1}^N \qty(\frac{\delta_{ik}}{N} \pdv{z} \frac{\delta_{kj}}{N} \pdv{w} -\frac{\delta_{ik}}{N} \pdv{w} \frac{\delta_{kj}}{N} \pdv{z} ) 
= \frac{1}{N^2} \delta_{ij} \qty(\pdv{z}\pdv{w} - \pdv{w}\pdv{z}) = 0.
\end{align}
\end{proof}

Thus $\fD_N^{(1)}$ is $\fD_N^{\prime(1)}$ acting on the traceless variables, with the trace variables $z,w$ treated
as spectators. Explicitly, for $f\in\mathcal H'_N$ and $q\in\mathbb C[z,w]$, we have
\begin{equation}
\fD_N^{(1)} \bigl[f(Z',W')\,q(z,w)\bigr] = \bigl[ \fD_N^{\prime(1)}f(Z',W')\bigr]\,q(z,w),
\end{equation}
where the substitution maps in Definition \ref{def:substitution maps} are understood.

\begin{definition}[Single-letter-trace counting operators]
\label{def:single-letter-trace-counting}
Define the single-letter-trace counting operators by
\begin{equation}
\hN{(s)} \equiv \hN{(s)}_z + \hN{(s)}_w , \qquad
\hN{(s)}_x \equiv \frac{1}{N} \sum_{a,b=1}^N X_{aa} \, \frac{\partial}{\partial X_{bb}} = x \frac{\partial}{\partial x} \,.
\label{def:slcnt}
\end{equation}
The equality follows from \eqref{def:sl-trace} and \eqref{eq:sl-deriv}. 
The operators $\hN{(s)}_z$ and $\hN{(s)}_w$ count the degrees in $z$ and $w$ at fixed $Z'$ and $W'$, respectively.
\end{definition}

Using the lemma above and two facts about the traceless and trace derivatives, we can prove an important commutativity lemma.
\begin{lemma}
    \label{lemma: weyl dil commute}
    The one-loop dilatation operator commutes with the operators $x, \pdv{x}$ and $\hN{(s)}_{x}$ 
\begin{equation}
    \comm{ \fD^{(1)}_{N}}{\hN{(s)}_x} =
    \comm{ \fD^{(1)}_{N}}{x} =
    \comm{ \fD^{(1)}_{N}}{\pdv{x}} =
     0 \,,
\end{equation}
where $x \in \{z,w\}$.
\end{lemma}
\begin{proof}
Since $\poly(Z,W) \in \cH_N$ has a finite expansion into elements of $\mathcal{H}'_N \otimes \mathbb{C}[z,w]$,
\begin{equation}
    \label{eq: exp Fe}
    \poly(Z,W) = \sum_{a,b = 0}^\infty \poly_{a,b}(Z',W')z^a w^b,
\end{equation}
it suffices to prove this for polynomials of the form $\poly'(Z',W')z^a w^b$. 
The substitution maps are omitted below to simplify the notation.

First, for $(X,x) \in \{(Z,z),(W,w)\}$, the traceless derivatives satisfy \eqref{eq: traceless D on trace}:
\begin{equation}
   \label{eq:traceless deriv orthogonality}
\projD^{X'}_{ji} ( x ) =  0 \, .
\end{equation}
It follows that for $\poly'(Z',W') \in \cH_N'$
\begin{equation}
   \projD^{X'}_{ji}  \qty( z^{a} w^{b}\poly'(Z',W') ) =
   z^{a}w^{b}
   \projD^{X'}_{ji} \Big( \poly'(Z',W') \Big),  \quad \forall a,b \in \mathbb{Z}_{\geq 0} \,
\end{equation}
and using Lemma \ref{lemma:oneloop factorisation}
\begin{equation}
    \comm{ \fD^{(1)}_{N}}{x} =  0 \,.
\end{equation}
Secondly, using 
\begin{equation}
\label{eq: trace deriv orthogonality}
\pdv{}{x} X'_{ij} = \sum_{k=1}^N \pdv{}{X_{kk}}   \qty(X_{ij} - \delta_{ij}x) = \sum_{k=1}^N \qty(\delta_{ki} \delta_{kj} - \frac{\delta_{ij}}{N}) = 0
\end{equation}
we have
\begin{equation}
\begin{aligned}
 \fD^{(1)}_N \pdv[c]{z} \pdv[d]{w} \poly'(Z',W')z^a w^b 
&=\Big\{  \fD^{\prime(1)}_N \poly'(Z',W') \Big\} \, \Big\{ \pdv[c]{z} \pdv[d]{w} z^{a} w^b \Big\} \\
&=\pdv[c]{z} \pdv[d]{w} \,  \fD^{(1)}_N \poly'(Z',W')z^a w^b  
\quad \forall c,d \in \mathbb{Z}_{\geq 0}.
\end{aligned}
\end{equation}
Combining the two results we get
\begin{equation}
    \comm{ \fD^{(1)}_{N}}{\hN{(s)}_x} = 0 \,.
\end{equation}
\end{proof}

\begin{lemma}
    \label{lemma: coefficient eigenspace factorisation}
    Let $\poly^{(e)}(Z,W) \in \mathcal{H}_N$ be a non-zero polynomial with expansion
\begin{equation}
\poly^{(e)}(Z,W) = \sum_{a,b = 0}^\infty \poly^{(e)}_{a,b}(Z',W')z^a w^b.
\label{eq: exp Fe2}
\end{equation}
This is an eigenvector with eigenvalue $e$ if and only if every nonzero coefficient $\poly^{(e)}_{a,b}(Z',W')$ is an eigenvector with eigenvalue $e$.
    \begin{equation}
        \fD^{(1)}_N \poly^{(e)}(Z,W) =e \poly^{(e)}(Z,W) 
        \Longleftrightarrow 
        \fD^{\prime(1)}_N \poly^{(e)}_{a,b}(Z',W') =e \poly^{(e)}_{a,b}(Z',W') , \quad \forall a,b \geq 0.
    \end{equation}
\end{lemma}
\begin{proof}
The easy direction assumes $\poly^{(e)}_{a,b}(Z',W') \in \mathcal{H}'_N(e)$ are eigenvectors of the dilatation operator with eigenvalue $e$
\begin{equation}
    \fD^{\prime(1)}_N \poly_{a,b}^{(e)}(Z',W') = e \poly_{a,b}^{(e)}(Z',W'), \quad \forall a,b \geq 0.
\end{equation}
Lemma \ref{lemma: weyl dil commute} gives 
\begin{equation}
    \fD^{(1)}_N \qty(z^{a}w^{b}\poly_{a,b}^{(e)}(Z',W')) 
    = z^{a}w^{b} \,  \fD^{\prime(1)}_N \poly_{a,b}^{(e)}(Z',W')
    = e z^{a}w^{b}\poly^{(e)}_{a,b}(Z',W').
\end{equation}
By linearity, we get that $\poly^{(e)}(Z,W)$ is an eigenvector with eigenvalue $e$.

For the other direction, we assume $\poly^{(e)}(Z,W) \in \mathcal{H}_N(e)$ and deduce that $\poly^{(e)}_{a,b}(Z',W') \in \mathcal{H}'_N(e)$.
For this, note that the single-letter-trace counting operators 
have non-negative integer eigenvalues as operators on $\mathcal{H}_N$
\begin{equation}
    \mathrm{Spectrum}(\hN{(s)}_x) = \{0,1,2,\dots \}.
\end{equation}
They can be used to construct maps extracting the coefficients $\poly^{(e)}_{a,b}(Z',W')$.
By assumption, $\poly^{(e)}(Z,W)$ has bounded degree and so $\poly^{(e)}_{a,b}(Z',W')$ all vanish for some finite upper limits $a > A, b > B$.
Using this, we define
\begin{equation}
\coeffN_z(a) = \frac{1}{a!}\pdv[a]{z}\prod_{\substack{c=0\\c \neq a}}^A \frac{\hN{(s)}_z-c}{a-c}, \qquad
\coeffN_w(b) =\frac{1}{b!}\pdv[b]{w} \prod_{\substack{d=0\\d \neq b}}^B \frac{\hN{(s)}_w-d}{b-d} \,.
\end{equation}
They satisfy 
\begin{equation}
\coeffN_z(a) \, \coeffN_w(b)\qty(\poly^{(e)}(Z,W)) = \poly^{(e)}_{a,b}(Z',W') \,,
\end{equation}
and commute with the dilatation operator since they are polynomials in $\hN{(s)}_x, \pdv{x}$.
A short computation gives 
\begin{equation}
    \begin{aligned}
     \fD^{\prime(1)}_N \poly^{(e)}_{a,b} (Z',W')
    &= \fD^{(1)}_N \qty[\coeffN_z(a) \, \coeffN_w(b) \, \poly^{(e)}(Z,W) ] \\
    &= \coeffN_z(a) \, \coeffN_w(b) \, \fD^{(1)}_N \poly^{(e)}(Z,W) \\
    &= \coeffN_z(a) \, \coeffN_w(b) \, e\poly^{(e)}(Z,W) \\
    &= e \poly^{(e)}_{a,b}(Z',W').
    \end{aligned}
\end{equation}
\end{proof}

From any non-zero eigenvector with the expansion \eqref{eq: exp Fe2}, we can construct a non-zero eigenvector with the same eigenvalue that depends only on the traceless variables $Z',W'$.
Choose $a_*$ to be the smallest value of $a$ for which $\poly^{(e)}_{a,b}(Z',W')\neq0$ for some $b$. Differentiating $a_*$ times with respect to $z$ and setting $z=0$ gives
\begin{equation}
\left.\frac{1}{a_*!}\,\partial_z^{a_*}\poly^{(e)}(Z,W)\right|_{z=0}=\sum_{b\geq0}\poly^{(e)}_{a_*,b}(Z',W')w^b\neq0.
\end{equation}
With $a_*$ fixed, choose $b_*$ to be the smallest value of $b$ for which $\poly^{(e)}_{a_*,b}(Z',W')\neq0$. Following a similar step, we find
\begin{equation}
\left.\frac{1}{a_*!b_*!}\,\partial_z^{a_*}\partial_w^{b_*}\poly^{(e)}(Z,W)\right|_{z=w=0}=\poly^{(e)}_{a_*,b_*}(Z',W')\neq0.
\end{equation}
By Lemma \ref{lemma: weyl dil commute}, these operations preserve the eigenvalue. The resulting polynomial is therefore a non-zero traceless eigenvector with eigenvalue $e$, providing a non-zero lowest component of a tower.

\begin{remark}[Permutation description of the single-letter-trace counting operator]
\label{app:single-letter-trace SL}

The single-letter-trace counting operators \eqref{def:slcnt} have a permutation description. Their action on the $U(N)$-invariant operators $\swmap^{(N)}_{m,n}(g)$ for $g \in S_L$ can be computed in a similar way to Appendix \ref{app:induced action D1} as
\begin{equation}
\hN{(s)}_z \cdot \swmap^{(N)}_{m,n}(g) = \frac{1}{N} \sum_{p=1}^m \swmap^{(N)}_{m,n} \bigl(g\,\spar{p,g(p)}\bigr) , \qquad
\hN{(s)}_w \cdot \swmap^{(N)}_{m,n}(g) = \frac{1}{N} \sum_{q=m+1}^{L} \swmap^{(N)}_{m,n} \bigl(g\,\spar{q,g(q)}\bigr) .
\end{equation}
These identities extend linearly to $\mathbb{C}[S_L]$. 
The total single-letter-trace counting operator satisfies the following intertwining relation similar to \eqref{eq:intertwine},
\begin{equation}
\hN{(s)} \cdot \swmap^{(N)}_{m,n}(g) = \swmap^{(N)}_{m,n} \Bigl( \cnN{(s)}_{\gss} (g) \Bigr) \Big|_{\gss = 1/N} \,, \qquad
\cnN{(s)}_{\gss} (g) = 
{\sf Fix} (g) \, g + \xi \sum_{\substack{i=1\\ i\ne g(i)}}^L g (i,g(i)) 
\label{app:Ns induced SL}
\end{equation}
where ${\sf Fix} (g)$ counts the number of fixed points in $g$.

On a polynomial $\poly'(Z',W')z^a w^b$, the operator $\hN{(s)}$ has eigenvalue $(a+b)$. 
The intertwining relation \eqref{app:Ns induced SL} implies that the permutation operator $\cnN{(s)}_{\gss}$ preserves $\ker\swmap^{(N)}_{m,n}$ at $\xi=1/N$.
The block-triangular argument used in Lemma \ref{cor:finite-N factorisation} gives
\begin{equation}
\left.\det_{\cA(m,n)}(x-\cnN{(s)}_{\gss})\right|_{\gss=1/N}=Q^{(s,N)}_{m,n}(x)\det_{\mathcal{H}_N(m,n)}(x-\hN{(s)}),
\end{equation}
where $Q^{(s,N)}_{m,n}(x)$ is the characteristic polynomial in the null sector.
\end{remark}

\section{Background on computational checks}
Here we give some background on the computational aspects of checking the proposal in Subsection  \ref{subsec:derangement hilb}.

\subsection{Algorithm for intersection of vector subspaces}
\label{apx: null space algo}
A basis for the intersection of two subspaces, which was used in the first computational check, can be computed using standard linear algebra; see e.g. \cite[Section 12.4]{golub1996} for the numerical treatment and \cite[{Chapter 2}]{Cohen1993} for exact (integer) arithmetic.

Let $U,V \subseteq W$ be subspaces of $W$ with bases
\begin{equation}
    W = \Span(w_a \, \vert \, 1 \leq a \leq A), \quad 
    U = \Span(u_i \, \vert \, 1 \leq i \leq I), \quad 
    V = \Span(v_p \, \vert \, 1 \leq p \leq P) \,.
\end{equation}
We introduce matrices $U,V$ for the expansion of $u_i, v_p$ in $w_a$
\begin{equation}
    \label{eq:subspace bases}
    u_i = \sum_{a=1}^A U_{ai} w_a, \quad 
    v_p = \sum_{a=1}^A V_{ap} w_a \,.
\end{equation}
By definition, elements $y \in U \cap V$ are the elements of $W$ that can be written as
\begin{equation}
    y = \sum_{i=1}^I Y^{(U)}_{i}u_i \text{ and } y = \sum_{p=1}^P Y^{(V)}_p v_p,
\end{equation}
for coefficient vectors $Y^{(U)} \in \mathbb{C}^{I}$ and $Y^{(V)} \in \mathbb{C}^{P}$.
Taking the difference of these two equalities gives
\begin{equation}
    \sum_{i=1}^I Y^{(U)}_{i}u_i - \sum_{p=1}^P Y^{(V)}_p v_p = 0
\end{equation}
or using \eqref{eq:subspace bases} and extracting the coefficients of $w_a$
\begin{equation}
    \label{eq:intsec eq}
    \sum_{i=1}^I Y^{(U)}_{i}U_{ai} - \sum_{p=1}^P Y^{(V)}_p V_{ap} = 0\,. 
\end{equation}
This is an equation for the nullspace of a matrix:
Define the vector $Y \in \mathbb{C}^{I+P}$
\begin{equation}
    Y=(Y^{(U)}_1,\dots,Y^{(U)}_I,Y^{(V)}_1,\dots,Y^{(V)}_P)^t
\end{equation}
and matrix
\begin{equation}
    K = 
    \qty(
    \begin{matrix}
        U_{11} & \dots & U_{1I} & - V_{11} & \dots & -V_{1P} \\
        \vdots & & & & & \vdots \\ 
        U_{A1} & \dots & U_{AI} & - V_{A1} & \dots & -V_{AP}
    \end{matrix}
    )
\end{equation}
then equation \eqref{eq:intsec eq} reads
\begin{equation}
    KY = 0.
\label{app:nullspace K}
\end{equation}
In other words, elements of the intersection $U \cap V$ correspond to elements of the nullspace of $K$.
Conversely, if $Y$ is a solution to the above equation, then the first $I$ entries $Y^{(U)}_i$ correspond to the coordinates in the $u_i$ basis and the last $P$ entries $Y^{(V)}_p$ correspond to the coordinates in the $v_p$ basis.
The coordinates $Y^{(W)}_a$ in the basis $w_a$ can be deduced from either of them using the $U$ or $V$ matrix:
\begin{equation}
    Y^{(W)}_a = \sum_{i=1}^I U_{ai} Y^{(U)}_{i} = \sum_{p=1}^P V_{ap} Y^{(V)}_{p}.
\end{equation}

For the intersection $\cD_N(m,n)$, take $U=\cA_N(m,n)$, $V=\cD(m,n)$ and $W=\cA(m,n)$. 
The coordinate matrices $B_N$ and ${B}_{\dos}$ are computed from a basis of $\ker K$ as follows.
Let $Y_q=(Y^{(U)}_{1q},\dots,Y^{(U)}_{Iq},Y^{(V)}_{1q},\dots,Y^{(V)}_{Pq})^t$ be a basis for $\ker K$, then
\begin{equation}
    (B_N)_{i q} = Y^{(U)}_{iq} \,,\qquad 
    (B_{\dos})_{p q} = Y^{(V)}_{pq} \,, \qquad
    q \in \{1,\dots,\dim \ker K\}.
\end{equation}
The ambient coordinate matrix is then $B=UB_N=V{B}_{\dos}$, with the notation of \eqref{def:BtB}.
For our construction of $\cD_N(m,n)$, we choose bases whose ambient coordinate matrices $U_{ai}$ and $V_{ap}$ have integer entries, so that $K$ is an integer matrix.
The equation \eqref{app:nullspace K} can be solved by using integer linear algebra, as reviewed in \cite{Padellaro:2024rld}.

\subsection{Gram matrix for derangements}\label{app:Gram diagram}

The Gram matrix $\mathcal{G}_{SU(N)}(m,n)$ in \eqref{def:SUN Gram bb}, namely the free two-point functions of $SU(N)$ theory, can be computed as a finite expansion subtracting the single-letter-trace contributions with factors of $1/N$, and the coefficient at each order represented as a sum over diagrams \cite{Corley:2002mj}. For a single complex matrix in the half-BPS sector, this expansion was organised there as a sum over the number of $1/N$ insertions, weighted by a binomial coefficient counting the contractions taken from the trace part of the propagator. With two matrix species the contractions are no longer exchangeable, and the binomial is replaced by a sum over subsets. By refining the argument in this way we obtain a closed formula for $\mathcal{G}_{SU(N)}(m,n)$ for a pair of derangements.

\begin{proposition}
    Let $g_i,g_j \in D_{m+n}$ be a pair of derangements. Define the derangement $g_i \otimes g_j \in D_{2m+2n}$ as
    \begin{equation}
        g_i \otimes g_j = g_i \prod_{l=1}^{m+n} (l,l+m+n)g_j\prod_{l=1}^{m+n} (l,l+m+n)\,.
    \end{equation}
    In other words, it is the derangement where $g_i$ is viewed as a permutation of $\{1,\dots,m+n\}$ while $g_j$ is viewed as a permutation of the numbers $\{m+n+1,\dots, 2m+2n\}$
    The number of cycles in a permutation $g$ is denoted $\mathrm{cycles}(g)$.
    The Wick contractions in the inner product are encoded in a sum over $S_m \times S_n$ and a sum over subsets $S \subseteq \{1,\dots,m+n\}$ which correspond to the choice of the first or second term in \eqref{eq:suN wick} for each contraction,
    \begin{equation}
        \begin{aligned}
            \langle \swmap'^{(N)}_{m,n}(g_i)(Z',W') &\swmap'^{(N)}_{m,n}(g_j^{-1})(Z'^\dagger,W'^\dagger) \rangle_{SU(N)} = \\
            &\sum_{\substack{h \in S_m \times S_n \\
            S \subseteq \{1,\dots,m+n\}}} (-N)^{\vert S \vert-m-n}
            N^{\mathrm{cycles}([\prod_{l \in S} (l,m+n+l)][(h g_i h^{-1}) \otimes g_j^{-1}])} 
        \end{aligned}
    \label{app:SUN-2pt der}
    \end{equation}
which generalises the free two-point function of the $U(N)$ gauge theory \eqref{eq: U(N) inner prod}.
\end{proposition}

The formula \eqref{app:SUN-2pt der} gives the two-point functions of the $SU(N)$-invariant polynomials labelled by derangements $D_{m+n}$\,.
For the orbit basis of DOS, we use $\swmap'^{(N)}_{m,n}(P_{m,n}(g_i))=m!n!\,\swmap'^{(N)}_{m,n}(g_i)$. 
The Gram matrix in \eqref{def:SUN Gram bb} is obtained by multiplying the matrix ${B}_{\dos}$ in \eqref{def:BtB} as
\begin{equation}
    (\mathcal{G}_{SU(N)}(m,n))_{pq}=\sum_{i,j=1}^{r}(B_{\dos})_{ip}(B_{\dos})_{jq}\left\langle\swmap'^{(N)}_{m,n}( P_{m,n}(g_i) )(Z',W')\,\swmap'^{(N)}_{m,n}(\cS( P_{m,n}(g_j) ))(Z'^\dagger,W'^\dagger)\right\rangle_{SU(N)}.
\end{equation}

\begin{proof}
    We give a diagrammatic proof.
    The Wick contraction has a diagrammatic realisation as
\begin{equation}
\tikzset{every picture/.style={line width=0.75pt}} 
\langle
\vcenter{\hbox{
\begin{tikzpicture}[x=0.75pt,y=0.75pt,yscale=-1,xscale=1]
\draw    (130,160) -- (130,180) ;
\draw    (130,120) -- (130,140) ;
\draw    (180,160) -- (180,180) ;
\draw    (180,120) -- (180,140) ;
\draw   (110,140) -- (150,140) -- (150,160) -- (110,160) -- cycle ;
\draw   (160,140) -- (200,140) -- (200,160) -- (160,160) -- cycle ;
\draw (119,143) node [anchor=north west][inner sep=0.75pt]    {$Z'$};
\draw (167,142) node [anchor=north west][inner sep=0.75pt]    {$Z^{\prime \dagger }$};
\end{tikzpicture}
}}
\rangle  
=
\vcenter{\hbox{
\begin{tikzpicture}[x=0.75pt,y=0.75pt,yscale=-1,xscale=1]
\draw    (370,170) -- (420,190) ;
\draw    (420,170) -- (370,190) ;
\draw    (370,190) -- (370,210) ;
\draw    (370,150) -- (370,170) ;
\draw    (420,190) -- (420,210) ;
\draw    (420,150) -- (420,170) ;
\draw   (360,170) -- (430,170) -- (430,190) -- (360,190) -- cycle ;
\end{tikzpicture}
}}
- \frac{1}{N}
\vcenter{\hbox{
\begin{tikzpicture}[x=0.75pt,y=0.75pt,yscale=-1,xscale=1]
\draw    (280,250) -- (280,270) ;
\draw    (280,210) -- (280,250) ;
\draw    (330,250) -- (330,270) ;
\draw    (330,210) -- (330,250) ;
\draw   (270,230) -- (340,230) -- (340,250) -- (270,250) -- cycle ;
\end{tikzpicture}
}}
    \end{equation}
    The invariant $\swmap'^{(N)}_{m,n}(g)$ corresponds to the diagram with $m$ boxes labelled $Z'$ and $n$ boxes labelled $W'$ with strands permuted by $g$
\begin{equation}
    \swmap'^{(N)}_{m,n}(g) = 
\vcenter{\hbox{
\begin{tikzpicture}[x=0.75pt,y=0.75pt,yscale=-1,xscale=1]
\draw    (270,70) -- (270,90) ;
\draw    (270,30) -- (270,50) ;
\draw   (250,50) -- (290,50) -- (290,70) -- (250,70) -- cycle ;
\draw    (350,70) -- (350,90) ;
\draw    (350,30) -- (350,50) ;
\draw   (330,50) -- (370,50) -- (370,70) -- (330,70) -- cycle ;
\draw  [dash pattern={on 0.84pt off 2.51pt}]  (300,60) -- (320,60) ;
\draw    (400,70) -- (400,90) ;
\draw    (400,30) -- (400,50) ;
\draw   (380,50) -- (420,50) -- (420,70) -- (380,70) -- cycle ;
\draw    (480,70) -- (480,90) ;
\draw    (480,30) -- (480,50) ;
\draw   (460,50) -- (500,50) -- (500,70) -- (460,70) -- cycle ;
\draw  [dash pattern={on 0.84pt off 2.51pt}]  (430,60) -- (450,60) ;
\draw    (250,30) -- (500,30) ;
\draw   (260,90) -- (490,90) -- (490,120) -- (260,120) -- cycle ;
\draw    (480,140) -- (480,120) ;
\draw    (400,140) -- (400,120) ;
\draw    (350,140) -- (350,120) ;
\draw    (270,140) -- (270,120) ;
\draw    (500,140) -- (250,140) ;
\draw (259,52) node [anchor=north west][inner sep=0.75pt]    {$Z'$};
\draw (339,52) node [anchor=north west][inner sep=0.75pt]    {$Z'$};
\draw (389,52) node [anchor=north west][inner sep=0.75pt]    {$W'$};
\draw (469,52) node [anchor=north west][inner sep=0.75pt]    {$W'$};
\draw (365,100) node [anchor=north west][inner sep=0.75pt]    {$g$};
\end{tikzpicture}
}}
\end{equation}
The horizontal lines identify the top and bottom strands, representing a trace over an $N$-dimensional vector space on which $Z'$ and $W'$ act.
We evaluate the free two-point functions of $SU(N)$ theory by summing over all possible Wick contractions, which translates into the sum over $S_m \times S_n$ as
\begin{equation}
\begin{aligned}
        &\langle \swmap'^{(N)}_{m,n}(g_i)(Z',W') \swmap'^{(N)}_{m,n}(g_j^{-1})(Z'^\dagger,W'^\dagger) \rangle_{SU(N)} = \\[3mm]
        &\Big \langle
\vcenter{\hbox{
\vspace{-3.5em}
\begin{tikzpicture}[x=0.75pt,y=0.75pt,yscale=-1,xscale=1]
\draw    (270,70) -- (270,90) ;
\draw    (270,30) -- (270,50) ;
\draw   (250,50) -- (290,50) -- (290,70) -- (250,70) -- cycle ;
\draw    (350,70) -- (350,90) ;
\draw    (350,30) -- (350,50) ;
\draw   (330,50) -- (370,50) -- (370,70) -- (330,70) -- cycle ;
\draw  [dash pattern={on 0.84pt off 2.51pt}]  (300,60) -- (320,60) ;
\draw    (400,70) -- (400,90) ;
\draw    (400,30) -- (400,50) ;
\draw   (380,50) -- (420,50) -- (420,70) -- (380,70) -- cycle ;
\draw    (480,70) -- (480,90) ;
\draw    (480,30) -- (480,50) ;
\draw   (460,50) -- (500,50) -- (500,70) -- (460,70) -- cycle ;
\draw  [dash pattern={on 0.84pt off 2.51pt}]  (430,60) -- (450,60) ;
\draw    (250,30) -- (500,30) ;
\draw   (260,90) -- (490,90) -- (490,120) -- (260,120) -- cycle ;
\draw    (480,140) -- (480,120) ;
\draw    (400,140) -- (400,120) ;
\draw    (350,140) -- (350,120) ;
\draw    (270,140) -- (270,120) ;
\draw    (500,140) -- (250,140) ;
\draw (259,52) node [anchor=north west][inner sep=0.75pt]    {$Z'$};
\draw (339,52) node [anchor=north west][inner sep=0.75pt]    {$Z'$};
\draw (389,52) node [anchor=north west][inner sep=0.75pt]    {$W'$};
\draw (469,52) node [anchor=north west][inner sep=0.75pt]    {$W'$};
\draw (365,100) node [anchor=north west][inner sep=0.75pt]    {$g_i$};
\end{tikzpicture}
\begin{tikzpicture}[x=0.75pt,y=0.75pt,yscale=-1,xscale=1]
\draw    (270,70) -- (270,90) ;
\draw    (270,30) -- (270,50) ;
\draw   (250,50) -- (290,50) -- (290,70) -- (250,70) -- cycle ;
\draw    (350,70) -- (350,90) ;
\draw    (350,30) -- (350,50) ;
\draw   (330,50) -- (370,50) -- (370,70) -- (330,70) -- cycle ;
\draw  [dash pattern={on 0.84pt off 2.51pt}]  (300,60) -- (320,60) ;
\draw    (400,70) -- (400,90) ;
\draw    (400,30) -- (400,50) ;
\draw   (380,50) -- (420,50) -- (420,70) -- (380,70) -- cycle ;
\draw    (480,70) -- (480,90) ;
\draw    (480,30) -- (480,50) ;
\draw   (460,50) -- (500,50) -- (500,70) -- (460,70) -- cycle ;
\draw  [dash pattern={on 0.84pt off 2.51pt}]  (430,60) -- (450,60) ;
\draw    (250,30) -- (500,30) ;
\draw   (260,90) -- (490,90) -- (490,120) -- (260,120) -- cycle ;
\draw    (480,140) -- (480,120) ;
\draw    (400,140) -- (400,120) ;
\draw    (350,140) -- (350,120) ;
\draw    (270,140) -- (270,120) ;
\draw    (500,140) -- (250,140) ;
\draw (259,52) node [anchor=north west][inner sep=0.75pt]    {$Z'^\dagger$};
\draw (339,52) node [anchor=north west][inner sep=0.75pt]    {$Z'^\dagger$};
\draw (389,52) node [anchor=north west][inner sep=0.75pt]    {$W'^\dagger$};
\draw (469,52) node [anchor=north west][inner sep=0.75pt]    {$W'^\dagger$};
\draw (365,95) node [anchor=north west][inner sep=0.75pt]    {$g_j^{-1}$};
\end{tikzpicture}
}}
\Big \rangle = \\[4.5em]
&\sum_{h \in S_m \times S_n}
\vcenter{\hbox{
\begin{tikzpicture}[x=0.75pt,y=0.75pt,yscale=-1,xscale=1]
\draw (40,10)   -- (40,300) ;
\draw (120,10)  -- (120,300) ;
\draw (180,10)  -- (180,300) ;
\draw (250,10)  -- (250,300) ;
\draw (310,10)  -- (310,300) ;
\draw (390,10)  -- (390,300) ;
\draw (450,10)  -- (450,300) ;
\draw (520,10)  -- (520,300) ;
\draw (20,10)  -- (270,10) ;
\draw (290,10) -- (540,10) ;
\draw (20,300)  -- (270,300) ;
\draw (290,300) -- (540,300) ;
\draw [dash pattern={on 0.84pt off 2.51pt}] (80,70)   -- (100,70) ;
\draw [dash pattern={on 0.84pt off 2.51pt}] (210,70)  -- (230,70) ;
\draw [dash pattern={on 0.84pt off 2.51pt}] (350,70)  -- (370,70) ;
\draw [dash pattern={on 0.84pt off 2.51pt}] (480,70)  -- (500,70) ;
\draw (10,70)   -- (20,50) ;
\draw (20,90)   -- (10,70) ;
\draw (90,100)  -- (100,80) ;
\draw (100,120) -- (90,100) ;
\draw (150,120) -- (160,100) ;
\draw (160,140) -- (150,120) ;
\draw (220,140) -- (230,120) ;
\draw (230,160) -- (220,140) ;
\draw (340,70)  -- (330,50) ;
\draw (330,90)  -- (340,70) ;
\draw (420,100) -- (410,80) ;
\draw (410,120) -- (420,100) ;
\draw (480,120) -- (470,100) ;
\draw (470,140) -- (480,120) ;
\draw (550,140) -- (540,120) ;
\draw (540,160) -- (550,140) ;
\filldraw[fill=white,draw=black] (20,60)   -- (60,60)   -- (60,80)   -- (20,80)   -- cycle ;
\filldraw[fill=white,draw=black] (100,90)  -- (140,90)  -- (140,110) -- (100,110) -- cycle ;
\filldraw[fill=white,draw=black] (160,110) -- (200,110) -- (200,130) -- (160,130) -- cycle ;
\filldraw[fill=white,draw=black] (230,130) -- (270,130) -- (270,150) -- (230,150) -- cycle ;
\filldraw[fill=white,draw=black] (290,60)  -- (330,60)  -- (330,80)  -- (290,80)  -- cycle ;
\filldraw[fill=white,draw=black] (370,90)  -- (410,90)  -- (410,110) -- (370,110) -- cycle ;
\filldraw[fill=white,draw=black] (430,110) -- (470,110) -- (470,130) -- (430,130) -- cycle ;
\filldraw[fill=white,draw=black] (500,130) -- (540,130) -- (540,150) -- (500,150) -- cycle ;
\filldraw[fill=white,draw=black] (20,170) -- (270,170) -- (270,200) -- (20,200) -- cycle ;
\filldraw[fill=white,draw=black] (20,210) -- (270,210) -- (270,240) -- (20,240) -- cycle ;
\filldraw[fill=white,draw=black] (20,250) -- (270,250) -- (270,280) -- (20,280) -- cycle ;
\filldraw[fill=white,draw=black] (290,210) -- (540,210) -- (540,240) -- (290,240) -- cycle ;
\draw (29,62)   node [anchor=north west][inner sep=0.75pt] {$Z'$};
\draw (109,92)  node [anchor=north west][inner sep=0.75pt] {$Z'$};
\draw (169,112) node [anchor=north west][inner sep=0.75pt] {$W'$};
\draw (239,132) node [anchor=north west][inner sep=0.75pt] {$W'$};
\draw (137,175) node [anchor=north west][inner sep=0.75pt] {$h^{-1}$};
\draw (139,220) node [anchor=north west][inner sep=0.75pt] {$g_{i}$};
\draw (142,257) node [anchor=north west][inner sep=0.75pt] {$h$};
\draw (299,60)  node [anchor=north west][inner sep=0.75pt] {$Z^{\prime \dagger }$};
\draw (379,90)  node [anchor=north west][inner sep=0.75pt] {$Z^{\prime \dagger }$};
\draw (437,110) node [anchor=north west][inner sep=0.75pt] {$W^{\prime \dagger }$};
\draw (506,130) node [anchor=north west][inner sep=0.75pt] {$W^{\prime \dagger }$};
\draw (389,213) node [anchor=north west][inner sep=0.75pt] {$g_{j}^{-1}$};
\end{tikzpicture}
}}
\end{aligned}
\end{equation}
The left diagram is equivalent to a diagram with a single permutation $h g_i h^{-1}$.
It is clear from the diagram that this can be viewed as $(2m+2n)$-stranded diagram, and the pair of permutations $hg_ih^{-1}, g_j^{-1}$ can be combined into a single permutation $hg_ih^{-1} \otimes g_j^{-1}$ of the $2m+2n$ strands.

Each contraction receives a contribution from either the first or the second term in \eqref{eq:suN wick}.
Consider, for example, the contribution from the first term \eqref{eq:suN wick} in the top contraction.
This corresponds to inserting a permutation diagram $(1,m+n+1)$.
The contribution from the second term is simply the identity permutation on $\{1,\dots,2m+2n\}$ strands weighted by $-\tfrac{1}{N}$.
We have to sum over all such choices, which corresponds to summing over subsets $S \subseteq \{1,\dots,m+n\}$ for which we pick the first contribution.
For example, a subset $S = \{a,b,c,d\} \subseteq \{1,\dots,m+n\}$ corresponds to inserting the permutation $(a,m+n+a)(b,m+n+b)(c,m+n+c)(d,m+n+d)$ into the diagram with a weight $(-\tfrac{1}{N})^{m+n-4}$ coming from the second term for all other contractions.
Therefore, the previous diagram evaluates to 
\begin{equation}
    \sum_{\substack{h \in S_m \times S_n \\ S \subseteq \{1,\dots,m+n\}}}
    \qty(\frac{-1}{N})^{m+n-|S|  }
\vcenter{\hbox{
        \resizebox{.6\linewidth}{!}{
\begin{tikzpicture}[x=0.75pt,y=0.75pt,yscale=-1,xscale=1]
\draw    (40,90) -- (40,140) ;
\draw    (120,90) -- (120,140) ;
\draw  [dash pattern={on 0.84pt off 2.51pt}]  (70,110) -- (90,110) ;
\draw    (180,90) -- (180,140) ;
\draw    (250,90) -- (250,140) ;
\draw  [dash pattern={on 0.84pt off 2.51pt}]  (200,110) -- (220,110) ;
\draw    (20,90) -- (540,90) ;
\draw   (20,210) -- (540,210) -- (540,240) -- (20,240) -- cycle ;
\draw    (250,260) -- (250,240) ;
\draw    (180,260) -- (180,240) ;
\draw    (120,260) -- (120,240) ;
\draw    (40,260) -- (40,240) ;
\draw    (540,260) -- (20,260) ;
\draw    (310,190) -- (310,210) ;
\draw    (310,90) -- (310,140) ;
\draw    (390,190) -- (390,210) ;
\draw    (390,90) -- (390,140) ;
\draw  [dash pattern={on 0.84pt off 2.51pt}]  (340,110) -- (360,110) ;
\draw    (450,190) -- (450,210) ;
\draw    (450,90) -- (450,140) ;
\draw    (520,190) -- (520,210) ;
\draw    (520,90) -- (520,140) ;
\draw  [dash pattern={on 0.84pt off 2.51pt}]  (470,110) -- (490,110) ;
\draw    (520,260) -- (520,240) ;
\draw    (450,260) -- (450,240) ;
\draw    (390,260) -- (390,240) ;
\draw    (310,260) -- (310,240) ;
\draw   (20,140) -- (540,140) -- (540,190) -- (20,190) -- cycle ;
\draw    (40,210) -- (40,190) ;
\draw    (120,210) -- (120,190) ;
\draw    (180,210) -- (180,190) ;
\draw    (250,210) -- (250,190) ;
\draw (233,212) node [anchor=north west][inner sep=0.75pt]    {$hg_{i} h^{-1} \otimes g_{j}^{-1}$};
\draw (200,160) node [anchor=north west][inner sep=0.75pt]    {$\prod_{l \in S} \,( l,m+n+l)$};
\end{tikzpicture}
}}}
\end{equation}
Lastly, each diagram evaluates to a power of $N$, with power determined by the number of cycles in
\begin{equation}
    \prod_{l \in S} (l,m+n+l) [h g_i h^{-1} \otimes g_j^{-1}].
\end{equation}
\end{proof}

\section{Mixing matrix data}\label{app:mixing matrix data}

Below, we present the one-loop mixing matrices in a SEP basis of $\cA^{\rm HWS}(m,n)$, where $F_a \equiv 1-a \gss$. 
We partition each matrix into blocks according to the value of $\ell(R)$. 
The mixing matrix $( H_{\gss}^{(1)})^{SR}$ has nonzero entries only when $\ell(S) \in \{ \ell(R)-1, \ell(R), \ell(R)+1 \}$; this is a consequence of known selection rules \cite{deMelloKoch:2007rqf,Brown:2008rs}.
All entries satisfying $\ell(R)=\ell(S)+1$ are proportional to $F_{\ell(S)}$, which vanishes when $1/\gss=N=\ell(S)$.

The space $\cA^{\rm HWS}(3,2)$ is spanned by states carrying the $R$-labels
\begin{equation}
\{ [4,1], [3,2], [3,1,1] \times 2, [2,2,1], [2,1,1,1]\}
\end{equation}
where $R=[3,1,1]$ occurs with multiplicity two \cite{Padellaro:2024rld}.
The one-loop mixing matrix in the HWS of $\cA(3,2)$, expressed in a SEP basis, is
\begin{equation}
( H_{\gss}^{(1)})^{SR}=
\begin{tikzpicture}[baseline=(M.center)]
\matrix (M) [matrix of math nodes,left delimiter={(},right delimiter={)},ampersand replacement=\&,nodes={font=\small,inner xsep=3mm,inner ysep=2mm}] {
2\&4F_{-3}\&\dfrac{F_2}{3}\&2F_2\&0\&0\\
\dfrac45\&\dfrac{8F_{-3}}5\&\dfrac{2F_2}{15}\&\dfrac{4F_2}{5}\&0\&0\\
\dfrac65\&\dfrac{12F_{-3}}5\&\dfrac25\&-\dfrac{24\gss}{5}\&\dfrac65\&\dfrac{12F_3}{5}\\
\dfrac65\&\dfrac{12F_{-3}}5\&-\dfrac{4\gss}{5}\&\dfrac{12}{5}\&-\dfrac65\&-\dfrac{12F_3}{5}\\
0\&0\&\dfrac{F_{-2}}3\&-2F_{-2}\&2\&4F_3\\
0\&0\&\dfrac{2F_{-2}}{15}\&-\dfrac{4F_{-2}}5\&\dfrac45\&\dfrac{8F_3}{5}\\
};
\node[fit=(M-1-1)(M-6-2),inner sep=0pt] (R2) {};
\node[fit=(M-1-3)(M-6-5),inner sep=0pt] (R3) {};
\node[fit=(M-1-6)(M-6-6),inner sep=0pt] (R4) {};
\node[fit=(M-1-1)(M-2-6),inner sep=0pt] (S2) {};
\node[fit=(M-3-1)(M-5-6),inner sep=0pt] (S3) {};
\node[fit=(M-6-1)(M-6-6),inner sep=0pt] (S4) {};
\node[fit=(M-1-1)(M-6-6),inner sep=0pt] (B) {};
\coordinate (V2) at ([xshift=4mm]R2.east);
\coordinate (V3) at ([xshift=2mm]R3.east);
\draw (V2|-B.north)--(V2|-B.south);
\draw (V3|-B.north)--(V3|-B.south);
\coordinate (H2) at ([yshift=-2mm]S2.south);
\coordinate (H3) at ([yshift=-2mm]S3.south);
\draw (B.west|-H2)--(B.east|-H2);
\draw (B.west|-H3)--(B.east|-H3);
\coordinate (T) at ([yshift=3mm]M.north);
\node[anchor=base,font=\small] at (T-|R2.center) {$\ell(R)=2$};
\node[anchor=base,font=\small] at (T-|R3.center) {$\ell(R)=3$};
\node[anchor=base,font=\small] at (T-|R4.center) {$\ell(R)=4$};
\coordinate (L) at ([xshift=3mm]M.east);
\node[anchor=west,xshift=2mm,font=\small] at (L|-S2.center) {$\ell(S)=2$};
\node[anchor=west,xshift=2mm,font=\small] at (L|-S3.center) {$\ell(S)=3$};
\node[anchor=west,xshift=2mm,font=\small] at (L|-S4.center) {$\ell(S)=4$};
\end{tikzpicture}
\end{equation}

The physical one-loop dilatation operator $\fD_N^{(1)}$ is Hermitian with respect to the free two-point-function inner product of $\cN=4$ SYM. 
Its matrix representation in a SEP basis need not be Hermitian in the ordinary sense, because that basis need not be orthonormal in the physical inner product. For $\ell(R)>N$, SEP basis elements have zero norm and represent null states.
As a result, $( H_{\gss}^{(1)})^{SR}$ is not a Hermitian matrix.

With the same notation, the one-loop mixing matrices in the HWS of $\cA(4,2)$ and $\cA(5,2)$ in a SEP basis are presented as follows.
For $\cA(m,2)$, the block columns and rows correspond to $\ell(R)$ and $\ell(S)$ taking values in $\{2,3,\ldots,m+1\}$, arranged in increasing order.

\begin{table}[!htbp]
\centering
\setlength{\arraycolsep}{2pt}
\setlength{\extrarowheight}{3mm}
\begin{adjustbox}{max width=\linewidth}
\tiny
${\everymath{\displaystyle}
\begin{array}{|cccc|cccccc|cccc|c|}
\hline
 \frac{6 F_{-3}}{5} & -\frac{8 F_{-3}}{5} & -6 F_{-1} & 0 & -\frac{21 F_2}{40} & \frac{3 F_2}{10} & -\frac{9 F_2}{40} & 0 & 0 & 0 & 0 & 0 & 0 & 0 & 0 \\
 -\frac{2 F_{-3}}{5} & \frac{11 F_{-\frac{8}{11}}}{5} & 2 F_{-1} & -5 F_{-4} & \frac{7 F_2}{40} & -\frac{F_2}{10} & \frac{3 F_2}{40} & -\frac{F_2}{2} & -2 F_2 & 0 & 0 & 0 & 0 & 0 & 0 \\
 -\frac{2 F_{-3}}{3} & \frac{8 F_{-3}}{9} & \frac{10 F_{-1}}{3} & 0 & \frac{7 F_2}{24} & -\frac{F_2}{6} & \frac{F_2}{8} & 0 & 0 & 0 & 0 & 0 & 0 & 0 & 0 \\
 0 & -\frac{5}{9} & 0 & \frac{5 F_{-4}}{3} & 0 & 0 & 0 & \frac{F_2}{6} & \frac{2 F_2}{3} & 0 & 0 & 0 & 0 & 0 & 0 \\[2mm]
\hline
 -\frac{32 F_{-3}}{15} & \frac{128 F_{-3}}{45} & \frac{32 F_{-1}}{3} & 0 & \frac{28}{15} & \frac{32 \gss}{15} & \frac{4}{5} & 0 & 0 & \frac{32 F_1}{3} & 0 & 0 & -\frac{32 F_3}{15} & -\frac{128 F_3}{45} & 0 \\
 \frac{8 F_{-3}}{15} & -\frac{32 F_{-3}}{45} & -\frac{8 F_{-1}}{3} & 0 & \frac{14 \gss}{15} & \frac{4}{15} & \frac{2 \gss}{5} & 0 & 0 & \frac{8 F_1}{3} & 0 & 0 & -\frac{8 F_3}{15} & -\frac{32 F_3}{45} & 0 \\
 \frac{4 F_{-3}}{15} & -\frac{16 F_{-3}}{45} & -\frac{4 F_{-1}}{3} & 0 & \frac{2}{45} & -\frac{4 \gss}{15} & \frac{58}{45} & -\frac{4 F_{-3}}{9} & \frac{8 F_{-3}}{3} & -\frac{4 F_1}{3} & -\frac{4 F_3}{9} & -\frac{8 F_3}{3} & \frac{4
   F_3}{15} & \frac{16 F_3}{45} & 0 \\
 0 & 0 & 0 & 0 & -\frac{5}{36} & 0 & -\frac{25}{36} & \frac{2 F_{-3}}{9} & -\frac{4 F_{-3}}{3} & 0 & \frac{2 F_3}{9} & \frac{4 F_3}{3} & 0 & 0 & 0 \\
 0 & -\frac{10}{9} & 0 & \frac{10 F_{-4}}{3} & \frac{5}{36} & 0 & \frac{25}{36} & \frac{F_{12}}{9} & \frac{8 F_{-\frac{1}{2}}}{3} & 0 & -\frac{2 F_3}{9} & -\frac{4 F_3}{3} & 0 & 0 & 0 \\
 0 & 0 & 0 & 0 & \frac{7 F_{-2}}{24} & \frac{F_{-2}}{6} & \frac{F_{-2}}{8} & 0 & 0 & \frac{10 F_1}{3} & 0 & 0 & -\frac{2 F_3}{3} & -\frac{8 F_3}{9} & 0 \\[2mm]
\hline
 0 & 0 & 0 & 0 & -\frac{5}{36} & 0 & -\frac{25}{36} & \frac{2 F_{-3}}{9} & -\frac{4 F_{-3}}{3} & 0 & \frac{2 F_3}{9} & \frac{4 F_3}{3} & 0 & 0 & 0 \\
 0 & 0 & 0 & 0 & -\frac{5}{36} & 0 & -\frac{25}{36} & \frac{2 F_{-3}}{9} & -\frac{4 F_{-3}}{3} & 0 & -\frac{F_{-12}}{9} & \frac{8 F_{\frac{1}{2}}}{3} & 0 & -\frac{10}{9} & -\frac{10 F_4}{3} \\
 0 & 0 & 0 & 0 & -\frac{21 F_{-2}}{40} & -\frac{3 F_{-2}}{10} & -\frac{9 F_{-2}}{40} & 0 & 0 & -6 F_1 & 0 & 0 & \frac{6 F_3}{5} & \frac{8 F_3}{5} & 0 \\
 0 & 0 & 0 & 0 & -\frac{7 F_{-2}}{40} & -\frac{F_{-2}}{10} & -\frac{3 F_{-2}}{40} & 0 & 0 & -2 F_1 & \frac{F_{-2}}{2} & -2 F_{-2} & \frac{2 F_3}{5} & \frac{11 F_{\frac{8}{11}}}{5} & 5 F_4 \\[2mm]
\hline
 0 & 0 & 0 & 0 & 0 & 0 & 0 & 0 & 0 & 0 & \frac{F_{-2}}{6} & -\frac{2 F_{-2}}{3} & 0 & \frac{5}{9} & \frac{5 F_4}{3} \\[2mm]
\hline
\end{array}
}$
\end{adjustbox}
\caption{One-loop mixing matrix in the HWS of $\cA(4,2)$ in a SEP basis.}
\end{table}

\begin{sidewaystable}[p]
\centering
\setlength{\arraycolsep}{2pt}
\setlength{\extrarowheight}{3mm}
\begin{adjustbox}{max width=\linewidth, max totalheight=0.82\textheight}
\tiny
${\everymath{\displaystyle}
\begin{array}{|cccc|cccccccccc|ccccccc|cccc|c|}
\hline
 \frac{10 F_{-4}}{9} & -\frac{41}{9} F_{-\frac{110}{41}} & 6 F_{-5} & \frac{50 F_{-1}}{9} & \frac{F_2}{5} & \frac{8 F_2}{5} & 0 & 0 & \frac{4 F_2}{9} & \frac{F_2}{9} & -\frac{F_2}{3} & -\frac{F_2}{9} & 0 & 0 & 0 & 0 & 0 & 0 & 0 &
   0 & 0 & 0 & 0 & 0 & 0 & 0 \\
 -\frac{2 F_{-4}}{9} & \frac{19 F_{-\frac{22}{19}}}{9} & -6 F_{-5} & -\frac{10 F_{-1}}{9} & -\frac{F_2}{5} & -\frac{8 F_2}{5} & 0 & 0 & -\frac{4 F_2}{45} & -\frac{F_2}{45} & \frac{F_2}{15} & \frac{F_2}{45} & 0 & 0 & 0 & 0 & 0 & 0
   & 0 & 0 & 0 & 0 & 0 & 0 & 0 & 0 \\
 0 & -\frac{3}{7} & \frac{12 F_{-5}}{7} & 0 & \frac{2 F_2}{35} & \frac{16 F_2}{35} & 0 & 0 & 0 & 0 & 0 & 0 & 0 & 0 & 0 & 0 & 0 & 0 & 0 & 0 & 0 & 0 & 0 & 0 & 0 & 0 \\
 \frac{4 F_{-4}}{9} & -\frac{11 F_{-4}}{9} & 0 & \frac{20 F_{-1}}{9} & 0 & 0 & 0 & 0 & \frac{8 F_2}{45} & \frac{2 F_2}{45} & -\frac{2 F_2}{15} & -\frac{2 F_2}{45} & 0 & 0 & 0 & 0 & 0 & 0 & 0 & 0 & 0 & 0 & 0 & 0 & 0 & 0 \\
\hline
 0 & -\frac{15}{14} & \frac{30 F_{-5}}{7} & 0 & \frac{3}{14} & -\frac{3 F_{-20}}{7} & 0 & 0 & \frac{1}{14} & 0 & 0 & \frac{3}{7} & 0 & 0 & 0 & 0 & 0 & 0 & -\frac{F_3}{14} & -\frac{15 F_3}{14} & 0 & 0 & 0 & 0 & 0 & 0 \\
 0 & -\frac{15}{14} & \frac{30 F_{-5}}{7} & 0 & \frac{F_8}{14} & \frac{19 F_{-\frac{28}{19}}}{7} & 0 & 0 & -\frac{1}{14} & 0 & 0 & -\frac{3}{7} & 0 & 0 & 0 & 0 & 0 & 0 & \frac{F_3}{14} & \frac{15 F_3}{14} & 0 & 0 & 0 & 0 & 0 & 0
   \\
 0 & 0 & 0 & 0 & 0 & 0 & \frac{19 F_{-1}}{8} & \frac{7 F_{-1}}{8} & -\frac{7 F_{-3}}{40} & \frac{3 F_{-3}}{10} & -\frac{11 F_{-3}}{40} & -\frac{17 F_{-3}}{40} & -\frac{19 F_1}{8} & \frac{7 F_1}{8} & -\frac{7 F_3}{40} & -\frac{3
   F_3}{10} & \frac{11 F_3}{40} & -\frac{17 F_3}{40} & 0 & 0 & 0 & 0 & 0 & 0 & 0 & 0 \\
 0 & 0 & 0 & 0 & 0 & 0 & \frac{F_{-1}}{2} & \frac{11 F_{-1}}{4} & -\frac{F_{-3}}{40} & \frac{21 F_{-3}}{40} & -\frac{13 F_{-3}}{40} & -\frac{31 F_{-3}}{40} & -\frac{F_1}{4} & -\frac{F_1}{2} & \frac{13 F_3}{40} & \frac{3 F_3}{40} &
   -\frac{9 F_3}{40} & \frac{3 F_3}{40} & 0 & 0 & 0 & 0 & 0 & 0 & 0 & 0 \\
 2 F_{-4} & -\frac{11 F_{-4}}{2} & 0 & 10 F_{-1} & \frac{F_{-4}}{6} & -\frac{11 F_{-4}}{3} & -\frac{5 F_{-1}}{3} & \frac{5 F_{-1}}{6} & \frac{9 F_{-1}}{4} & \frac{F_1}{4} & \frac{F_{-39}}{12} & \frac{7 F_{\frac{3}{7}}}{12} &
   \frac{7 F_1}{6} & \frac{22 F_1}{3} & \frac{7 F_3}{12} & -\frac{5 F_3}{4} & \frac{5 F_3}{12} & -\frac{23 F_3}{12} & -\frac{F_3}{6} & -\frac{5 F_3}{2} & 0 & 0 & 0 & 0 & 0 & 0 \\
 0 & 0 & 0 & 0 & 0 & 0 & \frac{15 F_{-1}}{8} & \frac{45 F_{-1}}{8} & 0 & \frac{9 F_{-3}}{8} & -\frac{2 F_{-3}}{3} & -\frac{5 F_{-3}}{3} & -\frac{35 F_1}{24} & \frac{5 F_1}{24} & \frac{7 F_3}{12} & -\frac{F_3}{8} & -\frac{F_3}{4} &
   -\frac{F_3}{4} & 0 & 0 & 0 & 0 & 0 & 0 & 0 & 0 \\
 -\frac{4 F_{-4}}{9} & \frac{11 F_{-4}}{9} & 0 & -\frac{20 F_{-1}}{9} & 0 & 0 & -\frac{5 F_{-1}}{8} & -\frac{5 F_{-1}}{8} & \frac{7 F_{-\frac{85}{7}}}{72} & -\frac{7}{36} F_{-\frac{13}{7}} & \frac{3 F_{-\frac{11}{9}}}{8} &
   \frac{17 F_{-\frac{35}{17}}}{72} & \frac{11 F_1}{24} & \frac{37 F_1}{24} & \frac{F_3}{24} & -\frac{F_3}{4} & \frac{F_3}{8} & -\frac{3 F_3}{8} & 0 & 0 & 0 & 0 & 0 & 0 & 0 & 0 \\
 -\frac{2 F_{-4}}{9} & \frac{11 F_{-4}}{18} & 0 & -\frac{10 F_{-1}}{9} & \frac{F_{-4}}{6} & -\frac{11 F_{-4}}{3} & -\frac{5 F_{-1}}{12} & -\frac{5 F_{-1}}{3} & -\frac{F_{-7}}{72} & -\frac{25}{72} F_{-\frac{67}{25}} & \frac{5
   F_{-\frac{7}{5}}}{24} & \frac{109 F_{-\frac{103}{109}}}{72} & \frac{F_1}{3} & -\frac{7 F_1}{12} & -\frac{5 F_3}{24} & \frac{F_3}{8} & \frac{F_3}{24} & \frac{5 F_3}{24} & -\frac{F_3}{6} & -\frac{5 F_3}{2} & 0 & 0 & 0 & 0 & 0 & 0
   \\
 0 & 0 & 0 & 0 & 0 & 0 & -\frac{19 F_{-1}}{8} & -\frac{7 F_{-1}}{8} & \frac{7 F_{-3}}{40} & -\frac{3 F_{-3}}{10} & \frac{11 F_{-3}}{40} & \frac{17 F_{-3}}{40} & \frac{19 F_1}{8} & -\frac{7 F_1}{8} & \frac{7 F_3}{40} & \frac{3
   F_3}{10} & -\frac{11 F_3}{40} & \frac{17 F_3}{40} & 0 & 0 & 0 & 0 & 0 & 0 & 0 & 0 \\
 0 & 0 & 0 & 0 & 0 & 0 & \frac{F_{-1}}{4} & -\frac{F_{-1}}{2} & \frac{13 F_{-3}}{40} & -\frac{3 F_{-3}}{40} & \frac{9 F_{-3}}{40} & \frac{3 F_{-3}}{40} & -\frac{F_1}{2} & \frac{11 F_1}{4} & -\frac{F_3}{40} & -\frac{21 F_3}{40} &
   \frac{13 F_3}{40} & -\frac{31 F_3}{40} & 0 & 0 & 0 & 0 & 0 & 0 & 0 & 0 \\
\hline
 0 & 0 & 0 & 0 & 0 & 0 & -\frac{7 F_{-1}}{6} & \frac{22 F_{-1}}{3} & \frac{7 F_{-3}}{12} & \frac{5 F_{-3}}{4} & -\frac{5 F_{-3}}{12} & -\frac{23 F_{-3}}{12} & \frac{5 F_1}{3} & \frac{5 F_1}{6} & \frac{9 F_1}{4} & -\frac{F_{-1}}{4}
   & -\frac{F_{39}}{12} & \frac{7 F_{-\frac{3}{7}}}{12} & \frac{F_{-3}}{6} & -\frac{5 F_{-3}}{2} & 10 F_1 & -\frac{F_4}{6} & -\frac{11 F_4}{3} & -2 F_4 & -\frac{11 F_4}{2} & 0 \\
 0 & 0 & 0 & 0 & 0 & 0 & -\frac{35 F_{-1}}{24} & -\frac{5 F_{-1}}{24} & -\frac{7 F_{-3}}{12} & -\frac{F_{-3}}{8} & -\frac{F_{-3}}{4} & \frac{F_{-3}}{4} & \frac{15 F_1}{8} & -\frac{45 F_1}{8} & 0 & \frac{9 F_3}{8} & -\frac{2
   F_3}{3} & \frac{5 F_3}{3} & 0 & 0 & 0 & 0 & 0 & 0 & 0 & 0 \\
 0 & 0 & 0 & 0 & 0 & 0 & \frac{11 F_{-1}}{24} & -\frac{37 F_{-1}}{24} & -\frac{F_{-3}}{24} & -\frac{F_{-3}}{4} & \frac{F_{-3}}{8} & \frac{3 F_{-3}}{8} & -\frac{5 F_1}{8} & \frac{5 F_1}{8} & -\frac{7 F_{\frac{85}{7}}}{72} &
   -\frac{7 F_{\frac{13}{7}}}{36} & \frac{3 F_{\frac{11}{9}}}{8} & -\frac{17 F_{\frac{35}{17}}}{72} & 0 & 0 & \frac{20 F_1}{9} & 0 & 0 & -\frac{4 F_4}{9} & -\frac{11 F_4}{9} & 0 \\
 0 & 0 & 0 & 0 & 0 & 0 & -\frac{F_{-1}}{3} & -\frac{7 F_{-1}}{12} & -\frac{5 F_{-3}}{24} & -\frac{F_{-3}}{8} & -\frac{F_{-3}}{24} & \frac{5 F_{-3}}{24} & \frac{5 F_1}{12} & -\frac{5 F_1}{3} & -\frac{F_7}{72} & \frac{25
   F_{\frac{67}{25}}}{72} & -\frac{5 F_{\frac{7}{5}}}{24} & \frac{109 F_{\frac{103}{109}}}{72} & \frac{F_{-3}}{6} & -\frac{5 F_{-3}}{2} & -\frac{10 F_1}{9} & -\frac{F_4}{6} & -\frac{11 F_4}{3} & \frac{2 F_4}{9} & \frac{11 F_4}{18}
   & 0 \\
 0 & 0 & 0 & 0 & -\frac{2 F_{-4}}{21} & \frac{44 F_{-4}}{21} & 0 & 0 & -\frac{2}{21} & 0 & 0 & -\frac{4}{7} & 0 & 0 & \frac{2}{21} & 0 & 0 & \frac{4}{7} & \frac{4}{21} & -\frac{60 \gss}{7} & 0 & -\frac{2 F_4}{21} & -\frac{44 F_4}{21}
   & 0 & 0 & 0 \\
 0 & 0 & 0 & 0 & -\frac{2 F_{-4}}{21} & \frac{44 F_{-4}}{21} & 0 & 0 & -\frac{2}{21} & 0 & 0 & -\frac{4}{7} & 0 & 0 & -\frac{2}{21} & 0 & 0 & -\frac{4}{7} & -\frac{4 \gss}{7} & \frac{20}{7} & 0 & \frac{2 F_4}{21} & \frac{44 F_4}{21}
   & 0 & 0 & 0 \\
 0 & 0 & 0 & 0 & 0 & 0 & 0 & 0 & 0 & 0 & 0 & 0 & 0 & 0 & \frac{8 F_{-2}}{45} & -\frac{2 F_{-2}}{45} & \frac{2 F_{-2}}{15} & -\frac{2 F_{-2}}{45} & 0 & 0 & \frac{20 F_1}{9} & 0 & 0 & -\frac{4 F_4}{9} & -\frac{11 F_4}{9} & 0 \\
\hline
 0 & 0 & 0 & 0 & 0 & 0 & 0 & 0 & 0 & 0 & 0 & 0 & 0 & 0 & -\frac{1}{14} & 0 & 0 & -\frac{3}{7} & -\frac{F_{-3}}{14} & \frac{15 F_{-3}}{14} & 0 & \frac{3}{14} & \frac{3 F_{20}}{7} & 0 & \frac{15}{14} & -\frac{30 F_5}{7} \\
 0 & 0 & 0 & 0 & 0 & 0 & 0 & 0 & 0 & 0 & 0 & 0 & 0 & 0 & -\frac{1}{14} & 0 & 0 & -\frac{3}{7} & -\frac{F_{-3}}{14} & \frac{15 F_{-3}}{14} & 0 & -\frac{F_{-8}}{14} & \frac{19 F_{\frac{28}{19}}}{7} & 0 & -\frac{15}{14} & \frac{30
   F_5}{7} \\
 0 & 0 & 0 & 0 & 0 & 0 & 0 & 0 & 0 & 0 & 0 & 0 & 0 & 0 & -\frac{4 F_{-2}}{9} & \frac{F_{-2}}{9} & -\frac{F_{-2}}{3} & \frac{F_{-2}}{9} & 0 & 0 & -\frac{50 F_1}{9} & \frac{F_{-2}}{5} & -\frac{8 F_{-2}}{5} & \frac{10 F_4}{9} &
   \frac{41 F_{\frac{110}{41}}}{9} & -6 F_5 \\
 0 & 0 & 0 & 0 & 0 & 0 & 0 & 0 & 0 & 0 & 0 & 0 & 0 & 0 & -\frac{4 F_{-2}}{45} & \frac{F_{-2}}{45} & -\frac{F_{-2}}{15} & \frac{F_{-2}}{45} & 0 & 0 & -\frac{10 F_1}{9} & \frac{F_{-2}}{5} & -\frac{8 F_{-2}}{5} & \frac{2 F_4}{9} &
   \frac{19 F_{\frac{22}{19}}}{9} & -6 F_5 \\
\hline
 0 & 0 & 0 & 0 & 0 & 0 & 0 & 0 & 0 & 0 & 0 & 0 & 0 & 0 & 0 & 0 & 0 & 0 & 0 & 0 & 0 & -\frac{2 F_{-2}}{35} & \frac{16 F_{-2}}{35} & 0 & -\frac{3}{7} & \frac{12 F_5}{7} \\
\hline
\end{array}
}$
\end{adjustbox}
\caption{One-loop mixing matrix in the HWS of $\cA(5,2)$ in a SEP basis.}
\end{sidewaystable}

\clearpage
\paragraph{Integer eigenvectors.}

For small $(m,n)$, there exist special eigenstates whose eigenvalues are integers and receive no $1/N$ corrections.
Let us introduce the notation
\begin{equation}
\fD_N^{(1)}\,\cO_{m,n}(e)=e\,\cO_{m,n}(e)
\end{equation}
where $\cO_{m,n}(e)$ denotes a multi-trace operator with integer one-loop eigenvalue $e\in\mathbb{Z}$, chosen to be a $U(2)$ HWS for $m\geq n$.

One famous example is the Konishi descendant in the $su(2)$ sector,
\begin{equation}
\mathcal{O}_{2,2} (6) = {\rm tr}(WWZZ) - {\rm tr}(WZWZ) .
\label{app:Konishi}
\end{equation}
The Konishi descendant has one-loop eigenvalue $6$ exactly at any $N$. There is no other non-BPS state in the HWS sector of $\cA(2,2)$ \cite{Bianchi:1999ge}. The next simplest operator is
\begin{equation}
\begin{aligned}
\mathcal{O}_{3,2} (4) &= {\rm tr}(WWZZZ) - {\rm tr}(WZWZZ) 
- \frac{2}{N}\Big( {\rm tr}(WWZZ) - {\rm tr}(WZWZ) \Big){\rm tr}(Z) 
\\
&= \tr(W' W' Z' Z' Z') - \tr(W'Z'W'Z' Z')
\end{aligned}
\label{eq:O32+4} 
\end{equation}
where the $1/N$ terms disappear when written by traceless variables. 
Its one-loop eigenvalue is $4$ as long as this operator remains non-zero.
There is no other non-BPS state in the HWS sector of $\cA(3,2)$ without single-letter-traces.
The third simplest example is a three-excitation state corresponding to the singular Bethe root configuration \cite{Beisert:2003tq,Beisert:2004hm},
\begin{equation}
\mathcal{O}_{3,3} (6) = {\rm tr} (WWZWZZ) - {\rm tr} (WWZZWZ)
\label{eq:O33+6}
\end{equation}
which is called an exceptional operator \cite{Arutyunov:2012tx}.
This is the unique negative-parity operator in $\cA(3,3)$, and its one-loop eigenvalue is $6$ exactly at any $N$.\footnote{The parity is defined by taking the transpose $(Z_{ij}, W_{ij}) \mapsto (Z_{ji}, W_{ji})$ in \eqref{eq:Og-definition} and is realised as inverting the permutation $\swmap^{(N)}_{m,n}(g)\mapsto\swmap^{(N)}_{m,n}(g^{-1})$ \cite{Beisert:2003tq}.}

The single-letter-trace descendants of the above operators can be found at a larger $(m,n)$, which have integer eigenvalues.
For example, if we multiply only ${\rm tr} (Z)$'s, we find
\begin{align}
\mathcal{O}_{m,2} (6) 
&= \mathcal{O}_{2,2} (6) \,  {\rm tr}(Z)^{m-2} , \qquad (m \ge 2)
\\[2mm]
\mathcal{O}_{m,2} (4) 
&= \mathcal{O}_{3,2} (4) \,  {\rm tr}(Z)^{m-3} , \qquad (m \ge 3)
\\[2mm]
\mathcal{O}_{m,3} (6) &= \mathcal{O}_{3,3} (6) \,  {\rm tr}(Z)^{m-3} , \qquad (m \ge 3).
\end{align}
If we multiply both ${\rm tr} (Z)$ and ${\rm tr} (W)$, we find
\begin{align}
\mathcal{O}_{3,3} (4) &= 
\Big( {\rm tr}(WWZZZ) - {\rm tr}(WZWZZ) \Big) {\rm tr}(W)
- \Big( {\rm tr}(WWWZZ) - {\rm tr}(WWZWZ) \Big) {\rm tr}(Z) 
\notag \\
&= \mathcal{O}_{3,2} (4) \, {\rm tr}(W) - \mathcal{O}_{2,3} (4) \, {\rm tr}(Z) 
\label{eq:O33+4}
\end{align}
where $\cO_{2,3}(4)$ denotes the operator obtained from $\cO_{3,2}(4)$ by exchanging $Z$ and $W$.
The relative coefficients are chosen such that $\cO_{m,n} (e)$ becomes the HWS of the global $\su2$ symmetry.

\section{Characteristic polynomial data}\label{app:polydata}

Below we summarise the characteristic polynomials in the HWS sectors for gauge groups $U(N)$ and $SU(N)$ in the range $4 \leq m+n \leq 8$. 
The results for higher values of $m+n$ are too lengthy to be displayed here.
In the main text, these polynomials are denoted by $\det_{\cH_N^{\mathrm{HWS}}(m,n)}(x-\fD^{(1)}_N)$ and $\det_{\cH'_N{}^{\mathrm{HWS}}(m,n)}(x-\fD^{\prime(1)}_N)$, respectively. The parameter $\gss = 1/N$ is left symbolic in the stable regime.

\subsection{$U(N)$ case}
\begin{table}[!htbp]
\centering
\small
\setlength{\tabcolsep}{4pt}
\setlength{\extrarowheight}{1pt}
\renewcommand{\arraystretch}{1}
\begin{tabular}{ccc|l}
$m$ & $n$ & $N$ & $\det_{\cH_N^{\mathrm{HWS}}(m,n)} (x-  \fD^{(1)}_N)$ \\[1mm]
\hline
$2$ & $2$ & $N\ge 4$ & $(x - 6)\,x^{2}$ \\
 & & $3$ & $(x - 6)\,x^{2}$ \\
 & & $2$ & $(x - 6)\,x$ \\[1mm]
\hline
$3$ & $2$ & $N\ge 5$ & $(x - 6)\,(x - 4)\,x^{4}$ \\
 & & $4$ & $(x - 6)\,(x - 4)\,x^{4}$ \\
 & & $3$ & $(x - 6)\,(x - 4)\,x^{3}$ \\
 & & $2$ & $(x - 6)\,x$ \\[1mm]
\hline
$4$ & $2$ & $N\ge 6$ & $(x - 6)\,(x - 4)\,x^{10}\,(x^{3} - 16 x^{2} + (-40 \gss^{2} + 80) x + 80 \gss^{2} - 120)$ \\
 & & $5$ & $(x - 6)\,(x - 4)\,x^{10}\,(x^{3} - 16 x^{2} + \frac{392}{5} \, x - \frac{584}{5})$ \\
 & & $4$ & $(x - 6)\,(x - 4)\,x^{9}\,(x^{3} - 16 x^{2} + \frac{155}{2} \, x - 115)$ \\
 & & $3$ & $(x - 6)\,(x - 4)\,x^{6}\,(x^{2} - \frac{38}{3} \, x + \frac{100}{3})$ \\
 & & $2$ & $(x - 10)\,(x - 6)\,x^{2}$ \\[1mm]
\hline
$3$ & $3$ & $N\ge 6$ & $(x - 6)\,(x - 4)\,x^{2}$ \\
 & & $5$ & $(x - 6)\,(x - 4)\,x^{2}$ \\
 & & $4$ & $(x - 6)\,(x - 4)\,x^{2}$ \\
 & & $3$ & $(x - 6)\,(x - 4)\,x$ \\
 & & $2$ & $1$ \\[1mm]
\hline
$5$ & $2$ & $N\ge 7$ & $(x - 6)\,(x - 4)\,x^{17}\,(x^{3} - 16 x^{2} + (-40 \gss^{2} + 80) x + 80 \gss^{2} - 120) \ \times$ \\
& & & \hspace{55mm} $(x^{4} - 18 x^{3} + (-40 \gss^{2} + 116) x^{2} + (168 \gss^{2} - 312) x + 288)$ \\
 & & $6$ & $(x - 6)\,(x - 4)\,x^{17}\,(x^{3} - 16 x^{2} + \frac{710}{9} \, x - \frac{1060}{9})\,(x^{4} - 18 x^{3} + \frac{1034}{9} \, x^{2} - \frac{922}{3} \, x + 288)$ \\
 & & $5$ & $(x - 6)\,(x - 4)\,x^{16}\,(x^{3} - 16 x^{2} + \frac{392}{5} \, x - \frac{584}{5})\,(x^{4} - 18 x^{3} + \frac{572}{5} \, x^{2} - \frac{7632}{25} \, x + 288)$ \\
 & & $4$ & $(x - 6)\,(x - 4)\,x^{13}\,(x^{3} - 16 x^{2} + \frac{155}{2} \, x - 115)\,(x^{3} - 15 x^{2} + \frac{137}{2} \, x - 96)$ \\
 & & $3$ & $(x - 6)\,(x - 4)\,x^{8}\,(x^{2} - \frac{38}{3} \, x + \frac{100}{3})\,(x^{2} - \frac{38}{3} \, x + 36)$ \\
 & & $2$ & $(x - 10)\,(x - 6)\,x^{2}$ \\[1mm]
\hline
$4$ & $3$ & $N\ge 7$ & $(x - 6)\,(x - 5)\,(x - 4)\,x^{10}\,(x^{2} - 9 x - 40 \gss^{2} + 20)\,(x^{3} - 16 x^{2} + (-40 \gss^{2} + 80) x + 80 \gss^{2} - 120)$ \\
 & & $6$ & $(x - 6)\,(x - \frac{17}{3})\,(x - 5)\,(x - 4)\,(x - \frac{10}{3})\,x^{10}\,(x^{3} - 16 x^{2} + \frac{710}{9} \, x - \frac{1060}{9})$ \\
 & & $5$ & $(x - 6)\,(x - 5)\,(x - 4)\,x^{10}\,(x^{2} - 9 x + \frac{92}{5})\,(x^{3} - 16 x^{2} + \frac{392}{5} \, x - \frac{584}{5})$ \\
 & & $4$ & $(x - 6)\,(x - 5)\,(x - 4)\,x^{8}\,(x^{2} - 9 x + \frac{35}{2})\,(x^{3} - 16 x^{2} + \frac{155}{2} \, x - 115)$ \\
 & & $3$ & $(x - \frac{20}{3})\,(x - 6)\,(x - 4)\,x^{4}\,(x^{2} - \frac{38}{3} \, x + \frac{100}{3})$ \\
 & & $2$ & $(x - 10)\,x$ \\[1mm]
\hline
\end{tabular}
\caption{$U(N)$ characteristic polynomials for $4 \le m+n \le 7$.}\label{tbl:UN-combined-charpoly}
\end{table}

\begin{sidewaystable}[p]
\centering
\footnotesize
\setlength{\tabcolsep}{4pt}
\setlength{\extrarowheight}{1pt}
\renewcommand{\arraystretch}{1.13}
\begin{tabular}{ccc|p{0.90\textheight}}
$m$ & $n$ & $N$ & $\det_{\cH_N^{\mathrm{HWS}}(m,n)} (x-  \fD^{(1)}_N)$ \\[1mm]
\hline
$6$ & $2$ & $N\ge 8$ & $(x - 6)\,(x - 4)\,x^{32}\,(x^{3} - 16 x^{2} + (-40 \gss^{2} + 80) x + 80 \gss^{2} - 120)\,(x^{4} - 18 x^{3} + (-40 \gss^{2} + 116) x^{2} + (168 \gss^{2} - 312) x + 288)\,(x^{8} - 40 x^{7} + (-232 \gss^{2} + 684) x^{6} + (5976 \gss^{2} - 6512) x^{5} + (8576 \gss^{4} - 60976 \gss^{2} + 37616) x^{4} + (-100416 \gss^{4} + 310080 \gss^{2} - 134368) x^{3} + (-35840 \gss^{6} + 280960 \gss^{4} - 802880 \gss^{2} + 288064) x^{2} + (-164864 \gss^{6} + 288512 \gss^{4} + 934528 \gss^{2} - 336000) x - 186368 \gss^{6} - 1229312 \gss^{4} - 301056 \gss^{2} + 161280)$ \\
 & & $7$ & $(x - 6)\,(x - 4)\,x^{32}\,(x^{3} - 16 x^{2} + \frac{3880}{49} \, x - \frac{5800}{49})\,(x^{4} - 18 x^{3} + \frac{5644}{49} \, x^{2} - \frac{2160}{7} \, x + 288)\,(x^{8} - 40 x^{7} + \frac{33284}{49} \, x^{6} - \frac{313112}{49} \, x^{5} + \frac{87336768}{2401} \, x^{4} - \frac{307524064}{2401} \, x^{3} + \frac{4568065408}{16807} \, x^{2} - \frac{5324612864}{16807} \, x + \frac{2598738944}{16807})$ \\
 & & $6$ & $(x - 6)\,(x - 4)\,x^{31}\,(x^{3} - 16 x^{2} + \frac{710}{9} \, x - \frac{1060}{9})\,(x^{4} - 18 x^{3} + \frac{1034}{9} \, x^{2} - \frac{922}{3} \, x + 288)\,(x^{8} - 40 x^{7} + \frac{6098}{9} \, x^{6} - 6346 x^{5} + \frac{2910236}{81} \, x^{4} - \frac{3397468}{27} \, x^{3} + \frac{193897816}{729} \, x^{2} - \frac{225860096}{729} \, x + \frac{110782336}{729})$ \\
 & & $5$ & $(x - 6)\,(x - 4)\,x^{28}\,(x^{3} - 16 x^{2} + \frac{392}{5} \, x - \frac{584}{5})\,(x^{4} - 18 x^{3} + \frac{572}{5} \, x^{2} - \frac{7632}{25} \, x + 288)\,(x^{7} - \frac{186}{5} \, x^{6} + \frac{14264}{25} \, x^{5} - \frac{584424}{125} \, x^{4} + \frac{2762448}{125} \, x^{3} - \frac{37654144}{625} \, x^{2} + \frac{274079616}{3125} \, x - \frac{164351488}{3125})$ \\
 & & $4$ & $(x - 6)\,(x - 4)\,x^{22}\,(x^{3} - 16 x^{2} + \frac{155}{2} \, x - 115)\,(x^{3} - 15 x^{2} + \frac{137}{2} \, x - 96)\,(x^{6} - 35 x^{5} + \frac{975}{2} \, x^{4} - 3456 x^{3} + 13146 x^{2} - \frac{101833}{4} \, x + \frac{39319}{2})$ \\
 & & $3$ & $(x - 6)\,(x - 4)\,x^{12}\,(x^{2} - \frac{38}{3} \, x + \frac{100}{3})\,(x^{2} - \frac{38}{3} \, x + 36)\,(x^{3} - 22 x^{2} + \frac{1304}{9} \, x - \frac{2576}{9})$ \\
 & & $2$ & $(x - 14)\,(x - 10)\,(x - 6)\,x^{3}$ \\[1mm]
\hline
$5$ & $3$ & $N\ge 8$ & $(x - 6)\,(x - 5)\,(x - 6 \gss - 6)\,(x + 6 \gss - 6)\,(x - 4)^{2}\,x^{20}\,(x^{2} - 9 x - 40 \gss^{2} + 20)\,(x^{3} - 16 x^{2} + (-40 \gss^{2} + 80) x + 80 \gss^{2} - 120)\,(x^{4} - 18 x^{3} + (-40 \gss^{2} + 116) x^{2} + (168 \gss^{2} - 312) x + 288)\,(x^{4} - 18 x^{3} + (-40 \gss^{2} + 116) x^{2} + (432 \gss^{2} - 320) x - 960 \gss^{2} + 320)$ \\
 & & $7$ & $(x - \frac{48}{7})\,(x - 6)\,(x - \frac{36}{7})\,(x - 5)\,(x - 4)^{2}\,x^{20}\,(x^{2} - 9 x + \frac{940}{49})\,(x^{3} - 16 x^{2} + \frac{3880}{49} \, x - \frac{5800}{49})\,(x^{4} - 18 x^{3} + \frac{5644}{49} \, x^{2} - \frac{15248}{49} \, x + \frac{14720}{49})\,(x^{4} - 18 x^{3} + \frac{5644}{49} \, x^{2} - \frac{2160}{7} \, x + 288)$ \\
 & & $6$ & $(x - 7)\,(x - 6)\,(x - \frac{17}{3})\,(x - 5)^{2}\,(x - 4)^{2}\,(x - \frac{10}{3})^{2}\,x^{20}\,(x^{3} - 16 x^{2} + \frac{710}{9} \, x - \frac{1060}{9})\,(x^{3} - \frac{44}{3} \, x^{2} + 66 x - 88)\,(x^{4} - 18 x^{3} + \frac{1034}{9} \, x^{2} - \frac{922}{3} \, x + 288)$ \\
 & & $5$ & $(x - \frac{36}{5})\,(x - 6)\,(x - 5)\,(x - \frac{24}{5})\,(x - 4)^{2}\,x^{18}\,(x^{2} - 9 x + \frac{92}{5})\,(x^{3} - 16 x^{2} + \frac{392}{5} \, x - \frac{584}{5})\,(x^{4} - 18 x^{3} + \frac{572}{5} \, x^{2} - \frac{7632}{25} \, x + 288)\,(x^{4} - 18 x^{3} + \frac{572}{5} \, x^{2} - \frac{7568}{25} \, x + \frac{1408}{5})$ \\
 & & $4$ & $(x - \frac{15}{2})\,(x - 6)\,(x - 5)\,(x - 4)^{2}\,x^{13}\,(x^{2} - 9 x + \frac{35}{2})\,(x^{3} - 16 x^{2} + \frac{155}{2} \, x - 115)\,(x^{3} - 16 x^{2} + \frac{163}{2} \, x - 130)\,(x^{3} - 15 x^{2} + \frac{137}{2} \, x - 96)$ \\
 & & $3$ & $(x - 8)\,(x - 6)\,(x - 4)\,(x - \frac{20}{3})^{2}\,x^{6}\,(x^{2} - \frac{38}{3} \, x + \frac{100}{3})\,(x^{2} - \frac{38}{3} \, x + 36)$ \\
 & & $2$ & $(x - 10)\,x$ \\[1mm]
\hline
$4$ & $4$ & $N\ge 8$ & $(x - 5)\,x^{12}\,(x^{2} - 9 x - 40 \gss^{2} + 20)\,(x^{3} - 16 x^{2} + (-40 \gss^{2} + 80) x + 80 \gss^{2} - 120)\,(x^{7} - 48 x^{6} + (-412 \gss^{2} + 948) x^{5} + (11840 \gss^{2} - 9968) x^{4} + (30560 \gss^{4} - 132848 \gss^{2} + 60192) x^{3} + (-476800 \gss^{4} + 722080 \gss^{2} - 208480) x^{2} + (-627200 \gss^{6} + 2770560 \gss^{4} - 1912000 \gss^{2} + 383040) x + 2816000 \gss^{6} - 5824000 \gss^{4} + 2016000 \gss^{2} - 288000)$ \\
 & & $7$ & $(x - 5)\,x^{12}\,(x^{2} - 9 x + \frac{940}{49})\,(x^{3} - 16 x^{2} + \frac{3880}{49} \, x - \frac{5800}{49})\,(x^{7} - 48 x^{6} + \frac{46040}{49} \, x^{5} - \frac{476592}{49} \, x^{4} + \frac{138042000}{2401} \, x^{3} - \frac{465655360}{2401} \, x^{2} + \frac{828748800}{2401} \, x - \frac{29325056000}{117649})$ \\
 & & $6$ & $(x - \frac{17}{3})\,(x - 5)\,(x - \frac{10}{3})\,x^{12}\,(x^{3} - 16 x^{2} + \frac{710}{9} \, x - \frac{1060}{9})\,(x^{7} - 48 x^{6} + \frac{8429}{9} \, x^{5} - \frac{86752}{9} \, x^{4} + \frac{4578554}{81} \, x^{3} - \frac{15292000}{81} \, x^{2} + \frac{242066800}{729} \, x - \frac{172360000}{729})$ \\
 & & $5$ & $(x - 5)\,x^{11}\,(x^{2} - 9 x + \frac{92}{5})\,(x^{3} - 16 x^{2} + \frac{392}{5} \, x - \frac{584}{5})\,(x^{7} - 48 x^{6} + \frac{23288}{25} \, x^{5} - \frac{47472}{5} \, x^{4} + \frac{6865872}{125} \, x^{3} - \frac{4508992}{25} \, x^{2} + \frac{194345472}{625} \, x - \frac{27062272}{125})$ \\
 & & $4$ & $(x - 5)\,x^{8}\,(x^{2} - 9 x + \frac{35}{2})\,(x^{3} - 16 x^{2} + \frac{155}{2} \, x - 115)\,(x^{6} - \frac{91}{2} \, x^{5} + \frac{1617}{2} \, x^{4} - \frac{28827}{4} \, x^{3} + \frac{67983}{2} \, x^{2} - \frac{320935}{4} \, x + 73625)$ \\
 & & $3$ & $(x - 10)\,(x - \frac{20}{3})\,x^{4}\,(x^{2} - 22 x + \frac{752}{9})\,(x^{2} - \frac{38}{3} \, x + \frac{100}{3})$ \\
 & & $2$ & $(x - 20)\,(x - 10)\,x$ \\[1mm]
\hline
\end{tabular}
\caption{$U(N)$ characteristic polynomials for $m+n=8$.}\label{tbl:UN-combined-charpoly-L8}
\end{sidewaystable}

\clearpage
\subsection{$SU(N)$ case}

\begin{table}[!htbp]
\centering
\small
\setlength{\tabcolsep}{4pt}
\setlength{\extrarowheight}{1pt}
\renewcommand{\arraystretch}{1.13}
\begin{tabular}{ccc|l}
$m$ & $n$ & $N$ & $\det_{\cH'_N{}^{\mathrm{HWS}}(m,n)} (x-  \fD^{\prime(1)}_N)$ \\[1mm]
\hline
$2$ & $2$ & $N\ge 4$ & $(x - 6)\,x$ \\
 & & $3$ & $(x - 6)\,x$ \\
 & & $2$ & $(x - 6)$ \\[1mm]
\hline
$3$ & $2$ & $N\ge 5$ & $(x - 4)\,x$ \\
 & & $4$ & $(x - 4)\,x$ \\
 & & $3$ & $(x - 4)$ \\
 & & $2$ & $1$ \\[1mm]
\hline
$4$ & $2$ & $N\ge 6$ & $x^{3}\,(x^{3} - 16 x^{2} + (-40 \gss^{2} + 80) x + 80 \gss^{2} - 120)$ \\
 & & $5$ & $x^{3}\,(x^{3} - 16 x^{2} + \frac{392}{5} \, x - \frac{584}{5})$ \\
 & & $4$ & $x^{2}\,(x^{3} - 16 x^{2} + \frac{155}{2} \, x - 115)$ \\
 & & $3$ & $x\,(x^{2} - \frac{38}{3} \, x + \frac{100}{3})$ \\
 & & $2$ & $(x - 10)$ \\[1mm]
\hline
$3$ & $3$ & $N\ge 6$ & $(x - 6)$ \\
 & & $5$ & $(x - 6)$ \\
 & & $4$ & $(x - 6)$ \\
 & & $3$ & $(x - 6)$ \\
 & & $2$ & $1$ \\[1mm]
\hline
$5$ & $2$ & $N\ge 7$ & $x^{4}\,(x^{4} - 18 x^{3} + (-40 \gss^{2} + 116) x^{2} + (168 \gss^{2} - 312) x + 288)$ \\
 & & $6$ & $x^{4}\,(x^{4} - 18 x^{3} + \frac{1034}{9} \, x^{2} - \frac{922}{3} \, x + 288)$ \\
 & & $5$ & $x^{3}\,(x^{4} - 18 x^{3} + \frac{572}{5} \, x^{2} - \frac{7632}{25} \, x + 288)$ \\
 & & $4$ & $x^{2}\,(x^{3} - 15 x^{2} + \frac{137}{2} \, x - 96)$ \\
 & & $3$ & $x\,(x^{2} - \frac{38}{3} \, x + 36)$ \\
 & & $2$ & $1$ \\[1mm]
\hline
$4$ & $3$ & $N\ge 7$ & $(x - 5)\,x^{2}\,(x^{2} - 9 x - 40 \gss^{2} + 20)$ \\
 & & $6$ & $(x - \frac{17}{3})\,(x - 5)\,(x - \frac{10}{3})\,x^{2}$ \\
 & & $5$ & $(x - 5)\,x^{2}\,(x^{2} - 9 x + \frac{92}{5})$ \\
 & & $4$ & $(x - 5)\,x\,(x^{2} - 9 x + \frac{35}{2})$ \\
 & & $3$ & $(x - \frac{20}{3})$ \\
 & & $2$ & $1$ \\[1mm]
\hline
\end{tabular}
\caption{$SU(N)$ characteristic polynomials for $4 \le m+n \le 7$.}\label{tbl:SUN-combined-charpoly}
\end{table}

\begin{sidewaystable}[p]
\centering
\small
\setlength{\tabcolsep}{4pt}
\setlength{\extrarowheight}{1pt}
\renewcommand{\arraystretch}{1.13}
\begin{tabular}{ccc|p{0.80\textheight}}
$m$ & $n$ & $N$ & $\det_{\cH'_N{}^{\mathrm{HWS}}(m,n)} (x-  \fD^{\prime(1)}_N)$ \\[1mm]
\hline
$6$ & $2$ & $N\ge 8$ & $x^{8}\,(x^{8} - 40 x^{7} + (-232 \gss^{2} + 684) x^{6} + (5976 \gss^{2} - 6512) x^{5} + (8576 \gss^{4} - 60976 \gss^{2} + 37616) x^{4} + (-100416 \gss^{4} + 310080 \gss^{2} - 134368) x^{3} + (-35840 \gss^{6} + 280960 \gss^{4} - 802880 \gss^{2} + 288064) x^{2} + (-164864 \gss^{6} + 288512 \gss^{4} + 934528 \gss^{2} - 336000) x - 186368 \gss^{6} - 1229312 \gss^{4} - 301056 \gss^{2} + 161280)$ \\
 & & $7$ & $x^{8}\,(x^{8} - 40 x^{7} + \frac{33284}{49} \, x^{6} - \frac{313112}{49} \, x^{5} + \frac{87336768}{2401} \, x^{4} - \frac{307524064}{2401} \, x^{3} + \frac{4568065408}{16807} \, x^{2} - \frac{5324612864}{16807} \, x + \frac{2598738944}{16807})$ \\
 & & $6$ & $x^{7}\,(x^{8} - 40 x^{7} + \frac{6098}{9} \, x^{6} - 6346 x^{5} + \frac{2910236}{81} \, x^{4} - \frac{3397468}{27} \, x^{3} + \frac{193897816}{729} \, x^{2} - \frac{225860096}{729} \, x + \frac{110782336}{729})$ \\
 & & $5$ & $x^{6}\,(x^{7} - \frac{186}{5} \, x^{6} + \frac{14264}{25} \, x^{5} - \frac{584424}{125} \, x^{4} + \frac{2762448}{125} \, x^{3} - \frac{37654144}{625} \, x^{2} + \frac{274079616}{3125} \, x - \frac{164351488}{3125})$ \\
 & & $4$ & $x^{4}\,(x^{6} - 35 x^{5} + \frac{975}{2} \, x^{4} - 3456 x^{3} + 13146 x^{2} - \frac{101833}{4} \, x + \frac{39319}{2})$ \\
 & & $3$ & $x\,(x^{3} - 22 x^{2} + \frac{1304}{9} \, x - \frac{2576}{9})$ \\
 & & $2$ & $(x - 14)$ \\[1mm]
\hline
$5$ & $3$ & $N\ge 8$ & $(x - 4)\,(x - 6 \gss - 6)\,(x + 6 \gss - 6)\,x^{3}\,(x^{4} - 18 x^{3} + (-40 \gss^{2} + 116) x^{2} + (432 \gss^{2} - 320) x - 960 \gss^{2} + 320)$ \\
 & & $7$ & $(x - \frac{48}{7})\,(x - \frac{36}{7})\,(x - 4)\,x^{3}\,(x^{4} - 18 x^{3} + \frac{5644}{49} \, x^{2} - \frac{15248}{49} \, x + \frac{14720}{49})$ \\
 & & $6$ & $(x - 7)\,(x - 5)\,(x - 4)\,(x - \frac{10}{3})\,x^{3}\,(x^{3} - \frac{44}{3} \, x^{2} + 66 x - 88)$ \\
 & & $5$ & $(x - \frac{36}{5})\,(x - \frac{24}{5})\,(x - 4)\,x^{2}\,(x^{4} - 18 x^{3} + \frac{572}{5} \, x^{2} - \frac{7568}{25} \, x + \frac{1408}{5})$ \\
 & & $4$ & $(x - \frac{15}{2})\,(x - 4)\,x\,(x^{3} - 16 x^{2} + \frac{163}{2} \, x - 130)$ \\
 & & $3$ & $(x - 8)\,(x - \frac{20}{3})$ \\
 & & $2$ & $1$ \\[1mm]
\hline
$4$ & $4$ & $N\ge 8$ & $x^{4}\,(x^{7} - 48 x^{6} + (-412 \gss^{2} + 948) x^{5} + (11840 \gss^{2} - 9968) x^{4} + (30560 \gss^{4} - 132848 \gss^{2} + 60192) x^{3} + (-476800 \gss^{4} + 722080 \gss^{2} - 208480) x^{2} + (-627200 \gss^{6} + 2770560 \gss^{4} - 1912000 \gss^{2} + 383040) x + 2816000 \gss^{6} - 5824000 \gss^{4} + 2016000 \gss^{2} - 288000)$ \\
 & & $7$ & $x^{4}\,(x^{7} - 48 x^{6} + \frac{46040}{49} \, x^{5} - \frac{476592}{49} \, x^{4} + \frac{138042000}{2401} \, x^{3} - \frac{465655360}{2401} \, x^{2} + \frac{828748800}{2401} \, x - \frac{29325056000}{117649})$ \\
 & & $6$ & $x^{4}\,(x^{7} - 48 x^{6} + \frac{8429}{9} \, x^{5} - \frac{86752}{9} \, x^{4} + \frac{4578554}{81} \, x^{3} - \frac{15292000}{81} \, x^{2} + \frac{242066800}{729} \, x - \frac{172360000}{729})$ \\
 & & $5$ & $x^{3}\,(x^{7} - 48 x^{6} + \frac{23288}{25} \, x^{5} - \frac{47472}{5} \, x^{4} + \frac{6865872}{125} \, x^{3} - \frac{4508992}{25} \, x^{2} + \frac{194345472}{625} \, x - \frac{27062272}{125})$ \\
 & & $4$ & $x^{2}\,(x^{6} - \frac{91}{2} \, x^{5} + \frac{1617}{2} \, x^{4} - \frac{28827}{4} \, x^{3} + \frac{67983}{2} \, x^{2} - \frac{320935}{4} \, x + 73625)$ \\
 & & $3$ & $(x - 10)\,x\,(x^{2} - 22 x + \frac{752}{9})$ \\
 & & $2$ & $(x - 20)$ \\[1mm]
\hline
\end{tabular}
\caption{$SU(N)$ characteristic polynomials for $m+n=8$.}\label{tbl:SUN-combined-charpoly-L8}
\end{sidewaystable}

\clearpage
\subsection{Counting the quarter-BPS states}

The grand generating function for quarter-BPS states, including half-BPS states,
in the $\su2$ sector of $\cN=4$ SYM with gauge group $U(N)$ is
\cite{Kinney:2005ej,Dolan:2007rq}
\begin{equation}
Z_{\ge 1/4}^{U}(\nu,x,y)
=
\prod_{k,j\geq0}\frac{1}{1-\nu x^k y^j}
=
1+\sum_{N=1}^{\infty}\sum_{m,n=0}^{\infty}
d^{U(N)}_{m,n}\,\nu^N x^m y^n.
\end{equation}
The corresponding $SU(N)$ generating function is
\begin{equation}
Z_{\ge 1/4}^{SU}(\nu,x,y)
=
1+(1-x)(1-y)\bigl(Z_{\ge 1/4}^{U}(\nu,x,y)-1\bigr)
=
1+\sum_{N=1}^{\infty}\sum_{m,n=0}^{\infty}
d^{SU(N)}_{m,n}\,\nu^N x^m y^n.
\end{equation}
Here $\nu$ counts the gauge-group rank, and $x,y$ count the two scalar degrees.
We assign the constant term $1$ at $N=0$ for convenience.
The HWS multiplicities are obtained by expanding in $U(2)$ characters:
\begin{equation}
\begin{aligned}
Z_{\ge 1/4}^{U}(\nu,x,y)
&= 1+
\sum_{N=1}^{\infty}\nu^N
\sum_{L=0}^{\infty}
\sum_{\substack{\Lambda\vdash L\\\ell(\Lambda)\leq2}}
d^{U(N),\mathrm{HWS}}_{\Lambda}\,\chi^\Lambda(x,y),\\
Z_{\ge 1/4}^{SU}(\nu,x,y)
&= 1+
\sum_{N=1}^{\infty}\nu^N
\sum_{L=0}^{\infty}
\sum_{\substack{\Lambda\vdash L\\\ell(\Lambda)\leq2}}
d^{SU(N),\mathrm{HWS}}_{\Lambda}\,\chi^\Lambda(x,y).
\end{aligned}
\end{equation}
where
\begin{equation}
\begin{aligned}
d^{U(N),{\rm HWS}}_{[m,n]}
 &=d^{U(N)}_{m,n}-d^{U(N)}_{m+1,n-1},
\\[2mm]
d^{SU(N),{\rm HWS}}_{[m,n]}
 &=d^{U(N),{\rm HWS}}_{[m,n]}
 -d^{U(N),{\rm HWS}}_{[m,n-1]}
 -(1-\delta_{m,n})d^{U(N),{\rm HWS}}_{[m-1,n]}
 +d^{U(N),{\rm HWS}}_{[m-1,n-1]}.
\end{aligned}
\end{equation}
The first equation follows from the $u(2)$ character \eqref{def:U2 char}, and the second equation is Lemma \ref{prop:HWS dim comp}.

Below we summarise the numbers of zero modes, $\dim \ker \fD_N^{(1)}$ for $U(N)$ and $\dim \ker \fD_N^{\prime (1)}$ for $SU(N)$ in the HWS sectors with $1\leq m+n\leq8$.
We denote these zero-mode multiplicities by $N^{U(N),{\rm HWS}}_{\ge1/4,[m,n]}$ and $N^{SU(N),{\rm HWS}}_{\ge1/4,[m,n]}$.
These numbers are equal to the powers of $x$ in the corresponding characteristic polynomials. 
Every entry in the tables for the HWS sector agrees with the coefficients of the generating functions for any $N$,
\begin{equation}
N^{U(N),{\rm HWS}}_{\ge1/4,[m,n]} =d^{U(N),{\rm HWS}}_{[m,n]},\qquad
N^{SU(N),{\rm HWS}}_{\ge1/4,[m,n]} =d^{SU(N),{\rm HWS}}_{[m,n]}, \qquad \forall N \in \mathbb{N}.
\end{equation}
This agreement provides non-trivial evidence that the space of $SU(N)$-invariants at finite $N$ is described by the intersection.

The half-BPS contribution is obtained by setting $y=0$ in the generating functions.
A half-BPS HWS belongs to the $u(2)$ representation $[L,0]$. 
A genuinely quarter-BPS HWS belongs to a representation $[m,n]$ with $m\geq n\geq1$, for both $U(N)$ and $SU(N)$ gauge groups.

In the following tables, blank entries at $N>m+n$ indicate the stable value displayed in the $N=\infty$ column.

\subsubsection{$U(N)$ case}

The $N=1$ entry corresponds to $\tr(Z)^m \, \tr(W)^n$. They are half-BPS descendants, and become HWS only if $n=0$.

\begin{table}[h]
\centering
\small
\setlength{\tabcolsep}{4pt}
\setlength{\extrarowheight}{1pt}
\renewcommand{\arraystretch}{1.13}
\begin{tabular}{cc|ccccccccc}
$m$ & $n$ & $N=1$ & $2$ & $3$ & $4$ & $5$ & $6$ & $7$ & $8$ & $\infty$ \\[1mm]
\hline
$1$ & $0$ & $1$ &  &  &  &  &  &  &  & $1$ \\[1mm]
\hline
$1$ & $1$ & $0$ & $0$ &  &  &  &  &  &  & $0$ \\
$2$ & $0$ & $1$ & $2$ &  &  &  &  &  &  & $2$ \\[1mm]
\hline
$2$ & $1$ & $0$ & $1$ & $1$ &  &  &  &  &  & $1$ \\
$3$ & $0$ & $1$ & $2$ & $3$ &  &  &  &  &  & $3$ \\[1mm]
\hline
$2$ & $2$ & $0$ & $1$ & $2$ & $2$ &  &  &  &  & $2$ \\
$3$ & $1$ & $0$ & $1$ & $2$ & $2$ &  &  &  &  & $2$ \\
$4$ & $0$ & $1$ & $3$ & $4$ & $5$ &  &  &  &  & $5$ \\[1mm]
\hline
$3$ & $2$ & $0$ & $1$ & $3$ & $4$ & $4$ &  &  &  & $4$ \\
$4$ & $1$ & $0$ & $2$ & $4$ & $5$ & $5$ &  &  &  & $5$ \\
$5$ & $0$ & $1$ & $3$ & $5$ & $6$ & $7$ &  &  &  & $7$ \\[1mm]
\hline
$3$ & $3$ & $0$ & $0$ & $1$ & $2$ & $2$ & $2$ &  &  & $2$ \\
$4$ & $2$ & $0$ & $2$ & $6$ & $9$ & $10$ & $10$ &  &  & $10$ \\
$5$ & $1$ & $0$ & $2$ & $5$ & $7$ & $8$ & $8$ &  &  & $8$ \\
$6$ & $0$ & $1$ & $4$ & $7$ & $9$ & $10$ & $11$ &  &  & $11$ \\[1mm]
\hline
$4$ & $3$ & $0$ & $1$ & $4$ & $8$ & $10$ & $10$ & $10$ &  & $10$ \\
$5$ & $2$ & $0$ & $2$ & $8$ & $13$ & $16$ & $17$ & $17$ &  & $17$ \\
$6$ & $1$ & $0$ & $3$ & $8$ & $12$ & $14$ & $15$ & $15$ &  & $15$ \\
$7$ & $0$ & $1$ & $4$ & $8$ & $11$ & $13$ & $14$ & $15$ &  & $15$ \\[1mm]
\hline
$4$ & $4$ & $0$ & $1$ & $4$ & $8$ & $11$ & $12$ & $12$ & $12$ & $12$ \\
$5$ & $3$ & $0$ & $1$ & $6$ & $13$ & $18$ & $20$ & $20$ & $20$ & $20$ \\
$6$ & $2$ & $0$ & $3$ & $12$ & $22$ & $28$ & $31$ & $32$ & $32$ & $32$ \\
$7$ & $1$ & $0$ & $3$ & $10$ & $16$ & $20$ & $22$ & $23$ & $23$ & $23$ \\
$8$ & $0$ & $1$ & $5$ & $10$ & $15$ & $18$ & $20$ & $21$ & $22$ & $22$ \\[1mm]
\hline
\end{tabular}
\caption{The $U(N)$ HWS zero-mode multiplicities $N^{U(N),{\rm HWS}}_{\ge 1/4,[m,n]}$ for $1\le m+n\le8$.}
\label{tbl:UN-zero-modes}
\end{table}

\begin{table}[h]
\centering
\small
\setlength{\tabcolsep}{4pt}
\setlength{\extrarowheight}{1pt}
\renewcommand{\arraystretch}{1.13}
\begin{tabular}{cc|ccccccccc}
$m$ & $n$ & $N=1$ & $2$ & $3$ & $4$ & $5$ & $6$ & $7$ & $8$ & $\infty$ \\[1mm]
\hline
$1$ & $0$ & $1$ &  &  &  &  &  &  &  & $1$ \\[1mm]
\hline
$1$ & $1$ & $1$ & $2$ &  &  &  &  &  &  & $2$ \\
$2$ & $0$ & $1$ & $2$ &  &  &  &  &  &  & $2$ \\[1mm]
\hline
$2$ & $1$ & $1$ & $3$ & $4$ &  &  &  &  &  & $4$ \\
$3$ & $0$ & $1$ & $2$ & $3$ &  &  &  &  &  & $3$ \\[1mm]
\hline
$2$ & $2$ & $1$ & $5$ & $8$ & $9$ &  &  &  &  & $9$ \\
$3$ & $1$ & $1$ & $4$ & $6$ & $7$ &  &  &  &  & $7$ \\
$4$ & $0$ & $1$ & $3$ & $4$ & $5$ &  &  &  &  & $5$ \\[1mm]
\hline
$3$ & $2$ & $1$ & $6$ & $12$ & $15$ & $16$ &  &  &  & $16$ \\
$4$ & $1$ & $1$ & $5$ & $9$ & $11$ & $12$ &  &  &  & $12$ \\
$5$ & $0$ & $1$ & $3$ & $5$ & $6$ & $7$ &  &  &  & $7$ \\[1mm]
\hline
$3$ & $3$ & $1$ & $8$ & $19$ & $27$ & $30$ & $31$ &  &  & $31$ \\
$4$ & $2$ & $1$ & $8$ & $18$ & $25$ & $28$ & $29$ &  &  & $29$ \\
$5$ & $1$ & $1$ & $6$ & $12$ & $16$ & $18$ & $19$ &  &  & $19$ \\
$6$ & $0$ & $1$ & $4$ & $7$ & $9$ & $10$ & $11$ &  &  & $11$ \\[1mm]
\hline
$4$ & $3$ & $1$ & $10$ & $28$ & $44$ & $53$ & $56$ & $57$ &  & $57$ \\
$5$ & $2$ & $1$ & $9$ & $24$ & $36$ & $43$ & $46$ & $47$ &  & $47$ \\
$6$ & $1$ & $1$ & $7$ & $16$ & $23$ & $27$ & $29$ & $30$ &  & $30$ \\
$7$ & $0$ & $1$ & $4$ & $8$ & $11$ & $13$ & $14$ & $15$ &  & $15$ \\[1mm]
\hline
$4$ & $4$ & $1$ & $13$ & $42$ & $74$ & $95$ & $105$ & $108$ & $109$ & $109$ \\
$5$ & $3$ & $1$ & $12$ & $38$ & $66$ & $84$ & $93$ & $96$ & $97$ & $97$ \\
$6$ & $2$ & $1$ & $11$ & $32$ & $53$ & $66$ & $73$ & $76$ & $77$ & $77$ \\
$7$ & $1$ & $1$ & $8$ & $20$ & $31$ & $38$ & $42$ & $44$ & $45$ & $45$ \\
$8$ & $0$ & $1$ & $5$ & $10$ & $15$ & $18$ & $20$ & $21$ & $22$ & $22$ \\[1mm]
\hline
\end{tabular}
\caption{The $U(N)$ raw bidegree multiplicities $d^{U(N)}_{m,n}$ for $1\le m+n\le8$.}
\label{tbl:UN-raw-zero-modes}
\end{table}

\clearpage
\subsubsection{$SU(N)$ case}

\begin{table}[h]
\centering
\small
\setlength{\tabcolsep}{4pt}
\setlength{\extrarowheight}{1pt}
\renewcommand{\arraystretch}{1.13}
\begin{tabular}{cc|ccccccccc}
$m$ & $n$ & $N=1$ & $2$ & $3$ & $4$ & $5$ & $6$ & $7$ & $8$ & $\infty$ \\[1mm]
\hline
$1$ & $0$ & $0$ &  &  &  &  &  &  &  & $0$ \\[1mm]
\hline
$1$ & $1$ & $0$ & $0$ &  &  &  &  &  &  & $0$ \\
$2$ & $0$ & $0$ & $1$ &  &  &  &  &  &  & $1$ \\[1mm]
\hline
$2$ & $1$ & $0$ & $0$ & $0$ &  &  &  &  &  & $0$ \\
$3$ & $0$ & $0$ & $0$ & $1$ &  &  &  &  &  & $1$ \\[1mm]
\hline
$2$ & $2$ & $0$ & $0$ & $1$ & $1$ &  &  &  &  & $1$ \\
$3$ & $1$ & $0$ & $0$ & $0$ & $0$ &  &  &  &  & $0$ \\
$4$ & $0$ & $0$ & $1$ & $1$ & $2$ &  &  &  &  & $2$ \\[1mm]
\hline
$3$ & $2$ & $0$ & $0$ & $0$ & $1$ & $1$ &  &  &  & $1$ \\
$4$ & $1$ & $0$ & $0$ & $1$ & $1$ & $1$ &  &  &  & $1$ \\
$5$ & $0$ & $0$ & $0$ & $1$ & $1$ & $2$ &  &  &  & $2$ \\[1mm]
\hline
$3$ & $3$ & $0$ & $0$ & $0$ & $0$ & $0$ & $0$ &  &  & $0$ \\
$4$ & $2$ & $0$ & $0$ & $1$ & $2$ & $3$ & $3$ &  &  & $3$ \\
$5$ & $1$ & $0$ & $0$ & $0$ & $1$ & $1$ & $1$ &  &  & $1$ \\
$6$ & $0$ & $0$ & $1$ & $2$ & $3$ & $3$ & $4$ &  &  & $4$ \\[1mm]
\hline
$4$ & $3$ & $0$ & $0$ & $0$ & $1$ & $2$ & $2$ & $2$ &  & $2$ \\
$5$ & $2$ & $0$ & $0$ & $1$ & $2$ & $3$ & $4$ & $4$ &  & $4$ \\
$6$ & $1$ & $0$ & $0$ & $1$ & $2$ & $3$ & $3$ & $3$ &  & $3$ \\
$7$ & $0$ & $0$ & $0$ & $1$ & $2$ & $3$ & $3$ & $4$ &  & $4$ \\[1mm]
\hline
$4$ & $4$ & $0$ & $0$ & $1$ & $2$ & $3$ & $4$ & $4$ & $4$ & $4$ \\
$5$ & $3$ & $0$ & $0$ & $0$ & $1$ & $2$ & $3$ & $3$ & $3$ & $3$ \\
$6$ & $2$ & $0$ & $0$ & $1$ & $4$ & $6$ & $7$ & $8$ & $8$ & $8$ \\
$7$ & $1$ & $0$ & $0$ & $1$ & $2$ & $3$ & $4$ & $4$ & $4$ & $4$ \\
$8$ & $0$ & $0$ & $1$ & $2$ & $4$ & $5$ & $6$ & $6$ & $7$ & $7$ \\[1mm]
\hline
\end{tabular}
\caption{The $SU(N)$ HWS zero-mode multiplicities $N^{SU(N),{\rm HWS}}_{\ge 1/4,[m,n]}$ for $1\le m+n\le8$.}
\label{tbl:SUN-zero-modes}
\end{table}

\begin{table}[h]
\centering
\small
\setlength{\tabcolsep}{4pt}
\setlength{\extrarowheight}{1pt}
\renewcommand{\arraystretch}{1.13}
\begin{tabular}{cc|ccccccccc}
$m$ & $n$ & $N=1$ & $2$ & $3$ & $4$ & $5$ & $6$ & $7$ & $8$ & $\infty$ \\[1mm]
\hline
$1$ & $0$ & $0$ &  &  &  &  &  &  &  & $0$ \\[1mm]
\hline
$1$ & $1$ & $0$ & $1$ &  &  &  &  &  &  & $1$ \\
$2$ & $0$ & $0$ & $1$ &  &  &  &  &  &  & $1$ \\[1mm]
\hline
$2$ & $1$ & $0$ & $0$ & $1$ &  &  &  &  &  & $1$ \\
$3$ & $0$ & $0$ & $0$ & $1$ &  &  &  &  &  & $1$ \\[1mm]
\hline
$2$ & $2$ & $0$ & $1$ & $2$ & $3$ &  &  &  &  & $3$ \\
$3$ & $1$ & $0$ & $1$ & $1$ & $2$ &  &  &  &  & $2$ \\
$4$ & $0$ & $0$ & $1$ & $1$ & $2$ &  &  &  &  & $2$ \\[1mm]
\hline
$3$ & $2$ & $0$ & $0$ & $2$ & $3$ & $4$ &  &  &  & $4$ \\
$4$ & $1$ & $0$ & $0$ & $2$ & $2$ & $3$ &  &  &  & $3$ \\
$5$ & $0$ & $0$ & $0$ & $1$ & $1$ & $2$ &  &  &  & $2$ \\[1mm]
\hline
$3$ & $3$ & $0$ & $1$ & $3$ & $6$ & $7$ & $8$ &  &  & $8$ \\
$4$ & $2$ & $0$ & $1$ & $3$ & $6$ & $7$ & $8$ &  &  & $8$ \\
$5$ & $1$ & $0$ & $1$ & $2$ & $4$ & $4$ & $5$ &  &  & $5$ \\
$6$ & $0$ & $0$ & $1$ & $2$ & $3$ & $3$ & $4$ &  &  & $4$ \\[1mm]
\hline
$4$ & $3$ & $0$ & $0$ & $3$ & $7$ & $11$ & $12$ & $13$ &  & $13$ \\
$5$ & $2$ & $0$ & $0$ & $3$ & $6$ & $9$ & $10$ & $11$ &  & $11$ \\
$6$ & $1$ & $0$ & $0$ & $2$ & $4$ & $6$ & $6$ & $7$ &  & $7$ \\
$7$ & $0$ & $0$ & $0$ & $1$ & $2$ & $3$ & $3$ & $4$ &  & $4$ \\[1mm]
\hline
$4$ & $4$ & $0$ & $1$ & $5$ & $13$ & $19$ & $24$ & $25$ & $26$ & $26$ \\
$5$ & $3$ & $0$ & $1$ & $4$ & $11$ & $16$ & $20$ & $21$ & $22$ & $22$ \\
$6$ & $2$ & $0$ & $1$ & $4$ & $10$ & $14$ & $17$ & $18$ & $19$ & $19$ \\
$7$ & $1$ & $0$ & $1$ & $3$ & $6$ & $8$ & $10$ & $10$ & $11$ & $11$ \\
$8$ & $0$ & $0$ & $1$ & $2$ & $4$ & $5$ & $6$ & $6$ & $7$ & $7$ \\[1mm]
\hline
\end{tabular}
\caption{The $SU(N)$ raw bidegree multiplicities $d^{SU(N)}_{m,n}$ for $1\le m+n\le8$.}
\label{tbl:SUN-raw-zero-modes}
\end{table}

\bibliographystyle{utphys}
\bibliography{references}

\end{document}